\documentclass[12pt, draftclsnofoot, onecolumn]{IEEEtran}
\IEEEoverridecommandlockouts
\usepackage{cite}
\usepackage{amsmath,amssymb,amsfonts}
\usepackage{graphicx}
\usepackage{textcomp}
\usepackage{xcolor}
\usepackage{multirow}
\def\BibTeX{{\rm B\kern-.05em{\sc i\kern-.025em b}\kern-.08em
    T\kern-.1667em\lower.7ex\hbox{E}\kern-.125emX}}

\usepackage{mathrsfs}
\usepackage{amsthm}
\usepackage{diagbox}
\usepackage{verbatim}
\usepackage{dsfont}
\usepackage{bm}
\usepackage{subeqnarray}
\usepackage{cases}
\usepackage{subcaption}
\usepackage{afterpage}

\usepackage{caption}

\newtheorem{theorem}{Theorem}
\newtheorem{lemma}{Lemma}
\newtheorem{corollary}{Corollary}
\newtheorem{assumption}{Assumption}
\newtheorem{conjecture}{Conjecture}
\newtheorem{proposition}{Proposition}

\theoremstyle{definition}
\newtheorem{remark}{Remark}
\newtheorem{definition}{Definition}
\newtheorem{example}{Example}

\usepackage{algorithm}
\usepackage{algorithmic}
\usepackage{bbm}
\usepackage{hyperref}
\hypersetup{colorlinks = true
            linkcolor=red,
            anchorcolor=blue,
            citecolor=green}
              
\newcommand{\MF}{\mathcal{F}}
\newcommand{\MI}{\mathcal{I}}

\newcommand{\FF}{\mathbb{F}}
\newcommand{\RR}{\mathbb{R}}

\begin{document}

\title{Stitched Polar Codes}

 \author{
   \IEEEauthorblockN{Yuan Li\IEEEauthorrefmark{1}, Zicheng Ye\IEEEauthorrefmark{1}, Huazi Zhang\IEEEauthorrefmark{1}, Jun Wang\IEEEauthorrefmark{1}, Wen Tong\IEEEauthorrefmark{1}, Guiying Yan\IEEEauthorrefmark{2}\IEEEauthorrefmark{3}
   and Zhiming Ma\IEEEauthorrefmark{2}\IEEEauthorrefmark{3} }

   \IEEEauthorblockA{\IEEEauthorrefmark{1}
                     Huawei Technologies Co. Ltd.}\\
  \IEEEauthorblockA{\IEEEauthorrefmark{2}
                     School of Mathematical Sciences, University of Chinese Academy of Sciences}\\
   \IEEEauthorblockA{\IEEEauthorrefmark{3}
                     Academy of Mathematics and Systems Science, CAS } \\
    Email: \{liyuan299, yezicheng3, zhanghuazi, justin.wangjun, tongwen\}@huawei.com, \\ 
     yangy@amss.ac.cn, mazm@amt.ac.cn 
     \thanks{This work was supported by the National Key R\&D Program of China (2023YFA1009601).}}
    
\maketitle
                  
\begin{abstract}
In this paper, we introduce stitched polar codes, a novel generalization of Ar{\i}kan's regular polar codes. Our core methodology reconfigures the fundamental polarization process by stitching additional structures to enhance the reliability of less reliable information bits in the original code. This approach preserves the polar transformation structure and maintains the same encoding and decoding complexity. Thanks to the flexible configuration, stitched polar codes consistently outperform regular polar codes, effectively solving the performance degradation issue in rate-matched scenarios. Furthermore, we provide theoretical analysis on the weight spectrum and the polarization speed of stitched polar codes to prove their superiority.
\end{abstract}

\section{Introduction}

Polar codes \cite{Arikan2009}, introduced by Ar{\i}kan, represent a significant breakthrough in coding theory. When the code length $N$ approaches infinity, polar codes achieve channel capacity under successive cancellation (SC) decoding with a time complexity of $O(N\log N)$. For short to moderate code lengths, successive cancellation list (SCL) decoding \cite{Niu2012, Tal2015} can significantly enhance the error-correcting performance. 

The structure of polar codes, based on the Kronecker product, limits the original code length to powers of two. However, practical communication systems require variable code lengths for data transmission. To overcome this limitation and enable flexible code length adaptation, rate-matching techniques such as  puncturing and shortening have been developed for polar codes. 

Multiple rate-matching patterns have been proposed to enhance performance, including the quasi-uniform puncturing (QUP) \cite{Niu2013}, which ensures puncturing positions follow a quasi-uniform distribution after bit-reversal permutation. Recent research demonstrates that QUP polar codes achieve channel capacity \cite{Shuval2024}. Meanwhile, the bit-reversal shortening (BRS) approach introduced in \cite{Bioglio2017} offers a practical implementation path for shortened polar codes. The extensive studies \cite{Wang2014Novel, Miloslavskaya2015, Niu2016, Chandesris2017, Oliveira2018, Han2022, Li2023Two} have provided valuable insights and advancements in the design of rate-compatible polar codes. 

However, current rate-matched polar codes still suffer from performance degradation at non-power-of-two code lengths, even when the codes are sufficiently long. Specifically, when the code length exceeds a power-of-two value, the required signal-to-noise ratio (SNR) to achieve a target block error rate (BLER) exhibits an unexpected increase. This performance degradation deviates from the general coding principle that longer codes typically offer superior error correction capabilities. 

\subsection{Main Contributions}

In this paper, we propose stitched polar codes. The main contributions are summarized as follows:
\begin{itemize}
\item \emph{Performance enhancement}: Stitched polar codes eliminate the performance degradation inherent to existing rate-adapted versions of regular polar codes. Furthermore, they enhance performance even at power-of-2 lengths through their optimized polarization process. As shown in Fig. \ref{fig_SC3}, this results in a significantly smoother performance curve under fine-granuliarity simulation under different code lengths, consistently outperforming QUP and BRS polar codes by up to 0.3dB with the same decoding complexity. 
\item \emph{Efficient encoding and decoding}: Stitched polar codes are constructed by restructuring the basic $2\times 2$ polarization transforms in original polar codes. This method reuses existing basic units, avoiding the introduction of new ones or an increase in processing units. Consequently, stitched polar codes do not incur extra decoding complexity, and preserve an inherent hardware-friendly architecture that only requires $O(N\log N)$ f/g-operations for decoding a length-$N$ code.
\item \emph{Theoretical analysis}: Our theoretical analysis of the polarization speed for stitched polar codes demonstrates that the number of un-polarized bit-channels in stitched polar codes is upper-bounded by a constant fraction, strictly less than one, of that in regular polar codes. Unlike rate-matched regular polar codes, the proportion of un-polarized bit-channels in stitched polar codes does not increase when the code length marginally exceeds a power of two. This provides an explanation for the stable performance of stitched polar codes under fine-granularity simulation.
\end{itemize}

\begin{figure}[!t]
\centering
\includegraphics[width=0.5\textwidth]{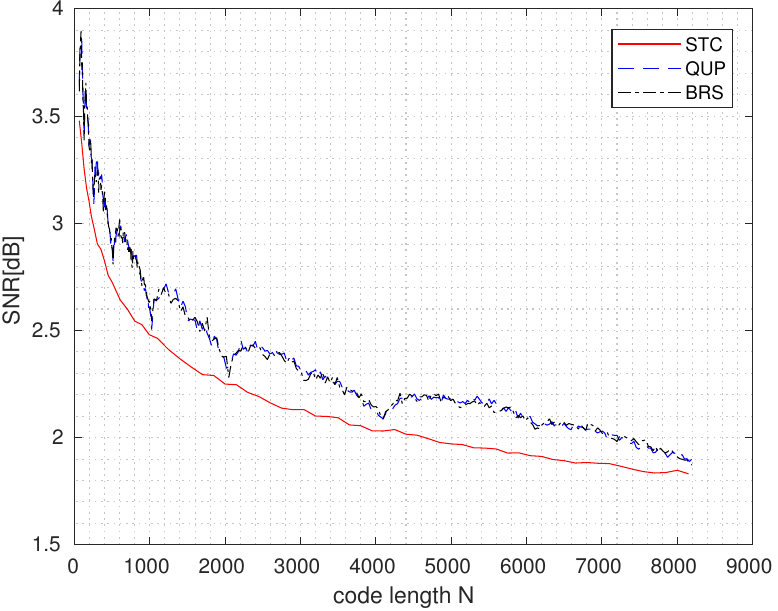}
\caption{SC decoding performance of regular and stitched (STC) polar codes with rate $R=1/2$ and BLER $=0.01$ over addictive white Gaussian noise (AWGN) channel.}
\label{fig_SC3}
\end{figure}

\subsection{Paper Organization}

The rest of this paper is organized as follows. We begin in Section II by providing simple examples to illustrate the fundamental idea behind stitched polar codes. Section III offers a comprehensive review of regular polar codes. Section IV defines stitched polar codes and describes their encoding and decoding algorithms. In Section V, two special types of stitched polar codes are defined, and their coset spectra are analysed. The specific construction methods for stitched polar codes are detailed in Section VI. In Section VII, we investigate the scaling constant for both regular and stitched polar codes including both power-of-two and non-power-of-two cases. Section VIII presents a discussion of related work and outlines several open problems. Our simulation results are presented in Section IX. Finally, Section X concludes the paper by summarizing our key technical and theoretical contributions.

\section{Examples of Stitched Polar Codes}

In this section, we provide simple illustrative examples of stitched polar codes to explain the mechanisms behind their enhanced performance.

For code length $N=5$ and dimension $K=2$, the top-performing rate-matched regular polar code is the QUP polar code, as shown in Fig. \ref{fig_N5K2_qup}. The stitched polar code is presented in Fig. \ref{fig_N5K2GG}, with the same length and dimension. The numbers in the figure represent the channel capacities, determined by density evolution (DE) \cite{Mori2009} over a binary erasure channel (BEC) with an erasure probability of 0.5. Compared to QUP polar code, the capacity of the least reliable information bit in the stitched polar codes is significantly improved (from 0.66 to 0.78), which translates to a substantially superior decoding performance.

The construction of this stitched polar code involves ``stitching" a length-1 polar code (blue) onto the length-4 base regular polar code (black) at the third bit of the base code's leftmost stage. And the stitching operation is marked in red. We stitch at this specific position because for the base code with $K=2$, the third bit is the least reliable information bit. After stitching, the third bit and the newly added bit go through an additional polarization stage. This process combines information from both bit-channels, thereby enhancing the capacity.

The motivation behind this optimization of polarization process is straightforward:  the main bottleneck in decoding is the least reliable information bits. Enhancing their reliabilities directly through stitching contributes to the greatest performance improvement.

\begin{figure}[!t]
\centering
  \begin{subfigure}[b]{0.4\textwidth}
  \centering
    \includegraphics[width=0.8\textwidth]{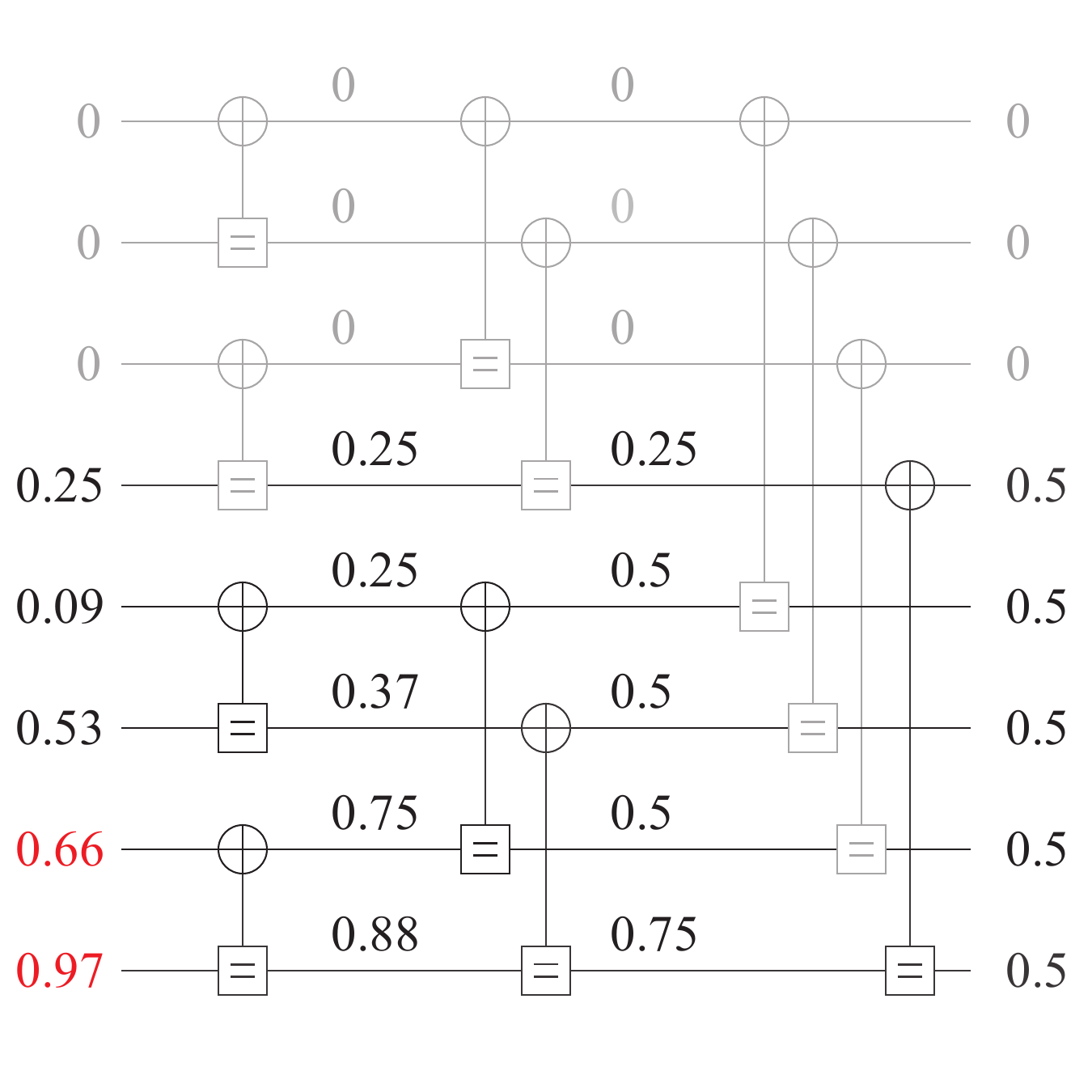}
    \caption{}
     \label{fig_N5K2_qup}
  \end{subfigure} 
  \begin{subfigure}[b]{0.4\textwidth}
  \centering
    \includegraphics[width=0.8\textwidth]{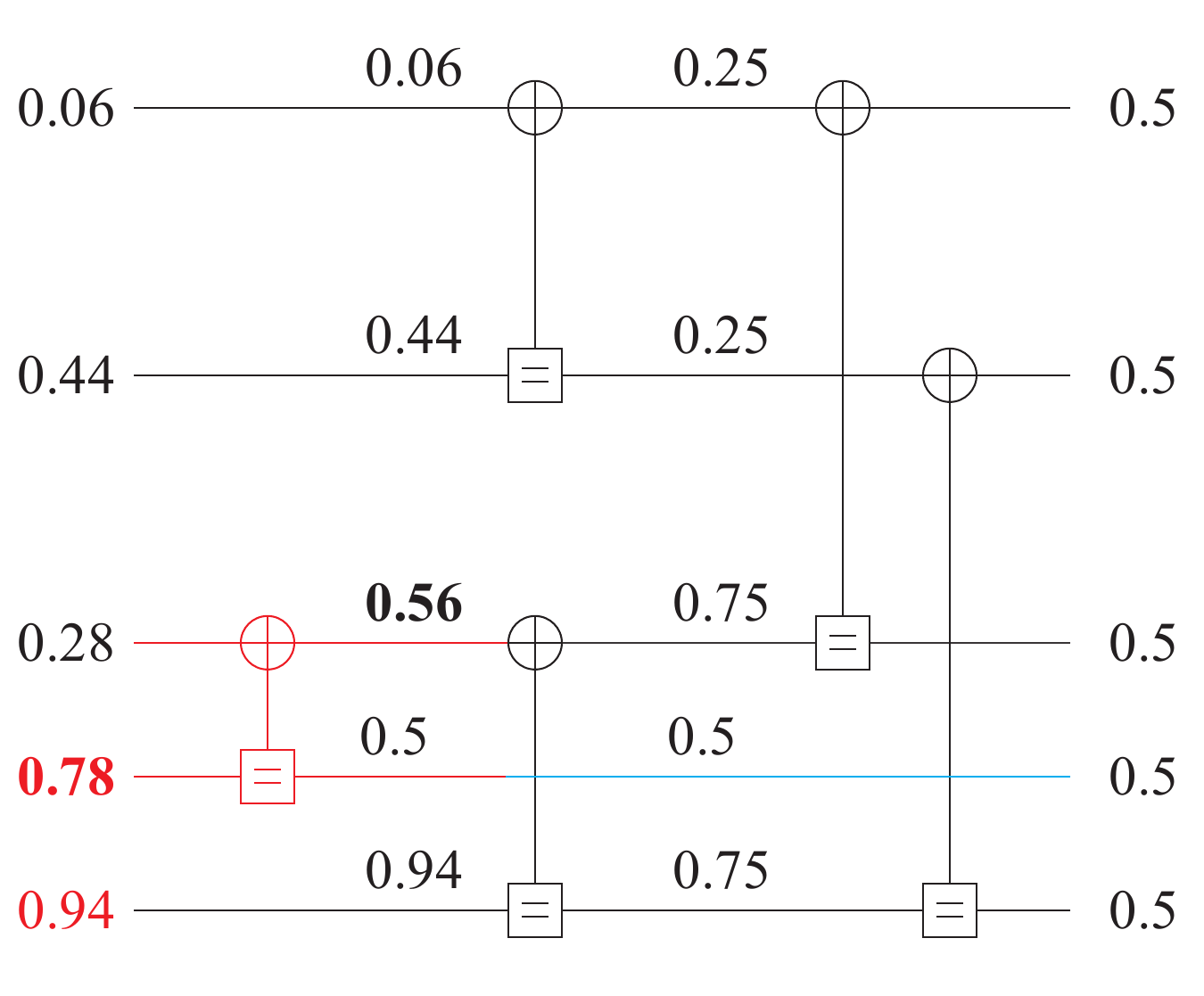}
    \caption{}
    \label{fig_N5K2GG}
  \end{subfigure} 
    \caption{The factor graph representations of (a) QUP polar code and (b) stitched polar code with $N=5$ and $K=2$. The reliabilities of information bits are marked in red.}
  \label{fig_N5K2}
\end{figure}

When the code dimension is $K=3$, we can apply a similar approach to strengthen the least reliable information bit, as demonstrated in Fig. \ref{fig_N5K3_a}. It is worth noting that the same stitched polar codes can also be obtained through different stitching methods. For example, in Fig. \ref{fig_N5K3_b}, the same stitched polar code is constructed by stitching a length-3 code (blue) and a length-2 code (black) at the rightmost position. More generally, both the stitched positions and the two constituent codes exhibit flexibility.  This inherent flexibility provides ample room for optimizing the code structure.

\begin{figure}[h]
  \centering
  \begin{subfigure}[b]{0.4\textwidth}
    \centering
    \includegraphics[width=0.8\textwidth]{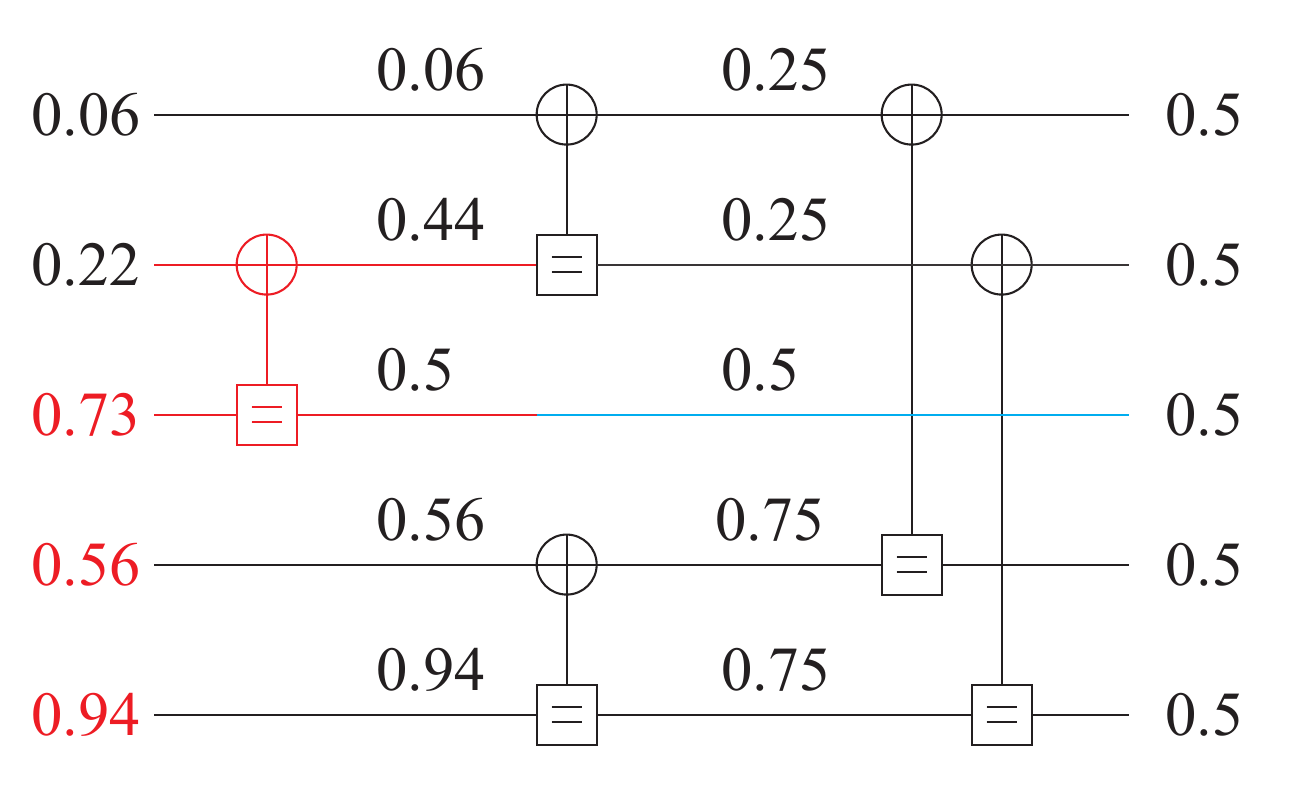}
    \caption{}
    \label{fig_N5K3_a}
  \end{subfigure}
  \hfill 
  \begin{subfigure}[b]{0.4\textwidth}
    \centering
    \includegraphics[width=0.8\textwidth]{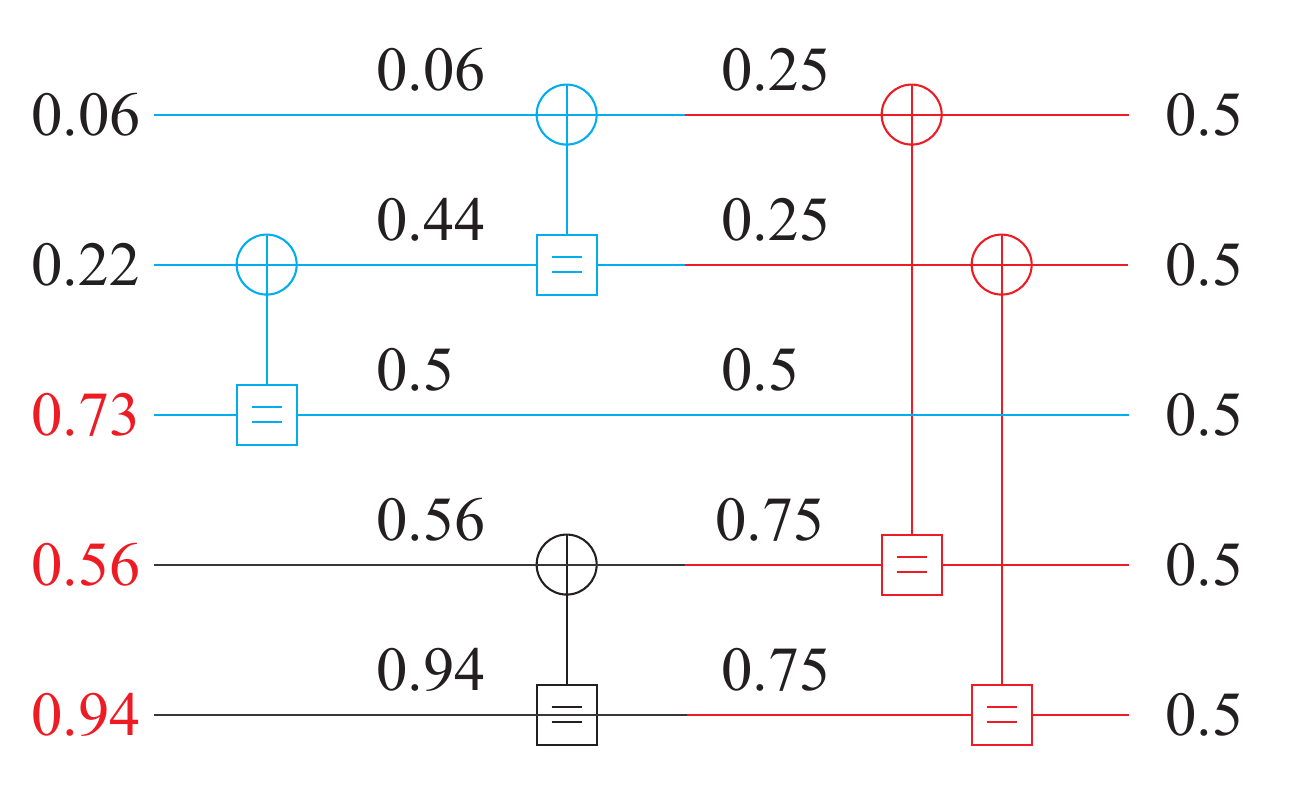}
    \caption{}
    \label{fig_N5K3_b}
  \end{subfigure}
  \caption{Two stitching perspectives for the same stitched polar code with $N=5$ and $K=3$.}
  \label{fig_N5K3}
\end{figure}

In short, thanks to the flexibility in stitched code construction, the stitching methods can be optimized for each target code length and dimension. Therefore, restructuring the polarization process to create asymmetric polarization significantly enhances the mutual information within the information set, in particular, the least reliable information bits. However, the code construction problem, along with a rigorous theoretical analysis and a practical engineering design, is substantially non-trivial and opens up new research opportunities.

\section{Preliminaries}

\subsection{Notations and Definitions}

We first introduce some notations that will be employed throughout the paper. Denote $[N] = \{1,\dots,N\}$ and $[M,N] = \{M,M+1,\dots,N\}$. For a set $A \subseteq [N]$, the complement set of $A$ in $[N]$ is denoted as $[N] \backslash A$.

For a vector $\bm{u}\in\FF_2^N$, denote $\bm{u}_A = (u_i)_{i\in A}$. If $A = [i,j]$, we also write  $\bm{u}_A = \bm{u}_i^j$. Let $\bm{u}_e = (u_i)_{i \text{ is even}}$ and $\bm{u}_o = (u_i)_{i \text{ is odd}}$ be the even and odd part of $\bm{u}$ respectively. Denote $d(\bm{u}, \bm{u}')$ to be Hamming distance between $\bm{u}$ and $\bm{u}'$ and $wt(\bm{u})$ to be Hamming weight of $\bm{u}$.

For $l\leq N$, $\bm{v}\in\FF_2^{l}, \bm{w}\in\FF_2^{N-l}, A\subseteq [N]$ with $|A|=l$, denote $\bm{v}\bowtie_A \bm{w}$ to be the length-$N$ vector such that $(\bm{v}\bowtie_A \bm{w})_A = \bm{v}$ and $(\bm{v}\bowtie_A \bm{w})_{[N]\backslash A} = \bm{w}$. For $\mathcal{V}\subseteq \FF_2^l, \mathcal{W}\subseteq \FF_2^{N-l}$, denote $\mathcal{V}\bowtie_A  \mathcal{W} = \{\bm{v} \bowtie_A \bm{w}\mid \bm{v}\in \mathcal{V}, \bm{w}\in \mathcal{W}\}$. Denote $d(\bm{v}, \mathcal{V}) = \min_{v'\in  \mathcal{V}}\{d(v, v')\}$ to be the Hamming distance between a vector and a vector set. Clearly, $d(\bm{u}, \mathcal{V}\bowtie_A  \mathcal{W}) = d(\bm{u}_A, \mathcal{V}) + d(\bm{u}_{[N]\backslash A}, \mathcal{W})$. 

For an $N\times N$ matrix $\bm{G}$ and sets $A_1,A_2\subseteq [N]$, let $\bm{G}_{A_1,A_2}$ be the submatrix of $\bm{G}$ with row indexed by $A_1$ and column indexed by $A_2$. 

Define $\log(\cdot)$ to be the logarithm to base 2 and $\oplus$ to be the addition modulo $2$.

\subsection{Polarization and Polar Codes}

Let $W: \mathcal{X} \rightarrow \mathcal{Y}$ be a binary memoryless symmetric (BMS) channel with input alphabet $\mathcal{X} = \{0,1\}$, output alphabet $\mathcal{Y}$, and transition probabilities $W(y|x), x \in \mathcal{X}, y \in \mathcal{Y}$. 

The symmetric capacity of the channel $W$ is 
\[
I(W) 	\triangleq \sum_{y \in \mathcal{Y}} \sum_{x \in \mathcal{X}} \frac{1}{2} W(y|x) \log \frac{W(y|x)}{\frac{1}{2} W(y|0) + \frac{1}{2} W(y|1)},
\]
and the Bhattacharyya parameter is 
\[
Z(W) 	\triangleq \sum_{y \in \mathcal{Y}} \sqrt{W(y|0) W(y|1)}.
\]

If $W$ is a BEC, then
\[I(W) 	= 1 - Z(W).\]

Let $W^0, W^1 :\mathcal{X}\to\mathcal{Y}$ be two BMS channels and the channel transformation $(W^0, W^1) \rightarrow (U^0, U^1)$ is defined as

\begin{equation}\label{eq:polar1}
U^0(y_1, y_2|x_1) = \sum_{x_2 \in \{0,1\}} \frac{1}{2} W^0(y_1|x_1 \oplus x_2)  W^1(y_2|x_2) 
\end{equation}

\begin{equation}\label{eq:polar2}
U^1(y_1, y_2, x_1|x_2) = \frac{1}{2}  W^0(y_1|x_1 \oplus x_2)  W^1(y_2|x_2) 
\end{equation}

For BEC, 
\[Z(U^0) 	= Z(W^0) + Z(W^1) - Z(W^0)Z(W^1);\]
\[Z(U^1) 	= Z(W^0)Z(W^1).\]

After channel combining and splitting operation on $N$ independent uses of $W$, we get $N$ polarized synthesized bit-channels $U_i^{(0)}$, defined as $(U_i^{(t)}, U_i^{(t)})\to (U_{2i-1}^{(t-1)}, U_{2i}^{(t-1)})$:

\begin{align}\label{eq:re_polar1}
\notag
& U_{2i-1}^{(t-1)}\left(\bm{y}_1^{2N}, \bm{u}_1^{2i-2}|u_{2i-1}\right) \\
\notag
&= \sum_{u_{2i}} \frac{1}{2} U_i^{(t)}\left(\bm{y}_1^N, \bm{u}_{1,o}^{2i-2} \oplus \bm{u}_{1,e}^{2i-2}|u_{2i-1} \oplus u_{2i}\right) \\
    &\quad \cdot U_i^{(t)}\left(\bm{y}_{N+1}^{2N}, \bm{u}_{1,e}^{2i-2}|u_{2i}\right);
\end{align}
\begin{align}\label{eq:re_polar2}
\notag
& U_{2i}^{(t-1)} \left(\bm{y}_1^{2N}, \bm{u}_1^{2i-1}|u_{2i}\right) \\
\notag
&= \frac{1}{2} U_i^{(t)}\left(\bm{y}_1^N, \bm{u}_{1,o}^{2i-2}\oplus \bm{u}_{1,e}^{2i-2}|u_{2i-1}\oplus u_{2i}\right)\\
&\cdot U_i^{(t)}\left(\bm{y}_{N+1}^{2N}, \bm{u}_{1,e}^{2i-2}|u_{2i}\right)
\end{align}
for $ 1\leq i\leq 2^{m-t}, 1\leq t\leq m$ and $U_{1}^{(m)} = W$.

The reliability of each bit-channel can be computed by DE \cite{Mori2009} and Gaussian Approximation (GA) \cite{Trifonov2012}. The information set $\MI\subseteq [N]$ consists of the $K$ most reliable bit-channels to carry messages. The remaining bit-channels in $\MF = [N]\backslash \MI$ form the frozen set, which are fixed to 0. The encoding process is formulated as
$$
\bm{c} = \bm{uF}_N.
$$
Here, $\bm{u}_{\MI} \in \{0,1\}^K, \bm{u}_{\MF} = \bm{0}$, $\bm{F} = \begin{bmatrix} 1&0 \\ 1&1 \end{bmatrix}$, $N=2^m$ and $\bm{F}_N=\bm{F}^{\otimes m}$, where $\otimes$ is Kronecker product.

\subsection{Successive Cancellation Decoding}

Let 
$$
L_{i,t} = \ln\frac{U_{i}^{(t)}(\cdot\mid 0)}{U_{i}^{(t)}(\cdot\mid 1)}
$$ 
denote the log likelihood ratio (LLR) for the channel $U_{t}^{(i)}$. Specifically, $L_{i,m}$ represents the received LLRs from physical channels. The LLRs $L_{i,t}$ are subsequently propagated from stage $t+1$ according to
\begin{align}\label{eq:SCf}
\notag
L_{i,t} &= f(L_{i,t+1},L_{i+2^t,t+1}) = \ln\left(\frac{e^{L_{i,t+1}+L_{i+2^t,t+1}}+1}{e^{L_{i,t+1}}+e^{L_{i+2^t,t+1}}}\right) \\
&= 2\tanh^{-1} \left(\tanh\frac{L_{i,t+1}}{2}\tanh\frac{L_{i+2^{t},t+1}}{2}\right);
\end{align}
\begin{align}\label{eq:SCg}
\notag
L_{i+2^t,t} &= g(\hat{u}_{i,t},L_{i,t+1},L_{i+2^t,t+1})\\
&= (-1)^{\hat{u}_{i,t}}L_{i,t}+L_{i+2^t,t+1}. 
\end{align}

At stage $0$, we make a hard decision
\begin{gather}\label{eq:SChard}
\hat{u}_{i,0} = \begin{cases}
    0 ,& \text{ if } i\in \MF; \\
    0 ,& \text{ if } i\in \MI, L_{i,0}\geq 0; \\
    1 ,& \text{ if } i\in \MI, L_{i,0}< 0. \\
\end{cases}
\end{gather}

Then $\hat{u}_{i,t}$ are propagated from stage $t-1$ according to
\[\hat{u}_{i,t} = \hat{u}_{i,t-1}\oplus \hat{u}_{i+2^t,t-1};\]
\[\hat{u}_{i+2^t,t} = \hat{u}_{i+2^t,t-1}.\]

\subsection{Speed of Polarization}
The speed of polarization quantifies how quickly polar codes approach channel capacity, thereby reflecting their performance at finite code lengths. The scaling law is a critical metric to quantify polarization speed, which indicates the necessary code length as the rate approaches channel capacity under a fixed error probability. According to \cite{Polyanskiy2010}, for every BMS channel $W$ with capacity $I(W)$, the required code length for a binary code with rate $R$ is at least $O(1/(I(W)-R)^{\mu})$ with error probability $P_e$, where $\mu$ is the scaling exponent. 

For polar codes, the following assumption holds:

\begin{assumption}[Scaling law]\cite{Hassani2014}\label{assum_scaling}
Assume that $W$ is a BMS channel with Bhattacharyya parameter $Z(W) = z$. Let $N = 2^m$ and $Z_N^{(i)}$ be the Bhattacharyya parameters of the bit-channels $W_m^{(i)}$. There exists $\mu \in (0, \infty)$ such that, for any $z, a, b \in (0, 1)$ such that $a < b$, the limit $\lim_{m \to\infty} 2^{m(\frac{1}{\mu}-1)} |\{1\leq i\leq N\mid Z_N^{(i)} \in [a, b])\}|$ exists in $(0, \infty)$. We denote this limit by $c(z, a, b)$. In other words,
\[
|\{1\leq i\leq N\mid Z_N^{(i)} \in [a, b])\}| = (c(z, a, b)+o(1))N^{\lambda}.
\]
Here, $\lambda = 1-1/\mu$.
\end{assumption}

It is known that $\mu = 3.63$ for regular polar codes over BEC \cite{Hassani2014}, compared to $\mu = 2$ for random codes. A length-64 kernel can reduce the scaling exponent $\mu$ from 3.627 to 2.87 \cite{Yao2019}.

\subsection{Rate Matching}

The mother code length of regular polar codes is limited to powers of two. To address practical applications, puncturing and shortening techniques are developed for polar codes.

\begin{definition}[Punctured Code and Puncturing Pattern]
Let $C(\MI)$ be a polar code with length $N=2^m$ and information set $\MI$, and $X\subseteq [N]$ be a set called puncturing pattern. The punctured code of $C(\MI)$ with puncturing pattern $X$ is $C_P(\MI,X) = \{\bm{c}_{[N]\backslash X} \mid \bm{c}\in C(\MI)\}$.
\end{definition}

\begin{definition}[Shortened Code and Shortening Pattern]
Let $C(\MI)$ be a polar code with length $N=2^m$ and information set $\MI$, and $Y\subseteq [N]$ be a set called by shortening pattern. The shortened code of $C(\MI)$ with shortening pattern $Y$ is $C_S(\MI,Y) = \{\bm{c}_{[N]\backslash Y} \mid \bm{c}\in C(\MI), \bm{c}_{Y} = \bm{0}\}$.
\end{definition}

Denote the quasi-uniform sequence to be $\bm{q}=(1,2,\dots,N)$. The QUP pattern $X_i$ \cite{Niu2013} is the set of the first $i$ bits in $\bm{q}$, i.e., $\{1,2,\dots,i\}$. The QUP polar code with length $N-i$ and information set $\MI$ is denoted by $C_P(\MI,X_i)$. 

For each $z\in \{1,2,...,2^m\}$, there is a unique binary representation $\bm{a}=(a_1,...,a_m)\in\FF_2^m$, where $a_1$ is the least significant bit, such that
\[z = D(\bm{a}) \triangleq \sum_{i=1}^{m} 2^{i-1} (a_i \oplus 1) + 1,\]
where $\oplus$ is the mod-$2$ sum in $\FF_2$.

Denote the bit-reversal sequence $\bm{q}'$ to be the bit reversal of $\bm{q}$, i.e., $ \bm{q}'_{D(a_m,\dots,a_1)} =  \bm{q}_{D(a_1,\dots,a_m)} = D(a_1,\dots,a_m)$. The BRS pattern $Y_i$ is the last $i$ bits of $\bm{q}'$, and the  BRS polar code \cite{Bioglio2017} with length $N-i$ and information set $\MI$ is denoted as $C_S(\MI,Y_i)$.

\section{Stitched Polar Codes}\label{sec:GPC}

In this section, we introduce stitched polar codes. We begin by presenting a coupling sequence representation for encoding. Based on this representation, we provide a rigorous definition of stitched polar codes.

For a binary vector $\bm{u} = (u_1,\dots, u_a, \dots, u_b, \dots, u_N)$, the operation $(a,b)$ transforms $\bm{u}$ to $(a,b) \circ  \bm{u} = (u_1,\dots, u_a\oplus u_b, \dots, u_b, \dots, u_N)$. In other word, $(a,b)$ represents a basic $2\times 2$ polarization transformation on the $a$-$th$ and $b$-$th$ bits. Then a coupling sequence of integer pairs $(a_1, b_1), \dots, (a_n, b_n)$ where $1\leq a_i < b_i \leq N$ for all $1\leq i\leq n$ describes an encoding process. Specifically, the message vector $\bm{u}$ with $\bm{u}_{\mathcal{F}} = \bm{0}$ is encoded into the codeword $\bm{x} = (a_n, b_n)\circ \dots \circ (a_1, b_1)\circ \bm{u}$. 

We define the intermediate vector $\bm{u}^{(i)} = (a_i, b_i)\circ\bm{u}^{(i-1)}$ for $1\leq i\leq n$ with the initial vector $\bm{u}^{(0)} = \bm{u}$. This recursive definition implies that $\bm{u}^{(i+1)} = \bm{u}^{(i)}\bm{T}^{(i)}$, where $\bm{T}^{(i)}$ is a $N\times N$ matrix with $T^{(i)}_{k,k} = 1$ for all $1\leq k\leq N$, $T^{(i)}_{b_i, a_i} = 1$, and the other entries are 0. Consequently, the generator matrix of the code is $\bm{G} = \bm{T}^{(1)} \cdots \bm{T}^{(n)}$, i.e., the codeword $\bm{x} = \bm{u}^{(n)} = \bm{uG}$. In other words, the generation of $\bm{G}$ begins with the identity matrix, then adding the $b_1$-$th$ column to the $a_1$-$th$ column, adding the $b_2$-$th$ column to the $a_2$-$th$ column, and so on.

\begin{example}\label{Ex:Enc_Gra}
Consider the coupling sequence $(3,4), (1,2), (3,5), (1,3), (2,5)$.

The encoding process:

Message vector: $(u_1, u_2, u_3, u_4, u_5)$

$(3,4)$: $(u_1, u_2, u_3\oplus u_4, u_4, u_5)$

$(1,2)$: $(u_1\oplus u_2, u_2, u_3\oplus u_4, u_4, u_5)$

$(3,5)$: $(u_1\oplus u_2, u_2, u_3\oplus u_4\oplus u_5, u_4, u_5)$.

$(1,3)$: $(u_1\oplus u_2 \oplus u_3\oplus u_4\oplus u_5, u_2, u_3\oplus u_4\oplus u_5, u_4, u_5)$

$(2,5)$: $(u_1\oplus u_2 \oplus u_3\oplus u_4\oplus u_5, u_2\oplus u_5, u_3\oplus u_4 \oplus u_5, u_4, u_5)$

Let the information set be $\MI=\{4,5\}$. For example, consider the information bits to be $\bm{u}_{\MI} = (1,0)$, which corresponds to the message vector $\bm{u} = (0,0,0,1,0)$. Applying the encoding procedure to $\bm{u}$ results in the codeword $\bm{x} = (1,0,1,1,0)$. The complete encoding process is depicted in Fig. \ref{fig_encode1}.
 
The generator matrix can be constructed by starting from an identity matrix, adding the forth column to the third, then adding the second column to the first, repeating this process subsequently. Finally, we have
$$
\bm{G} = \begin{bmatrix}
1 & 0  & 0  & 0  & 0 \\
1 & 1  & 0  & 0  & 0 \\
1 & 0  & 1  & 0  & 0 \\
1 & 0  & 1  & 1  & 0 \\
1 & 1  & 1  & 0  & 1
\end{bmatrix}.
$$
Then $\bm{x} = \bm{uG}$. 
\end{example}

\begin{figure}[!t]
\centering
\includegraphics[width=0.5\textwidth]{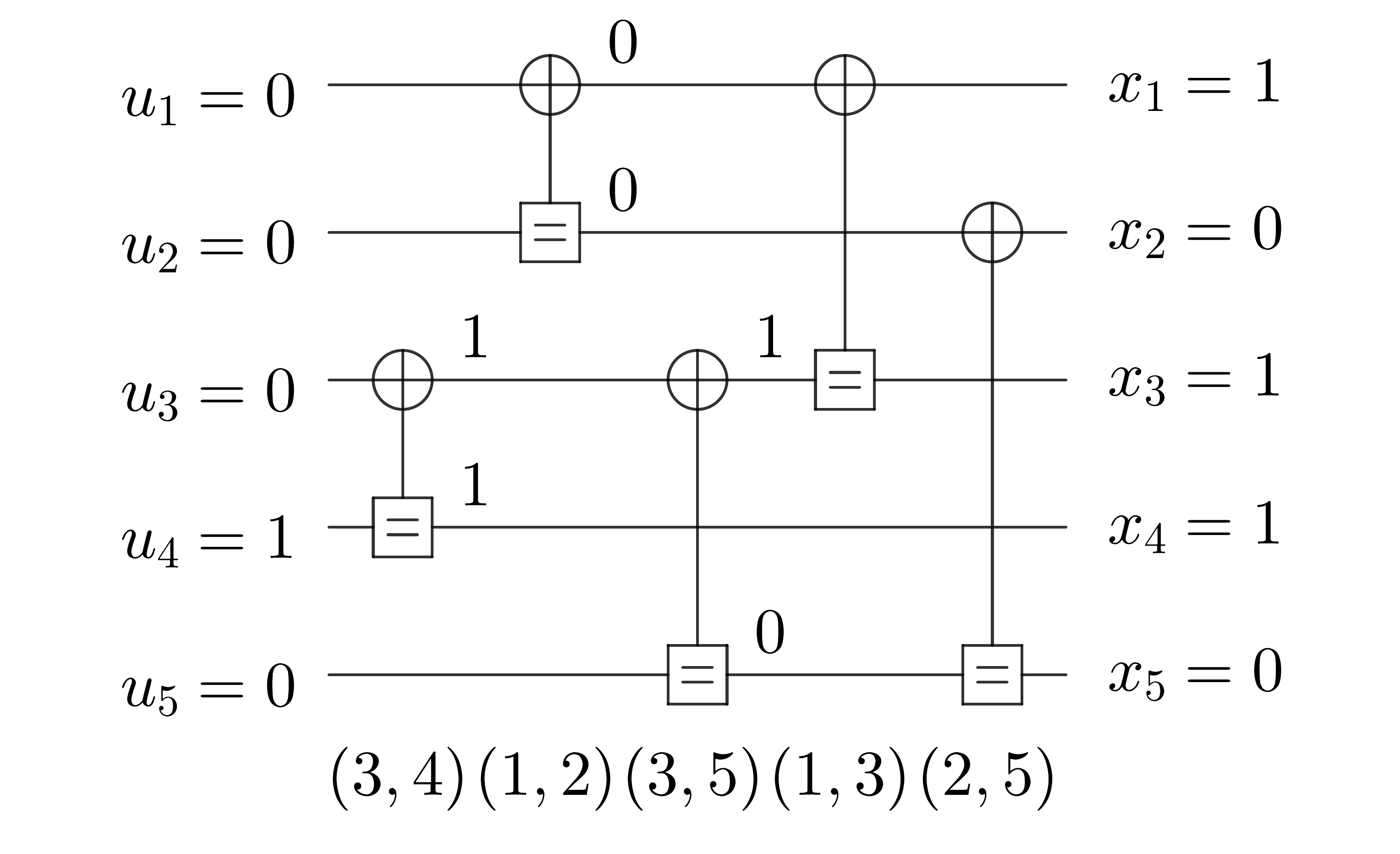}
\caption{The encoding process.}
\label{fig_encode1}
\end{figure}

The encoding process is summarized in Algorithm \ref{alg:IKE} with complexity $O(n)$.

\begin{figure}[!t]
\begin{algorithm}[H]
\caption{encoder($C$)}
\begin{algorithmic}[1]\label{alg:IKE}

\renewcommand{\algorithmicrequire}{\textbf{Input:}}
\renewcommand{\algorithmicensure}{\textbf{Output:}}
\REQUIRE the linear code $C = (a_1, b_1), \dots, (a_n, b_n)$, the message vector $\bm{u}$.
\ENSURE the codeword $\bm{x}$.
\STATE $\bm{x}\gets \bm{u}$;
\FOR{$i = 1$ to $n$}
 \STATE $x_{a_i}\gets x_{a_i}\oplus x_{b_i}$;
\ENDFOR

\end{algorithmic}
\end{algorithm}
\end{figure}

\subsection{Channel Combining and Splitting}

Proposition \ref{pro_CS} describes the channel combining and splitting process for coupling sequences, extending Eq. (\ref{eq:polar1}) and (\ref{eq:polar2}) from regular polar codes. The generalization is based on the combination of two different channels \cite{lelewang}. This allows SC decoding using f/g-operations.

\begin{proposition}\label{pro_CS}
Let $W:\mathcal{X}\to \mathcal{Y}$ be a BMS channels. The channel combining and splitting process of the coupling sequence $(a_1, b_1), \dots, (a_n, b_n)$ can be defined block-wise as $(W^{(i)}_{a_i}, W^{(i)}_{b_i}) \to (W^{(i-1)}_{a_i}, W^{(i-1)}_{b_i})$:
\begin{align}
&  W^{(i-1)}_{a_i} \left( \bm{y}_{A^{(i-1)}_{a_i}}, \mathcal{V}^{(i-1)}_{a_i} \mid u^{(i-1)}_{a_i} \right) \notag \\
= & \sum_{u^{(i-1)}_{b_i}} \frac{1}{2} W^{(i)}_{a_i} \left( \bm{y}_{A^{(i)}_{a_i}}, \mathcal{V}^{(i)}_{a_i} \mid u^{(i)}_{a_i} \oplus u^{(i)}_{b_i} \right) \notag \\
\cdot & W^{(i)}_{b_i} \left( \bm{y}_{A^{(i)}_{b_i}}, \mathcal{V}^{(i)}_{b_i} \mid u^{(i)}_{b_i} \right)
\end{align}
\begin{align}
&  W^{(i-1)}_{b_i} \left( \bm{y}_{A^{(i-1)}_{b_i}}, \mathcal{V}^{(i-1)}_{b_i} \mid u^{(i-1)}_{b_i} \right) \notag \\
= &  \frac{1}{2} W^{(i)}_{a_i} \left( \bm{y}_{A^{(i)}_{a_i}}, \mathcal{V}^{(i)}_{a_i} \mid u^{(i)}_{a_i} \oplus u^{(i-1)}_{b_i} \right)  \notag \\
\cdot & W^{(i)}_{b_i} \left( \bm{y}_{A^{(i)}_{b_i}}, \mathcal{V}^{(i)}_{b_i} \mid u^{(i)}_{b_i} \right).
\end{align}
The intermediate channel remains unchanged. Here, $W^{(k)}_{j}$ represent the $j$-$th$ channel after $k$ polarization stages. For $1\leq j\leq N$, $A_j^{(n)} = \{j\}$ and $\mathcal{V}_j^{(n)} = \varnothing$. For $1\leq i \leq n$, $A^{(i-1)}_{a_i} = A^{(i)}_{a_j} \cup A^{(i)}_{b_i}$, $A^{(i-1)}_{b_i} = A^{(i-1)}_{a_i}$ and $\mathcal{V}^{(i-1)}_{a_i} = \mathcal{V}^{(i)}_{a_i} \cup \mathcal{V}^{(i)}_{b_i}$. 
\end{proposition}

The channel combining and splitting requires that the $a_i$-$th$ and $b_i$-$th$ channels involved in the coupling $(a_i, b_i)$ carry independent information, similar to a regular polar transformation. However, not all coupling sequences are valid in this sense. Algorithm \ref{alg:CS} in Appendix A checks the validity of coupling sequences. Interested readers can find this algorithm and the proof of Proposition \ref{pro_CS} there.

According to Proposition \ref{pro_CS}, codes validated by Algorithm \ref{alg:CS} can undergo SC decoding using f/g-operations, as specified in Eq. (\ref{eq:SCf}) and (\ref{eq:SCg}) in their natural order.

\begin{definition}
A stitched polar code of length $N$ is defined as a validated coupling sequence $(a_1, b_1), \dots, (a_n, b_n)$ where $1\leq a_i < b_i \leq N$ for all $1\leq i\leq n$. 
\end{definition}

Clearly, the regular polar code can be represented as a special case of stitched polar codes: $(1, \frac{N}{2}+1), (2, \frac{N}{2}+2), \dots, (\frac{N}{2}, N), (1, \frac{N}{4}+1), (2, \frac{N}{4}+2) \dots, (\frac{3N}{4}, N), \dots , (1, 2), (3,4),\dots, (N-1, N)$. A rate-matched regular polar code is formed by selecting a subsequence of a mother-length polar code. This inherent restriction limits their performance. Stitched polar codes are designed by strategically arranging these basic transformations.

\begin{example}\label{Ex:Pol_Gra}

Consider the code $(3,4), (1,2), (3,5), (1,3), (2,5)$ from Example \ref{Ex:Enc_Gra}. Note that the codeword is given by $\bm{u}^{(n)} = \bm{uG} = (u_1\oplus u_2 \oplus u_3\oplus u_4\oplus u_5, u_2\oplus u_5, u_3\oplus u_4 \oplus u_5, u_4, u_5)$.  It can be validated that the code passes Algorithm \ref{alg:CS}.

The polarization process:

(2,5):
\begin{align*}
    & W^{(1)}_2 \left((y_2, y_5) \mid u_2\right) \\
    &= \sum_{u_5} \frac{1}{2} W^{(0)}_2\left(y_2 \mid u_2\oplus u_5\right)\cdot W^{(0)}_5\left(y_5 \mid u_5\right) 
\end{align*}
\begin{align*}
    & W^{(1)}_5 \left((y_2, y_5), u_2 \mid u_5\right) \\
    &= \frac{1}{2} W^{(0)}_2\left(y_2 \mid u_2\oplus u_5\right)\cdot W^{(0)}_5\left(y_5 \mid u_5\right) 
\end{align*}

(1,3):
\begin{align*}
    & W^{(2)}_1 \left((y_1, y_3) \mid u_1\oplus u_2\right) \\
    &= \sum_{u_3\oplus u_4 \oplus u_5} \frac{1}{2} W^{(1)}_1\left(y_1 \mid u_1\oplus u_2 \oplus u_3\oplus u_4\oplus u_5\right)\cdot \\
    & W^{(1)}_3\left(y_3 \mid  u_3\oplus u_4 \oplus u_5\right) \\
    \\
    & W^{(2)}_3 \left((y_1, y_3),  u_1\oplus u_2 \mid u_3\oplus u_4 \oplus u_5\right) \\
    &= \frac{1}{2} W^{(1)}_1\left(y_1 \mid u_1\oplus u_2 \oplus u_3\oplus u_4\oplus u_5\right)\cdot \\
    & W^{(1)}_3\left(y_3 \mid  u_3\oplus u_4 \oplus u_5\right)  
\end{align*}

(3,5):
\begin{align*}
    & W^{(3)}_3 \left((y_1, y_2, y_3, y_5),  (u_1\oplus u_2, u_2) \mid u_3\oplus u_4\right) \\
    &= \sum_{u_5} \frac{1}{2}W^{(2)}_3\left((y_1, y_3),  u_1\oplus u_2 \mid u_3\oplus u_4\oplus u_5\right)\cdot \\
    & W^{(2)}_5\left((y_2, y_5), u_2 \mid u_5\right) 
\end{align*}
\begin{align*}
    & W^{(3)}_5 \left((y_1, y_2, y_3, y_5), (u_1\oplus u_2, u_2, u_3\oplus u_4) \mid u_5 \right) \\
    &= \frac{1}{2} W^{(2)}_3\left((y_1, y_3),  u_1\oplus u_2 \mid u_3\oplus u_4\oplus u_5\right)\cdot \\
    & W^{(2)}_5\left((y_2, y_5), u_2 \mid u_5\right) 
\end{align*}

(1,2):
\begin{align*}
    & W^{(4)}_1 \left((y_1,y_2,y_3,y_5) \mid u_1\right) \\
    &= \sum_{u_2} \frac{1}{2} W^{(3)}_1 \left((y_1, y_3) \mid u_1\oplus u_2\right)\cdot \\
    & W^{(3)}_2\left((y_2, y_5) \mid u_2 \right) 
\end{align*}
\begin{align*}
    & W^{(4)}_2 \left((y_1,y_2,y_3,y_5), u_1 \mid u_2\right) \\
    &= \frac{1}{2} W^{(3)}_1 \left((y_1, y_3) \mid u_1\oplus u_2\right)\cdot W^{(3)}_2\left((y_2, y_5) \mid u_2 \right) 
\end{align*}

(3,4):
\begin{align*}
    & W^{(5)}_3 \left(\bm{y}_1^5,(u_1\oplus u_2, u_2) \mid u_3\right) \\
    &= \sum_{u_4} \frac{1}{2}  W^{(4)}_3 \left((y_1, y_2, y_3, y_5),  (u_1\oplus u_2, u_2) \mid u_3\oplus u_4\right) \cdot \\
    & W^{(4)}_4\left(y_4 \mid u_4\right) 
\end{align*}
\begin{align*}
    & W^{(5)}_4 \left(\bm{y}_1^5,(u_1\oplus u_2, u_2, u_3) \mid u_4\right) \\
    &= \frac{1}{2}  W^{(4)}_3 \left((y_1, y_2, y_3, y_5),  (u_1\oplus u_2, u_2) \mid u_3\oplus u_4\right) \cdot \\
    & W^{(4)}_4\left(y_4 \mid u_4\right) 
\end{align*}

And $W^{(i-1)}_{j} = W^{i)}_{j}$ for $j \neq a_i, b_i$.

\end{example}

Based on the polarization process, the reliability of channels can be computed by DE/GA as in regular polar codes. Let $R^{(i)}_k$ denote the probability density function of the LLR of $W^{(i)}_{k}$ when the all-zero information is transmitted. For each integer pair $(a_i, b_i)$, we calculate the densities as follows:
$$
R^{(i-1)}_{a_i} = R^{(i)}_{a_i} * R^{(i)}_{b_i}; R^{(i-1)}_{b_i} = R^{(i)}_{a_i} \boxed{*} R^{(i)}_{b_i}.
$$
If  $j \neq a_i, b_i$, then $R^{(i-1)}_{j}  = R^{(i)}_{j}$. Here, $R^{(i)}_{0}$ is the probability density function of LLR of the received message when 0 is transmitted, $*$ and $\boxed{*}$ are the convolution operations in a variable node domain and a check node domain, respectively \cite{Richardson2008}. Through DE/GA, information bit positions can be identified as the most reliable bit-channels.

\begin{example}\label{Ex:DE_Gra}
Consider the stitched polar code $(3,4), (1,2), (3,5), (1,3), (2,5)$ from Example \ref{Ex:Enc_Gra} and \ref{Ex:Pol_Gra}. For an addictive white Gaussian noise (AWGN) channel with SNR 6.021 dB, which corresponds to a mean LLR of 8, GA (Eq. (5)-(6) in \cite{Trifonov2012}) is applied to evaluate the mean LLR of synthetic channels, with the full calculation process illustrated in Fig. \ref{fig_GA}. 
\end{example}

\begin{figure}[!t]
\centering
\includegraphics[width=0.4\textwidth]{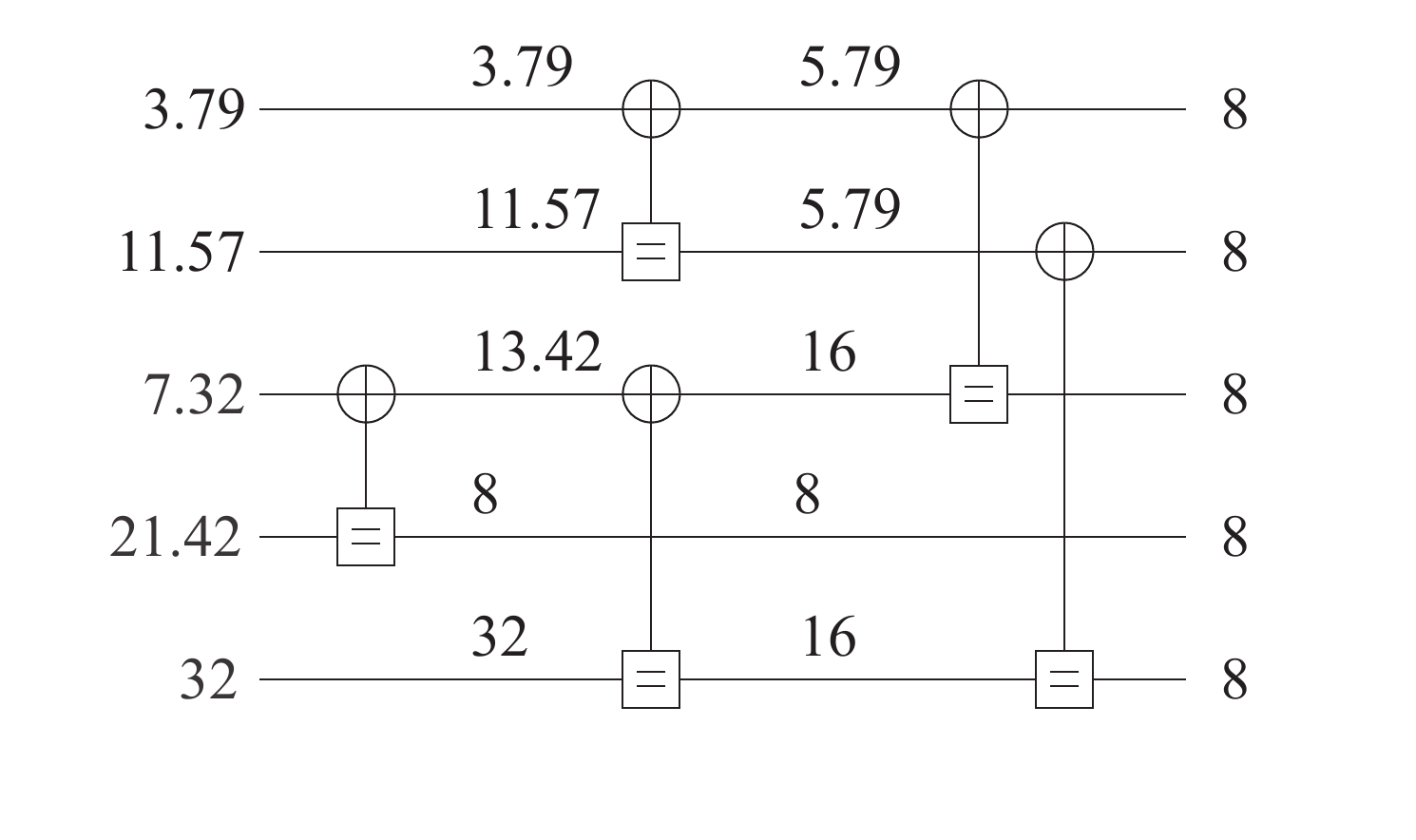}
\caption{The GA process of length-5 stitched polar code.}
\label{fig_GA}
\end{figure}

For short-length codes, we can find the best constructions using a computational search. Specifically, appendix B lists the optimal stitched polar code constructions with code length less than or equal to 8.

\subsection{SC Decoding for Stitched Polar Codes}

In this subsection, we elaborate on the SC decoding procedure for stitched polar codes. The decoding process follows a specific order. A $2\times 2$ transformation element $(a,b)$ acts as a decoder element with three states:

1. \emph{Receiving LLRs}: Upon receiving the LLRs $L_a$ and $L_b$, the element activates the f-operation defined in Eq. (\ref{eq:SCf}). This produces the updated LLR $L_a' = f(L_a, L_b)$, which is sent to the subsequent decoder element  associated with $a$. If $(a,b)$ is the last (leftmost) element  associated with $a$, then $\hat{u}_a$ is determined according to Eq. (\ref{eq:SChard}) and returned to $(a,b)$.

2. \emph{Receiving $\hat{u}_a$}: When the hard decision $\hat{u}_a$ is available, the element  activates the g-operation defined in Eq. (\ref{eq:SCg}). This generates the updated LLR $L_b' = g(L_a, L_b,\hat{u}_a)$ for the next element  associated with $b$, or $\hat{u}_b$ is determined and returned to $(a,b)$ if $(a,b)$ is the last element  associated with $b$.

3. \emph{Receiving $\hat{u}_b$}: When $\hat{u}_b$ is available, the element  calculates $\hat{u}_a' = \hat{u}_a\oplus \hat{u}_b$, $\hat{u}_b' = \hat{u}_b$ and then sends them back to the previous decoders associated with $a$ and $b$, respectively. 

Decoding begins at the rightmost elements, which receive the channel LLRs. Whenever a decoder element receives the messages specified in any of the above three cases, the corresponding decoding processing is triggered. This means the decoding scheduling is dependent on the stitching structure, rather than the fixed scheduling in regular polar codes. However, the decoding processing for each of the three cases is the same.

\begin{figure}[!t]
\begin{algorithm}[H]
\caption{stitched\_polar\_SC\_decoder($C, \bm{L}$)}
\begin{algorithmic}[1]\label{alg:DIKDS}

\renewcommand{\algorithmicrequire}{\textbf{Input:}}
\renewcommand{\algorithmicensure}{\textbf{Output:}}
\REQUIRE the stitched polar codes $C = (a_1, b_1), \dots, (a_n, b_n)$ with the information set $\MI$, the channel LLR vector $\bm{L}$.
\ENSURE the decoding result $\hat{\bm{u}}$.
\STATE Send $\bm{L}$ to the rightmost elements;
\WHILE{any element $(a,b)$ has unprocessed receiving messages}
 \STATE \textbf{switch}
 \STATE \textbf{Case 1:} Receiving LLRs $L_a$ and $L_b$ 
 \STATE $L_a' \gets f(L_a, L_b)$ // f-operation in Eq. (\ref{eq:SCf})
  \IF{$(a,b)$ is the leftmost element associated with $a$}
   \STATE $\hat{\bm{u}}_a \gets$ the hard decision of $L_a'$;  // Eq. (\ref{eq:SChard})
  \ELSE
   \STATE Send $L_a'$ to the next element associated with $a$;
  \ENDIF  

  \STATE \textbf{Case 2:} Receiving $\hat{\bm{u}}_a$ 
  \STATE $L_b' \gets g(L_a, L_b, \hat{u}_a)$  // g-operation in Eq. (\ref{eq:SCg})
   \IF{$(a,b)$ is the leftmost element associated with $b$}
  \STATE $\hat{\bm{u}}_b \gets$ the hard decision of $L_b'$;  // Eq. (\ref{eq:SChard})
   \ELSE
  \STATE Send $L_b'$ to the next element associated with $b$;
   \ENDIF 
   
  \STATE \textbf{Case 3:} Receiving $\hat{\bm{u}}_b$
  \STATE $\hat{u}_a' \gets \hat{u}_a \oplus \hat{u}_b$;
  \STATE $\hat{u}_b' \gets \hat{u}_b$;
  \STATE Return $\hat{u}_a'$ and $\hat{u}_b'$ to the previous elements associated with $a$ and $b$;
\ENDWHILE
\end{algorithmic}
\end{algorithm}
\end{figure}

Algorithm \ref{alg:DIKDS} demonstrates the SC decoding procedure described above. The decoding sequence can be pre-computed and stored. We employ a concise notation to represent the decoding operations: $(a,b,f)$ for $L'_a\gets f(L_a, L_b)$, $(a,b,g)$ for $L_b'\gets g(L_a, L_b, \hat{u}_a)$, $(a,b,\oplus)$ for $\hat{u}'_a \gets \hat{u}_a \oplus \hat{u}_b, \hat{u}_b' \gets \hat{u}_b$ and $(a,d)$ for hard decision for the $a$-$th$ bit. The complexity of this order-based decoding scales linearly with the number of decoder elements. 

\begin{example}
Consider the stitched polar code $(3,4), (1,2), (3,5), (1,3), (2,5)$ in Example \ref{Ex:Enc_Gra}, \ref{Ex:Pol_Gra} and \ref{Ex:DE_Gra}. The information set is $\{4, 5\}$.

The decoding sequence is $(1, 3, f), (2, 5, f), (1, 2, f), \\
(1, d), (1, 2, g), (2, d), (1, 2, \oplus), (1, 3, g), (2, 5, g), (3, 4, f), \\
(3, 5, f),  (3, d), (3, 4, g), (4, d), (3, 4, \oplus), (3, 5, g), (5, d), \\
(3, 5, \oplus), (1, 3, \oplus),  (2, 5, \oplus)$.  Here, we employ f-operation as $L_a'  = \text{sign}(L_a) \text{sign}(L_b) \min(|L_a|,|L_b|)$ for simplicity.

Assume the input LLR vector for the stitched polar code is $(2, 7.5, -4, -9, 3.5)$ over AWGN channel. We demonstrate the complete SC decoding process through a diagram in Fig. \ref{fig_decode}. In the diagram, the numbers in the red boxes represent the results calculated at each step. The numbers with a green background show the hard decision values transmitted from left to right. The red vertical lines indicate that the element is currently performing an f-operation, while the blue vertical lines indicate that the element is performing a g-operation.

Step 1: Elements (1,3) and (2,5) are activated to run the f-operation: $f(2, -4) = -2$ and $f(7.5, 3.5) = 3.5$. The results -2 and 3.5 are sent to the element (1,2).

Step 2: Element (1,2) runs the f-operation. As (1,2) is the last element for bit 1, $\hat{u}_1$ is determined. Since bit 1 is a frozen bit, $\hat{u}_1$ is set to 0 and sent back to the element.

Step 3: Element (1,2) runs the g-operation, calculated as $(-1)^0 \cdot -2 + 3.5 = 1.5$. Since bit 2 is also a frozen bit, $\hat{u}_2 = 0$ is sent back to the element.

Step 4: Element (1,2) calculates hard decisions $\hat{u}_1 \oplus \hat{u}_2$ and $\hat{u}_2$ are sent back to elements (1,3) and (2,5). Then both elements run the g-operation.

Step 5: Element (3,5) runs the f-operation.

Step 6: Element (3,4) runs the f-operation. $\hat{u}_3$ is frozen to 0.

Step 7: Element (3,4) runs the g-operation to determine the information bit $\hat{u}_4$. Given that the LLR is -15, the hard decision result is 1.

Step 8: Finally, the element (3,5) runs the g-operation, determining $\hat{u}_5$ to be 0. With that, we have completed the SC decoding. 

\begin{figure}[htbp]
  \begin{minipage}[b]{0.49\linewidth}
    \centering
    \includegraphics[width=\textwidth]{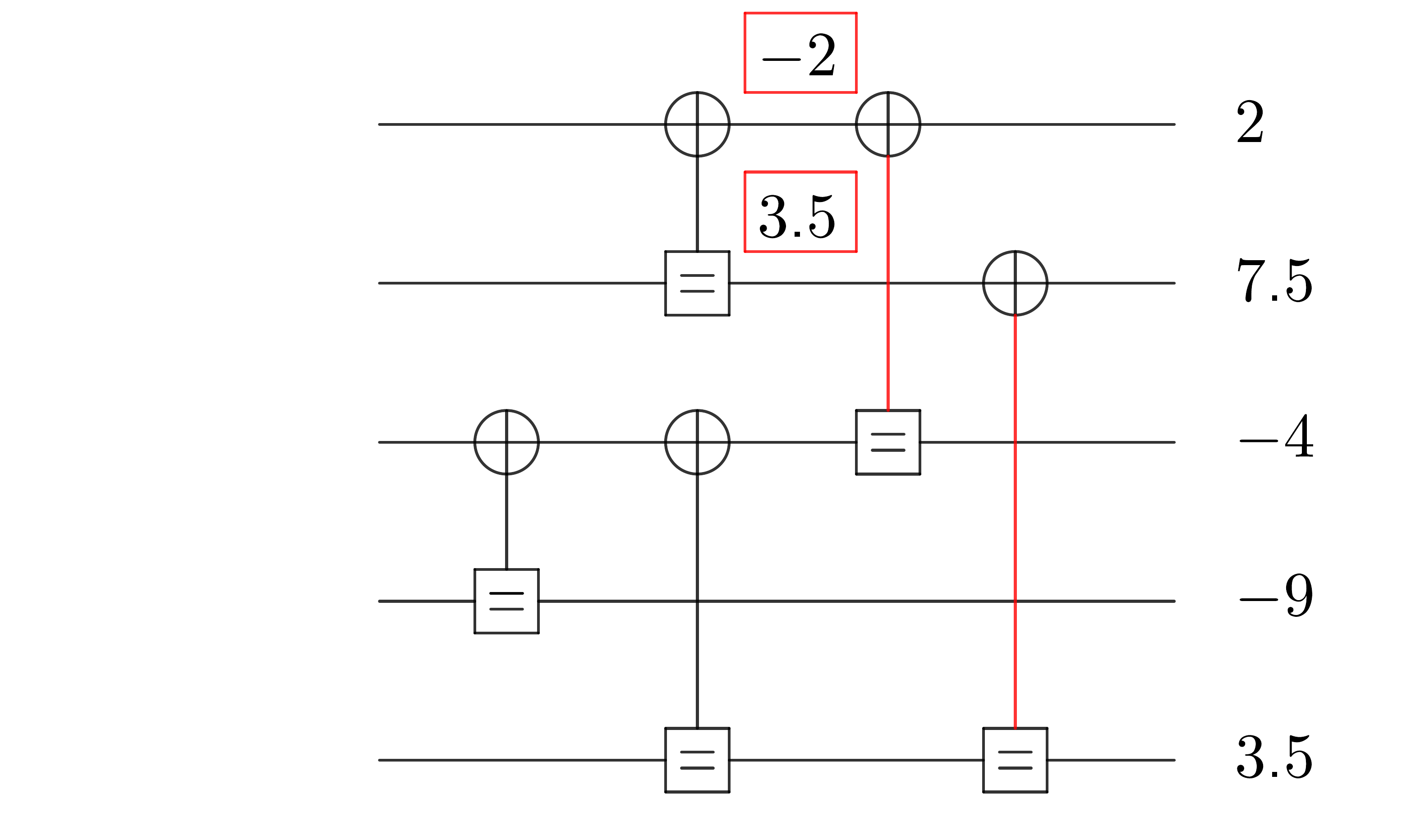}
    \caption*{step 1}
  \end{minipage}
  \hfill
  \begin{minipage}[b]{0.49\linewidth}
    \centering
    \includegraphics[width=\textwidth]{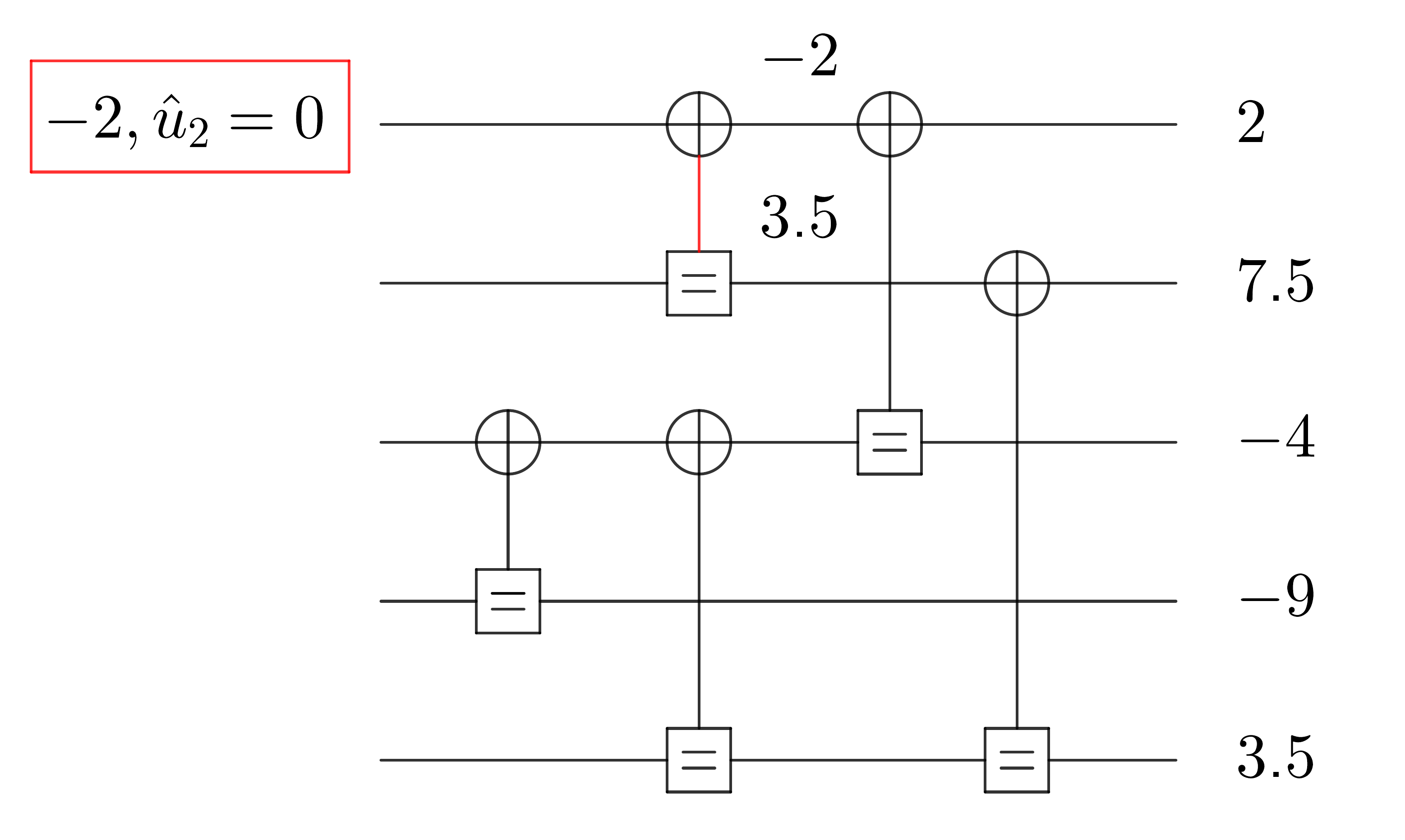}
    \caption*{step 2}
  \end{minipage}
    \begin{minipage}[b]{0.49\linewidth}
    \centering
    \includegraphics[width=\textwidth]{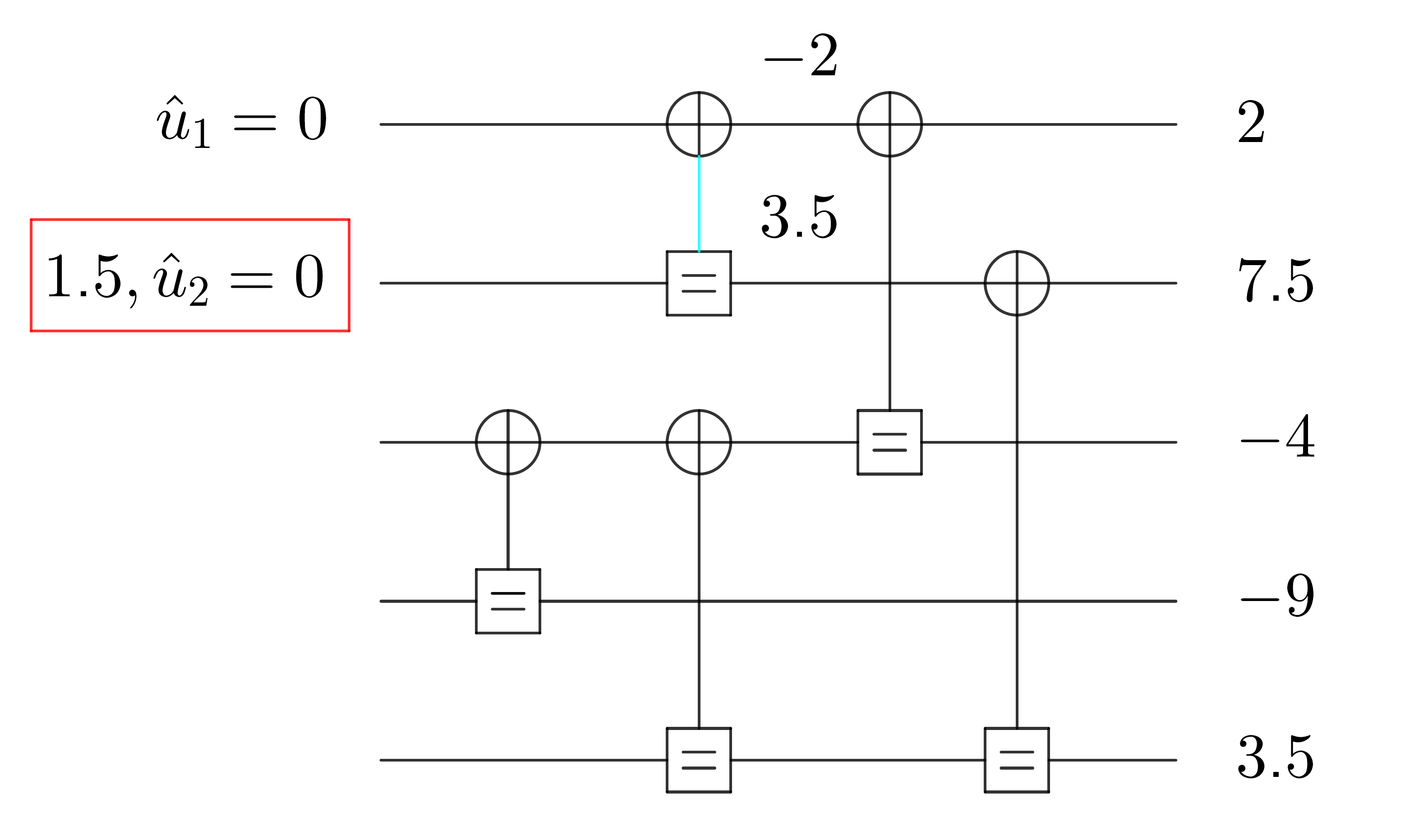}
    \caption*{step 3}
  \end{minipage}
  \hfill
  \begin{minipage}[b]{0.49\linewidth}
    \centering
    \includegraphics[width=\textwidth]{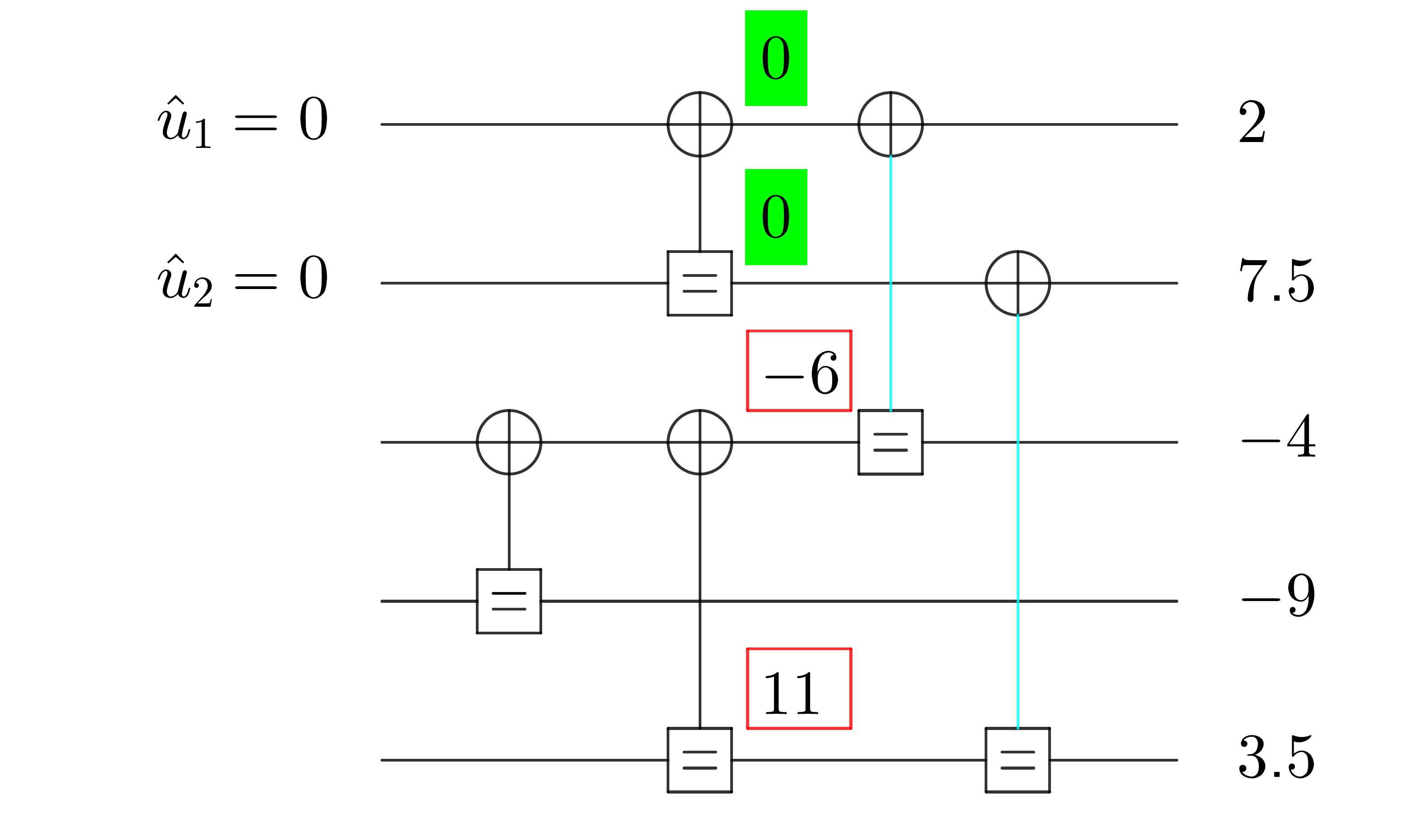}
    \caption*{step 4}
  \end{minipage}
    \begin{minipage}[b]{0.49\linewidth}
    \centering
    \includegraphics[width=\textwidth]{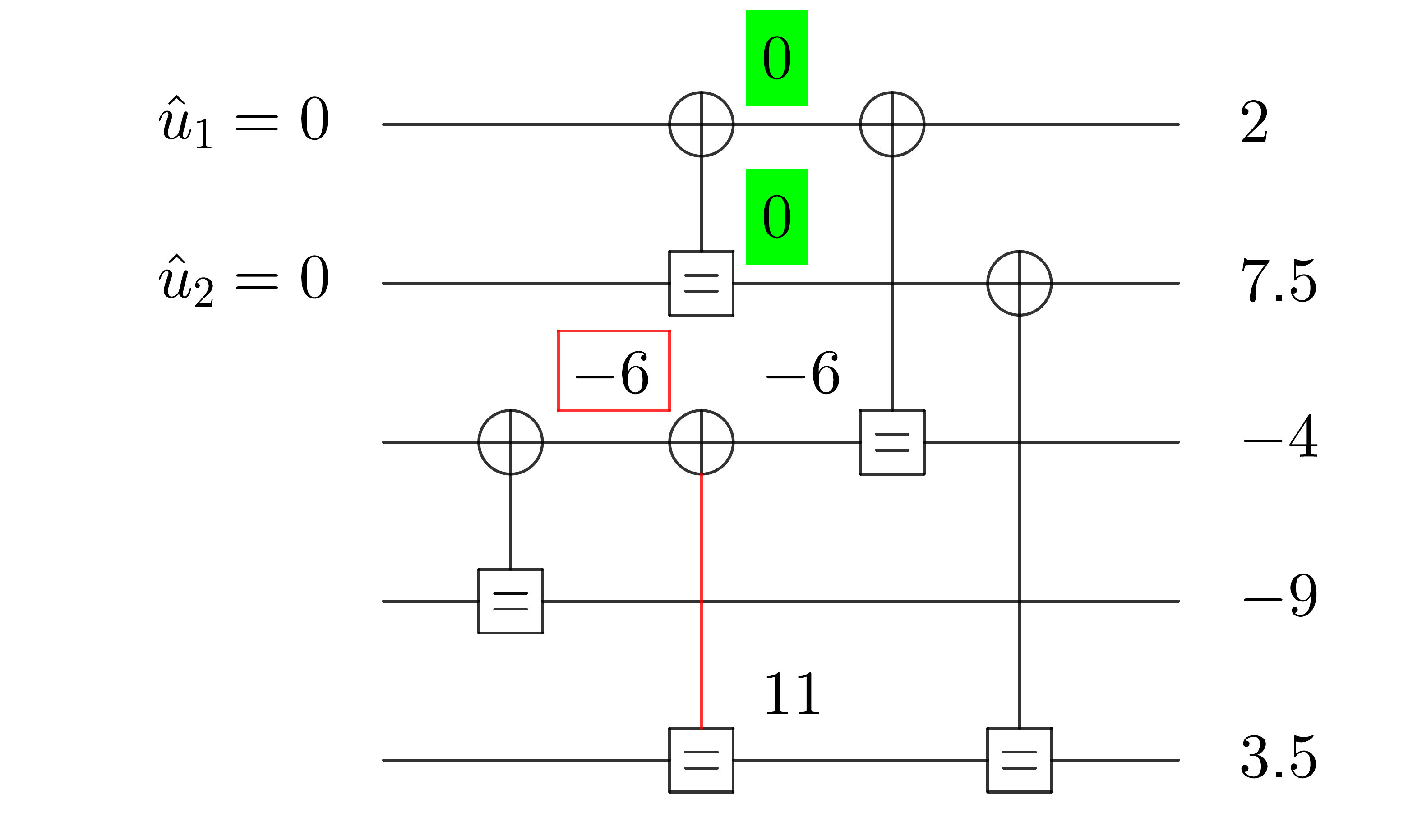}
    \caption*{step 5}
  \end{minipage}
  \hfill
  \begin{minipage}[b]{0.49\linewidth}
    \centering
    \includegraphics[width=\textwidth]{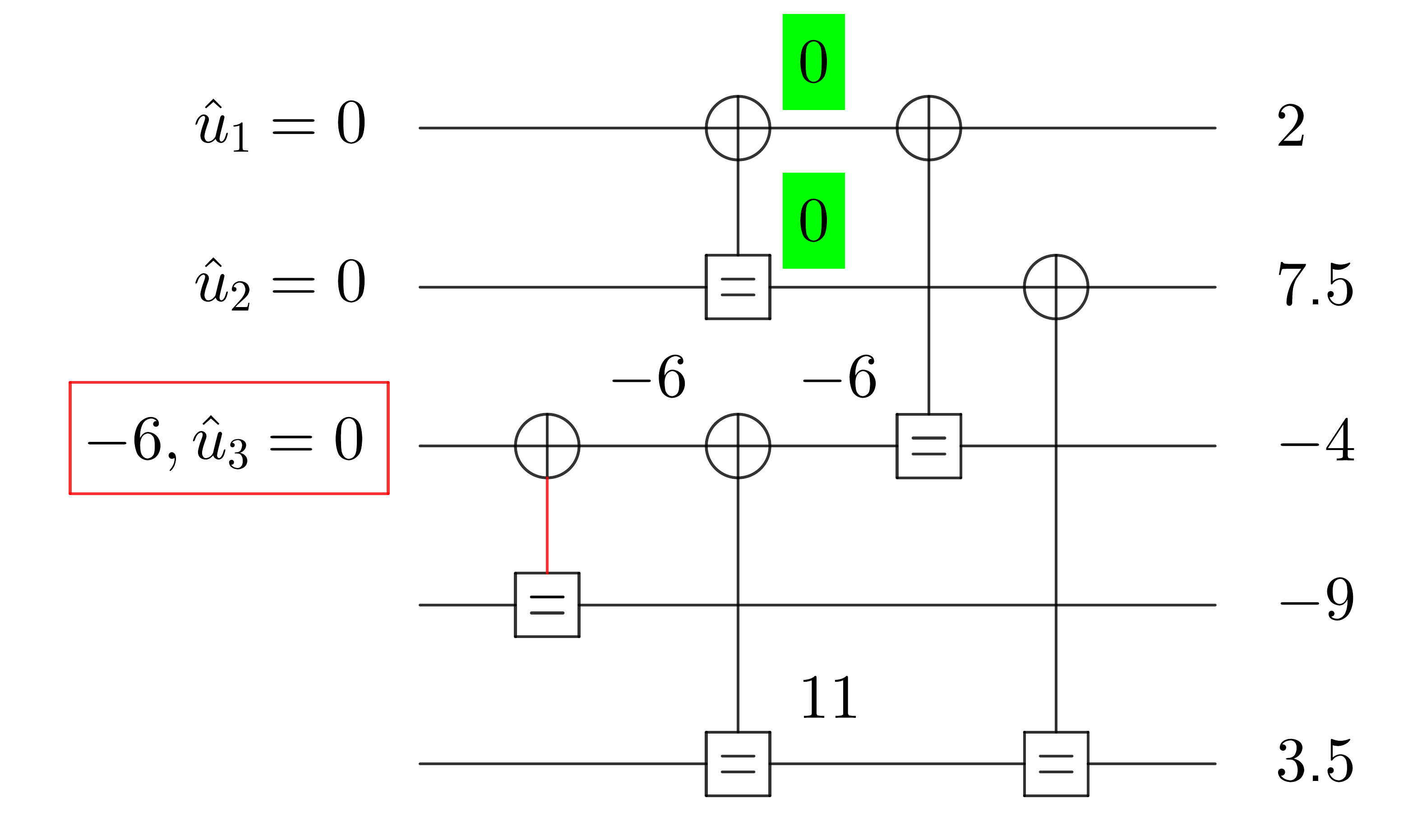}
    \caption*{step 6}
  \end{minipage}
    \begin{minipage}[b]{0.49\linewidth}
    \centering
    \includegraphics[width=\textwidth]{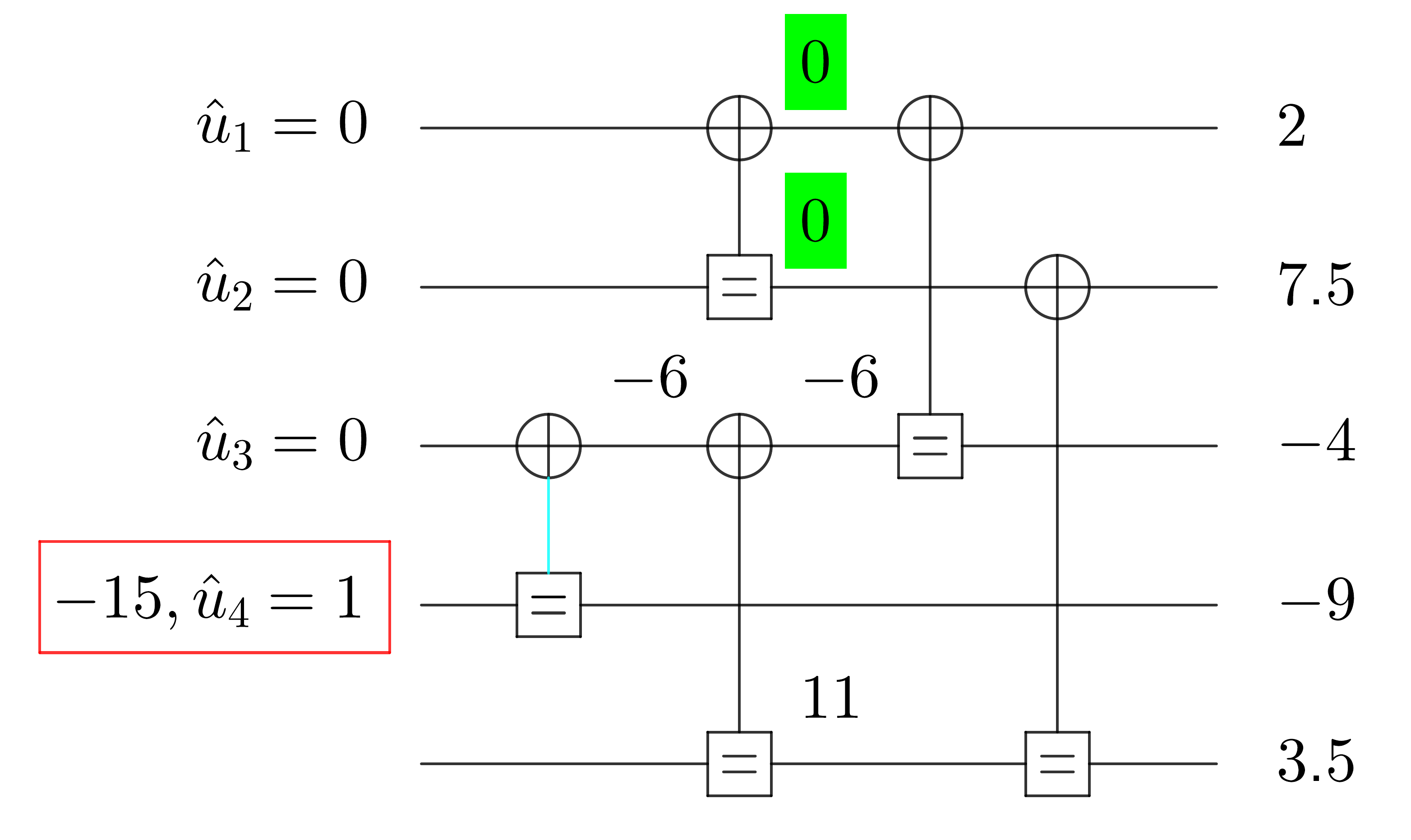}
    \caption*{step 7}
  \end{minipage}
  \hfill
  \begin{minipage}[b]{0.49\linewidth}
    \centering
    \includegraphics[width=\textwidth]{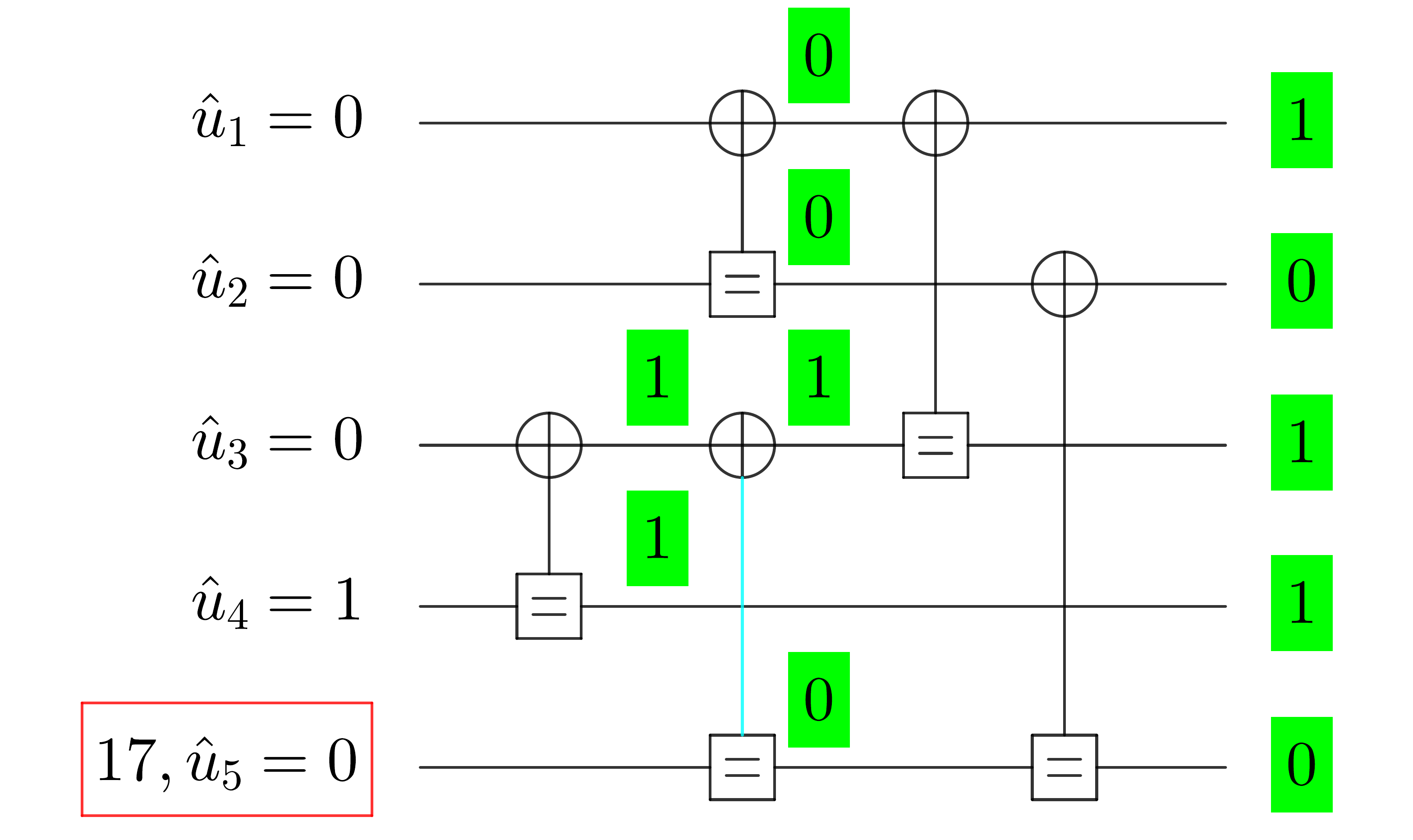}
    \caption*{step 8}
  \end{minipage}
\caption{The decoding process of length-5 stitched polar code.}
\label{fig_decode}
\end{figure}

\end{example}

The SC-list decoding algorithm \cite{Tal2015} can also be applied to stitched polar codes, as demonstrated in Algorithm \ref{alg:SCLD}.

\begin{figure}[!t]
\begin{algorithm}[H]
\caption{stitched\_polar\_SC\_list\_decoder($C, \bm{L}, S$)}
\begin{algorithmic}[1]\label{alg:SCLD}

\renewcommand{\algorithmicrequire}{\textbf{Input:}}
\renewcommand{\algorithmicensure}{\textbf{Output:}}
\REQUIRE the stitched polar codes $C = (a_1, b_1), \dots, (a_n, b_n)$ with the information set $\MI$, the channel LLR vector $\bm{L}$ and the list size $S$.
\ENSURE a list of $S$ decoding results and their path metrics (PM).
\STATE Initialize the number of paths $S_0\gets 0$;
\STATE Initialize the path metric $PM[i]\gets 0$ for $1\leq i\leq S$;
\FOR{$i=1$ to $N$}
  \FOR{$s=1$ to $S_0$}
  \STATE Apply Algorithm \ref{alg:DIKDS} to calculate  the $i$-$th$ bit-channel LLR $L_{i,s}$;
  \IF{$i$ is a frozen bit}
    \STATE $\hat{u}_{i,s} \gets 0$;
     \IF{$L_{i,s} < 0$}
      \STATE $PM[s] = PM[s] - L_{i,s}$; 
    \ENDIF  
  \ELSE
    \STATE Duplicate the $s$-$th$ path. The first copy remains the $s$-$th$ path with $\hat{u}_{i,s} \gets 0$, while the second copy becomes the $(s+S_0)$-$th$ path with $\hat{u}_{i,s+S_0} \gets 1$;
   \IF{$L_{i,s} < 0$}
      \STATE $PM[s] = PM[s] - L_{i,s}$; 
   \ELSE
      \STATE $PM[s+S_0] = PM[s+S_0] + L_{i,s}$; 
   \ENDIF  
 \ENDIF
\ENDFOR
\IF{$2S_0 > S$}
 \STATE Sort the paths by PM ascending order and keep the first $S$ paths;
\ENDIF 
\ENDFOR
\end{algorithmic}
\end{algorithm}
\end{figure}

\section{Left- and Right-Stitched Polar Codes}

In this section, we introduce two specialized stitched polar codes and analyse their decoding complexity and weight spectra. These methods enable a simple creation of longer stitched polar code $C$ by stitching two smaller, high-performance stitched polar codes $C'$ and $C''$ at either left or right side, as visually represented in Fig. \ref{fig_LR}. 

\begin{figure}[!t]
\centering
  \begin{subfigure}[b]{0.4\textwidth}
  \centering
    \includegraphics[width=0.6\textwidth, trim = 300 220 340 190, clip]{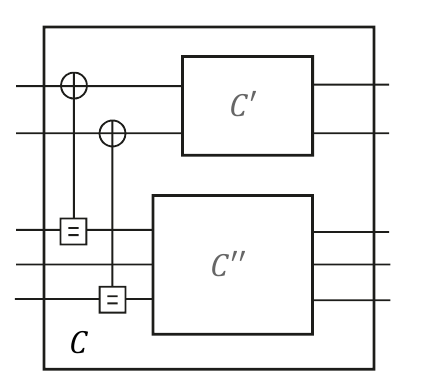}
    \caption{}
    \label{fig_LR_LEFT}
  \end{subfigure} 
  \begin{subfigure}[b]{0.4\textwidth}
  \centering
    \includegraphics[width=0.6\textwidth, trim = 300 220 340 190, clip]{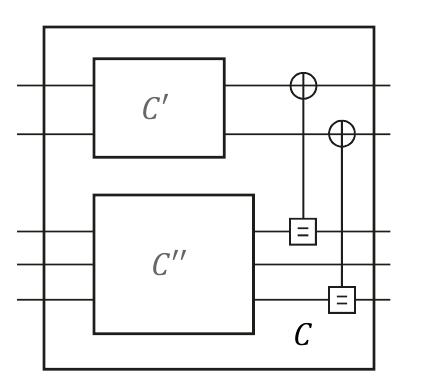}
    \caption{}
    \label{fig_LR_RIGHT}
  \end{subfigure} 
    \caption{Recursive construction of (a) left- and (b) right-stitched polar codes.}
  \label{fig_LR}
\end{figure}

Firstly, we introduce the notations used in this section. We always suppose that the channels of $C'$ are stitched to enhance the channels of $C''$.  Denote $C' =  (a'_1, b'_1), \dots, (a'_s, b'_s)$ and $C'' =  (a''_1, b''_1), \dots, (a''_t, b''_t)$. For a mapping $\pi$, $\pi(C') =  (\pi(a'_1), \pi(b'_1)), \dots, (\pi(a'_s), \pi(b'_s))$. The code lengths of $C, C'$ and $C''$ are denoted as $N, N'$ and $N''$, respectively. The generator matrices are $\bm{G}, \bm{G}'$ and $\bm{G''}$, where $\bm{g}^{(i)}, \bm{g}'^{(i)}$ and $\bm{g}''^{(i)}$ denote the $i$-$th$ row of $\bm{G}, \bm{G}'$ and $\bm{G}''$ respectively. Let $\bm{e}^{(i)}\in \FF_2^N$ be the vector with 1 at position $i$ and 0 otherwise, then $\bm{g}^{(i)} =  \bm{e}^{(i)}\bm{G}$. 

\subsubsection{Left-Stitched Polar Codes}

Assume $N'\leq N''$. The left-stitched polar code is shown in Fig. \ref{fig_LR_LEFT}. The $i$-$th$ channel of $C'$ is stitched to strengthen the $\gamma_i$-$th$ channel of $C''$ on the left with $\gamma_1<\gamma_2<\dots<\gamma_{N'}$. Denote $\Gamma=\{\gamma_1,\gamma_2,\dots,\gamma_{N'}\}\subseteq [N'']$. 

During SC decoding (Algorithm \ref{alg:DIKDS}), after the hard decision for the $i$-th bit in $C'$, the $\gamma_i$-th bit in $C''$ must subsequently be decoded. To ensure the SC decoding adheres to the natural order, we define reordering mappings $\pi' : [N'] \to [N]$ with $\pi'(i) = \gamma_i + i - 1$ and $\pi'' : [N''] \to [N]$ with $\pi''(i) = i - |\{j\in [N']\mid \pi'(j) < i\}|\}$. Then the bits from $C'$ occupy positions $\pi'(1)$ to $\pi'(N)$, while the bits from $C''$ occupy $\pi''(1)$ to $\pi''(N'')$. 

\begin{definition}
The left-stitched polar code stitched with $C'$ and $C''$ at position set $\Gamma$ on the left is $C = h_L(C', C'', \Gamma) = (\pi'(1), \pi'(1)+1), \dots, (\pi'(N'), \pi'(N')+1), \pi'(C'), \pi''(C'')$.
\end{definition}

The generator matrix of left-stitched polar code is provided by the following lemma.

\begin{lemma}\label{lemma_l1}
Let $\bm{u}$ be the message vector of $C = h_L(C', C'', \Gamma)$.  Denote $\Gamma_1=\{ \pi'(i) \mid i\in [N']\}$ and $\Gamma_2=\{\pi'(i)+1\mid i\in [N']\}$. The codewords in $C$ are obtained as $(\bm{u}_{\Gamma_1} \oplus  \bm{u}_{\Gamma_2}) \bm{G}' \bowtie_{\Gamma_1} (\bm{u}_{[N]\backslash \Gamma_1}) \bm{G}''$.

The generator matrix $\bm{G}$ of $C$ can be divided as $\bm{G}_{\Gamma_1, \Gamma_1} = \bm{G}', \bm{G}_{[N]\backslash \Gamma_1, [N]\backslash \Gamma_1} = \bm{G}'', \bm{G}_{\Gamma_1, [N]\backslash \Gamma_1} = \bm{0}, \bm{G}_{\Gamma_2, \Gamma_1} = \bm{G}',  \bm{G}_{[N]\backslash (\Gamma_1\cup \Gamma_2), \Gamma_1} = \bm{0}$.
\end{lemma}

\begin{proof}
$(\pi'(N'), \pi'(N')+1) \circ \dots \circ (\pi'(1), \pi'(1)+1) \circ \bm{u} = (\bm{u}_{\Gamma_1} \oplus  \bm{u}_{\Gamma_2}) \bowtie_{\Gamma_1} (\bm{u}_{[N]\backslash \Gamma_1}) $. Next, since $\pi'(C')$ only operates on the bits in $\Gamma_1$, and $\pi''(C'')$ on $[N]\backslash \Gamma_1$, the codeword is derived by encoding of $C'$ and $C''$.

To calculate the generator matrix, take $\bm{u} = \bm{e}^{(i)}$, then the codeword is $i$-$th$ row $\bm{g}^{(i)} = (\bm{e}^{(i)}_{\Gamma_1} \oplus  \bm{e}^{(i)}_{\Gamma_2}) \bm{G}' \bowtie_{\Gamma_1} (\bm{e}^{(i)}_{[N]\backslash \Gamma_1}) \bm{G}''$.

If $i = \gamma_t+ t-1\in \Gamma_1$, then $\bm{e}^{(i)}_{\Gamma_1} = \bm{e}^{(t)}$ and $\bm{e}^{(i)}_{[N]\backslash \Gamma_1} = \bm{0}$. Hence, $\bm{g}^{(i)} = \bm{g}'^{(t)} \bowtie_{\Gamma_1} \bm{0}$.

If $i = \gamma_t+ t\in \Gamma_2$, then $\bm{e}^{(i)}_{\Gamma_1} = \bm{0}$, $\bm{e}^{(i)}_{\Gamma_2} = \bm{e}^{(t)}$ and $\bm{e}^{(i)}_{[N]\backslash \Gamma_1} =  \bm{e}^{(\gamma_t)}$. Hence, $\bm{g}^{(i)} = \bm{g}'^{(t)} \bowtie_{\Gamma_1} \bm{g}''^{(\gamma_t)}$.

Otherwise, for $i\notin \Gamma_1\cup \Gamma_2$, $\bm{e}^{(i)}_{\Gamma_1} = \bm{0}, \bm{e}^{(i)}_{\Gamma_2} = \bm{0}$  and $\bm{e}^{(i)}_{[N]\backslash \Gamma_1} =  \bm{e}^{(j)}$, where $j = i - |\{k\in \Gamma_2\mid k<i\}|$. Hence, $\bm{g}^{(i)} = \bm{0} \bowtie_{\Gamma_1} \bm{g}''^{(j)}$.
\end{proof}

\begin{example}
Let $C$ be generated by stitching length-1 polar code $C'$ and length-4 polar code $C''$ at the position $\Gamma = \{3\}$. Now, 
$$
\bm{G}' = \bm{F}_1 = \begin{bmatrix}
1
\end{bmatrix};
\bm{G}'' = \bm{F}_4 =  \begin{bmatrix}
1 & 0 & 0 & 0\\
1 & 1 & 0 & 0\\
1 & 0 & 1 & 0\\
1 & 1 & 1 & 1\\
\end{bmatrix}.
$$

Following Lemma \ref{lemma_l1}, $\Gamma_1 = \{3\}, \Gamma_2 = \{4\}$. Then $\bm{G}$ can be divided as $\bm{G}_{3, 3} = 1, \bm{G}_{\{1,2,4,5\}, \{1,2,4,5\}} = \bm{F}_4, \bm{G}_{3, \{1,2,4,5\}} = \bm{0}, \bm{G}_{4, 3} = 1,  \bm{G}_{\{1,2,5\}, 3} = \bm{0}$. Therefore,

$$
\bm{G} = \begin{bmatrix}
1 & 0  & 0  & 0  & 0 \\
1 & 1  & 0  & 0  & 0 \\
0 & 0  & 1  & 0  & 0 \\
1 & 0  & 1  & 1  & 0 \\
1 & 1  & 0  & 1  & 1
\end{bmatrix}.
$$
The factor graph representation is depicted in Fig. \ref{fig_LEFT1}.

\begin{figure}[!t]
\centering
    \includegraphics[width=0.4\textwidth]{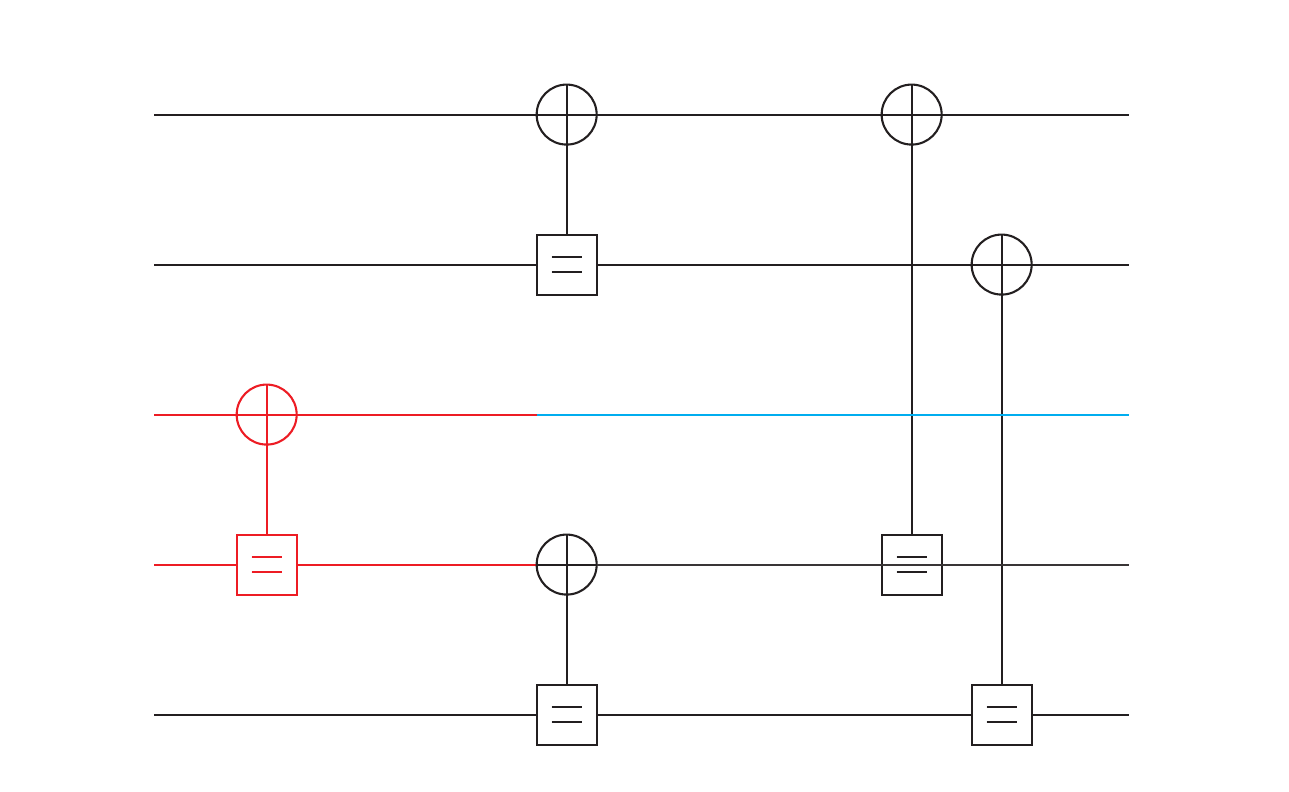}
    \caption{Left-stitched polar code $C = h_L(C', C'', \Gamma)$ coupled by $C'$ (blue) and $C''$ (black) at the position $\{3\}$.}
  \label{fig_LEFT1}
\end{figure}

\end{example}

\begin{remark}
The left-stitched polar code with $N' > N''$ is equivalent to one with $N' \leq N''$ by swapping the status of the two codes. This means the $\gamma_i$-th channel of $C''$ is stitched to strengthen the $i$-th channel of $C'$ – a configuration equivalent to $h_L(C', C'', \Gamma)$ (See Fig. \ref{fig_N5K2GG} and Fig. \ref{fig_LEFT1} as examples). Thus, we need only consider cases where $N' \leq N''$.
\end{remark}

\subsubsection{Right-Stitched Polar Codes}

The right-stitched polar code is shown in Fig. \ref{fig_LR_RIGHT}, which encompasses two distinct cases:
   
If $N'\leq N''$,  let $\Gamma' =\{\gamma'_1,\gamma'_2,\dots,\gamma'_{N'}\}\subseteq [N'']$ be the position set with $\gamma'_1<\gamma'_2<\dots<\gamma'_{N'}$. The $i$-$th$ channel of $C'$ is stitched to strengthen the $\gamma'_i$-$th$ channel of $C''$ on the right. 

If $N'> N''$, let $\Gamma'' =\{\gamma''_1,\gamma''_2,\dots,\gamma''_{N''}\}\subseteq [N']$ be the position set with $\gamma''_1<\gamma''_2<\dots<\gamma''_{N''}$. The $\gamma''_i$-$th$ channel of $C'$ is stitched to strengthen the $i$-$th$ channel of $C''$ on the right. 

Define $\tau : [N''] \to [N]$ with $\tau(i) = i + N'$. Then the bits in $C'$ reside within $[N]$, and the bits in $C''$ occupy $[N'+1,N]$ via the mapping $\tau$. 

\begin{definition}
If $N'\leq N''$, the right-stitched polar code stitched with $C'$ and $C''$ at the position set $\Gamma'$ on the right is $C = h_{R}(C', C'', \Gamma') = C', \tau(C''), (1, \tau(\gamma'_1)), \dots (N', \tau(\gamma'_{N'}))$. 

If $N'> N''$, $C = h_{R}(C', C'', \Gamma'') = C', \tau(C''), (\gamma''_1, \tau(1), \dots (\gamma''_{N'}, \tau(N''))$.
\end{definition}

The generator matrix is provided by the Lemma \ref{lemma_r1} and Lemma \ref{lemma_r2}. 

\begin{lemma}\label{lemma_r1}
Let $\bm{u}$ be the message vector of $C = h_{R}(C', C'', \Gamma')$ with $N'\leq N''$. Denote $\bm{x}' = \bm{u}_1^{N'} \bm{G}'$ and $\bm{x}'' = \bm{u}_{N'+1}^{N} \bm{G}''$.  the codewords in $C$ are obtained as $(\bm{x}' \oplus  \bm{x}''_{\Gamma'} ) \bowtie_{[N']} \bm{x}''$. 

The generator matrix of $h_{R}(C', C'', \Gamma')$ is 
$$
\bm{G} = \begin{bmatrix}
\bm{G}' & \bm{0} \\
\bm{P}  & \bm{G}''
\end{bmatrix},
$$
where $\bm{P} = \bm{G}''_{[N''], \Gamma'}$.
\end{lemma}

\begin{proof}
$\tau(C'') \circ C' \circ \bm{u} = \bm{x}' \bowtie_{[N']} \bm{x}''$. Next, $(i, \tau(\gamma'_i))$ adds the $\gamma'_i$-$th$ bit from $\bm{x}''$ to the $i$-$th$ bit in $\bm{x}'$. The final codeword then follows from this transformation.

For the generator matrix, take $\bm{u} = \bm{e}^{(i)}$. If $i \in [N']$, then $\bm{x}' = \bm{e}^{(i)} \bm{G}' = \bm{g}'^{(i)}$ and $\bm{x}'' = \bm{0}$, so $\bm{g}^{(i)} = \bm{g}'^{(i)} \bowtie_{[N']} \bm{0}$.

If $i \in [N'+1, N]$,  then $\bm{g}^{(i)} = \bm{g}''^{(i-N')}_{\Gamma'} \bowtie_{[N']} \bm{g}''^{(i-N')}$.
\end{proof}

\begin{example}\label{Ex:7}
Let $C$ be generated by stitching length-2 polar code $C'$ and length-3 shortened polar code $C''$ at the position $\Gamma' = \{1,3\}$. Here, 
$$
\bm{G}' = \begin{bmatrix}
1 & 0\\
1 & 1
\end{bmatrix};
\bm{G}'' = \begin{bmatrix}
1 & 0 & 0\\
1 & 1 & 0\\
1 & 0 & 1\\
\end{bmatrix}.
$$

Following Lemma \ref{lemma_r1}, 
$$
\bm{G} = \begin{bmatrix}
1 & 0  & 0  & 0  & 0 \\
1 & 1  & 0  & 0  & 0 \\
1 & 0  & 1  & 0  & 0 \\
1 & 0  & 1  & 1  & 0 \\
1 & 1  & 1  & 0  & 1
\end{bmatrix}.
$$
The factor graph representation is depicted in Fig. \ref{fig_RIGHT1}.

\begin{figure}[!t]
\centering
    \includegraphics[width=0.4\textwidth]{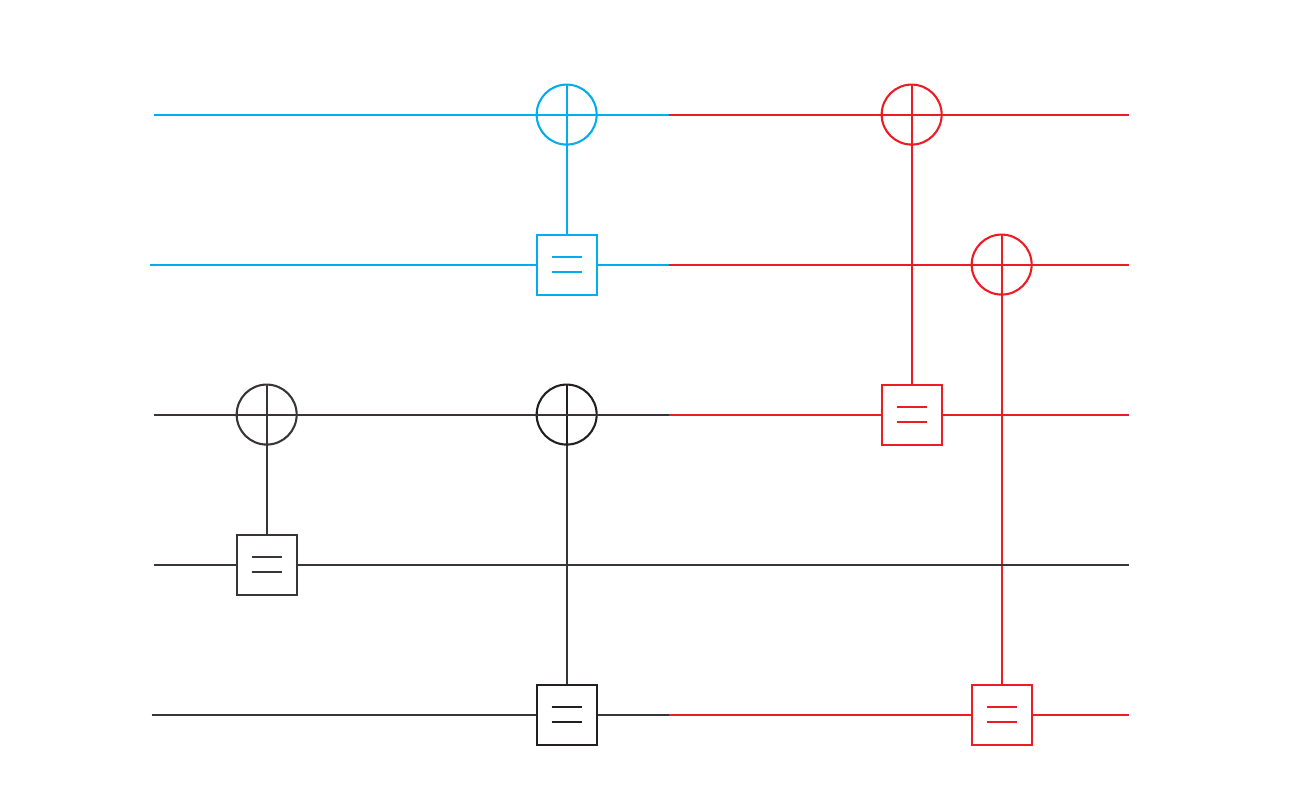}
    \caption{Right-stitched polar code $C$ stitched with $C'$ (blue) and  $C''$ (black) at the positions $\{1,3\}$.}
  \label{fig_RIGHT1}
\end{figure}
\end{example}

\begin{lemma}\label{lemma_r2}
Let $\bm{u}$ be the message vector of $C = h_{R}(C', C'', \Gamma'')$ with $N'> N''$. Denote $\bm{x}' = \bm{u}_1^{N'} \bm{G}'$ and $\bm{x}'' = \bm{u}_{N'+1}^{N} \bm{G}''$. If $N'> N''$, the codewords in $C$ are obtained as $ ((\bm{x}'_{\Gamma''} \oplus  \bm{x}'') \bowtie_{\Gamma''}  \bm{x}'_{[N']\backslash \Gamma''} ) \bowtie_{[N']}  \bm{x}''$.

The generator matrix of $h_{R}(C', C'',  \Gamma'')$ is
$$
\bm{G} = \begin{bmatrix}
\bm{G}' & \bm{0} \\
\bm{P}  & \bm{G}''
\end{bmatrix},
$$
where $\bm{P}_{[N''], \Gamma''} = \bm{G}'', \bm{P}_{[N''], [N'] \backslash \Gamma''} = \bm{0}$.
\end{lemma}

\begin{proof}
Clearly, $\tau(C'') \circ C' \circ \bm{u} = \bm{x}' \bowtie_{[N']} \bm{x}''$. Next, $(\gamma''_i, \tau(i))$ adds the $i$-$th$ bit in $\bm{x}''$ to the $\gamma''_i$-$th$ bit in $\bm{x}'$.  The final codeword then follows from this transformation.

For the generator matrix, take $\bm{u} = \bm{e}^{(i)}$. If $i \in [N']$, then $\bm{x}' = \bm{e}^{(i)} \bm{G}' = \bm{g}'^{(i)}$ and $\bm{x}'' = \bm{0}$, so $\bm{g}^{(i)} = \bm{g}'^{(i)} \bowtie_{[N']} \bm{0}$.

If $i \in [N'+1, N]$, then $\bm{g}^{(i)} = (\bm{g}''^{(i-N')} \bowtie_{\Gamma''}  \bm{0} ) \bowtie_{[N']} \bm{g}''^{(i-N')}$. 
\end{proof}

\begin{example}
Let $C$ be generated by stitching length-3 punctured polar code $C'$ and length-2 polar code $C''$ at the position $\Gamma'' = \{1,2\}$. Here, 
$$
\bm{G}' = \begin{bmatrix}
1 & 0 & 0\\
1 & 1 & 0\\
1 & 1 & 1\\
\end{bmatrix};
\bm{G}'' = \begin{bmatrix}
1 & 0\\
1 & 1
\end{bmatrix}.
$$

Following Lemma \ref{lemma_r2}, 
$$
\bm{G} = \begin{bmatrix}
1 & 0  & 0  & 0  & 0 \\
1 & 1  & 0  & 0  & 0 \\
1 & 1  & 1  & 0  & 0 \\
1 & 0  & 0  & 1  & 0 \\
1 & 1  & 0  & 1  & 1
\end{bmatrix}.
$$
The factor graph representation is depicted in Fig. \ref{fig_RIGHT2}.

\begin{figure}[!t]
\centering
    \includegraphics[width=0.4\textwidth]{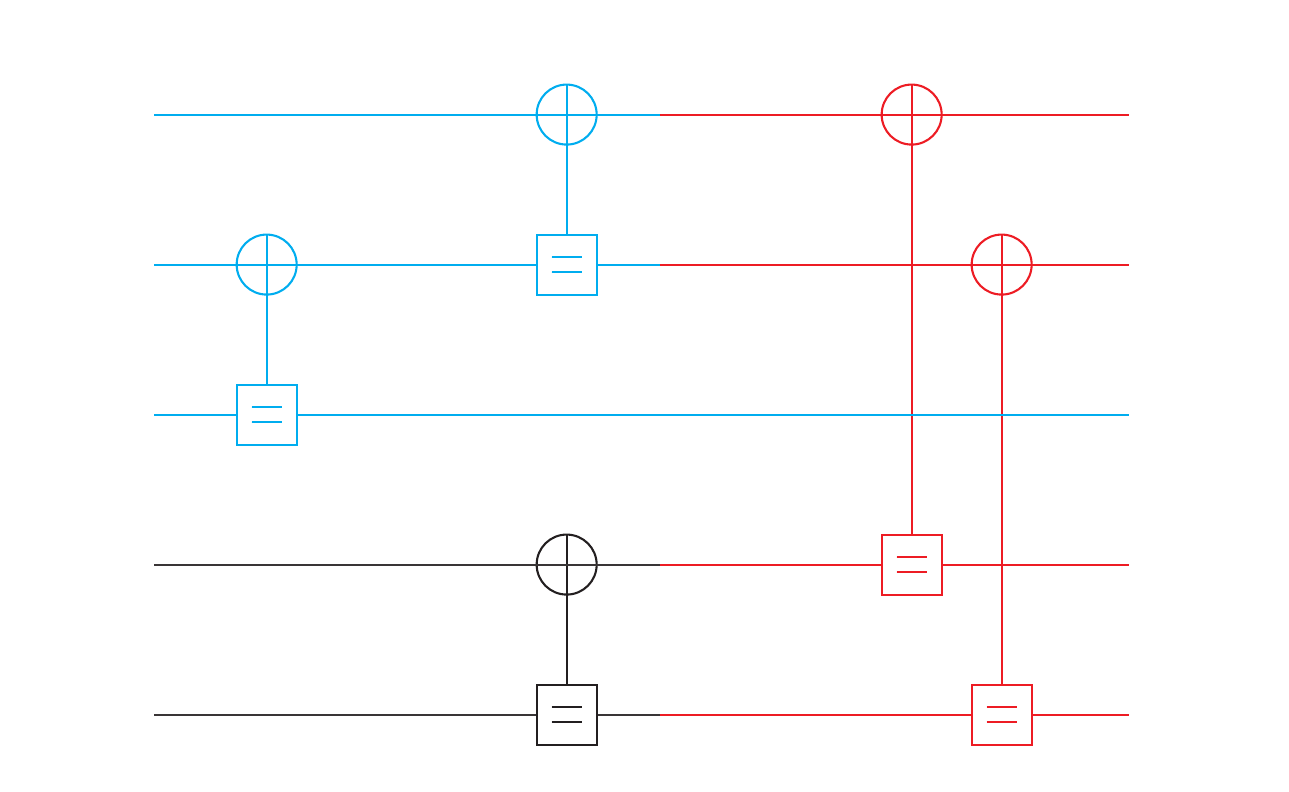}
    \caption{Right-stitched polar code $C$ stitched by $C'$ (blue) and $C''$ (black) at the positions $\{1,2\}$.}
  \label{fig_RIGHT2}
\end{figure}

\end{example}

\begin{remark}
When $N' = N''$, $C$ is the Plotkin sum \cite{Plotkin1960} of $C'$ and $C''$, that is,
$$
\bm{G} = \begin{bmatrix}
\bm{G}' & \bm{0} \\
\bm{G}''  & \bm{G}''
\end{bmatrix},
$$
\end{remark}

As shown, a length-$N$ regular polar code can be viewed as the stitching of two length-$N/2$ regular polar codes at positions $\{1,\dots,N/2\}$, either on the left or right side. Therefore, left- and right-stitched polar codes can be considered as a generalization of regular polar code.

\subsection{Encoding and Decoding Complexity for Left- and Right-Stitched Polar Codes}

We can utilize Algorithm \ref{alg:DIKDS} to encode and decode the left- and right-stitched polar codes. This involves combining the encoders and decoders for $C'$ and $C''$, along with an extra step at the stitched layer. Our analysis shows that the complexity of encoding and SC decoding for left- and right-stitched polar codes scales proportionally to the number of polarization steps.

\begin{proposition}
Let $s(N)$ denote the maximum number of $2\times 2$ transformations for a stitched polar code of length $N$. Then, $s(N)\leq  \frac{N}{2} \log N$. 
\end{proposition}

\begin{proof}
We prove this proposition by induction. Assume the code is constructed by combining two stitched polar code with length $N'$ and $N''$, regardless of whether the combination is performed  on the left or right. Then $s(N) \leq s(N')+s(N'')+\min(N',N'')$, where $s(N')$ and $s(N'')$ represent the maximum number of transformations in two component codes, respectively, and $\min(N', N'')$ accounts for the coupling between the two codes.

By induction hypothesis, we have
$$
s(N) \leq  \frac{N'}{2} \log N'+ \frac{N-N'}{2} \log (N-N')+\min(N',N-N').
$$ 
The expression is maximized when $N' = N/2$. Substituting $N' = N/2$ gives 
$$
s(N)\leq 2\cdot \frac{N}{4} \log \frac{N}{2}+\frac{N}{2} = \frac{N}{2}\log N.
$$
\end{proof}

Therefore, the complexity of encoding and SC decoding for stitched polar codes is $O(N \log N)$, which is the same as that of regular polar codes. 

\subsection{Weight Spectra of Left- and Right-Stitched Polar Codes}

In this subsection, we study the coset spectrum of left- and right-stitched polar codes. The coset spectrum $\bm{D} = (D_i)_{i=1}^N$ of $\bm{G}$ is defined as $D_i = d(\bm{g}^{(i)}, \text{span}\{\bm{g}^{(i+1)},\dots,\bm{g}^{(N)}\})$. The coset spectra of $C'$ and $C''$ are denoted as $\bm{D}'$ and $\bm{D}''$, respectively. The coset spectra of stitched polar codes are generally better than those of regular polar codes for designed code lengths and rates. This reflects that stitched polar codes possess a superior minimum distance and improved SC decoding performance \cite{Polyanskaya2020}.

\begin{example}
Consider the stitched polar code from Example \ref{Ex:7}, which has a coset spectrum $(1,2,1,3,3)$. For dimension $K=2$, the minimum distance is 3. This value is superior to the minimum distance of 2 found in the BRS polar code (with a spectrum of $(1,2,2,2,4)$) and the QUP polar code (with a spectrum of $(1,1,2,2,5)$).
\end{example}

Lemma \ref{lemma_l2} provides the coset spectrum of left-stitched polar codes. We again denote $\Gamma_1=\{\gamma_i + i - 1 \mid i\in [N']\}$. 

\begin{lemma}\label{lemma_l2}
For left-stitched polar code $C = h_{L}(C', C'', \Gamma)$, $D_{\gamma_t+ t-1} = \min(D'_t, D''_{\gamma_t})$, $D_{\gamma_t+ t} = D'_t + D''_{\gamma_t}$, and for other $i$, $D_i = D''_j$ with $j = i - |\{t\in \Gamma_1\mid t <i\}|$.
\end{lemma}

\begin{proof}
Denote $\mathcal{W}^{(i)}$ to be the linear space span$\{\bm{g}^{(i+1)},\dots,\bm{g}^{(N)}\}$. $\mathcal{W}'^{(i)}$ and $\mathcal{W}''^{(i)}$ are defined similarly. Note that $\mathcal{W}^{(i)} = \mathcal{W}'^{(i-j)} \bowtie_{\Gamma_1} \mathcal{W}''^{(j)}$ for all $1\leq i\leq N$ and $j = i - |\{t\in \Gamma_1\mid t <i\}|$.

1) $i = \gamma_t + t - 1$ for some $1\leq t\leq N'$. From Lemma \ref{lemma_l1}, $\bm{g}^{(\gamma_t+ t-1)} = \bm{g}'^{(t)} \bowtie_{\Gamma_1} \bm{0}$, $\bm{g}^{(\gamma_t + t - 1)}\oplus \bm{g}^{(\gamma_t + t)} = \bm{0}  \bowtie_{\Gamma_1} \bm{g}''^{(\gamma_t)}$ and $\mathcal{W}^{(\gamma_t+ t)} = \mathcal{W}'^{(t)} \bowtie_{\Gamma_1} \mathcal{W}''^{(\gamma_t)}$. Thus, $d(\bm{g}^{(i)}, \mathcal{W}^{(i)}) = \min \{d(\bm{g}^{(\gamma_t + t - 1)}, \mathcal{W}^{(\gamma_t + t)}),  d(\bm{g}^{(\gamma_t + t - 1)}\oplus \bm{g}^{(\gamma_t + t)}, \mathcal{W}^{(\gamma_t + t)})\} = \min \{d(\bm{g}'^{(t)} \bowtie_{\Gamma_1} \bm{0}, \mathcal{W}'^{(t)} \bowtie_{\Gamma_1} \mathcal{W}''^{(\gamma_t)} ),  d(\bm{0}  \bowtie_{\Gamma_1} \bm{g}''^{(\gamma_t)}, \mathcal{W}'^{(t)} \bowtie_{\Gamma_1} \mathcal{W}''^{(\gamma_t)} )\} = \min(D'_t, D''_{\gamma_t})$.

2) $i = \gamma_t + t$ for some $1\leq t\leq N'$.  From Lemma \ref{lemma_l1}, $d(\bm{g}^{(i)}, \mathcal{W}^{(i)}) =  d(\bm{g}'^{(t)},  \mathcal{W}'^{(t)}) + d(\bm{g}''^{(\gamma_t)}, \mathcal{W}''^{(\gamma_t)}) = D'_t + D''_{\gamma_t}$.

3) Otherwise, $\bm{g}^{(i)} = \bm{0} \bowtie_{\Gamma_1} \bm{g}''^{(j)} $, then $d(\bm{g}^{(i)}, \mathcal{W}^{(i)}) =   d(\bm{0} \bowtie_{\Gamma_1} \bm{g}''^{(j)}, \mathcal{W}'^{(i-j)} \bowtie_{\Gamma_1} \mathcal{W}''^{(j)}) = D''_{j}$.
\end{proof}

Lemma \ref{lemma_r3} provides a bound for the coset spectrum of right-stitched polar codes when $N'\leq N''$.

\begin{lemma}\label{lemma_r3}
Denote $C = h_{R}(C', C'', \Gamma')$ and $N'\leq N''$. Then $D_i = D'_i$ for $i \leq N'$ and  $D_i \leq \min(D''_{i-N'}+ N', 2D''_{i-N'})$ for $i > N'$.
\end{lemma}

\begin{proof}
Denote $\bm{w}$ to be a vector in the span$\{\bm{g}^{(i+1)},\dots,\bm{g}^{(N)}\}$. 

1) When $i\leq N'$, denote $\bm{w} = \bm{w}' + \bm{w}''$ where $\bm{w}' \in$ span$\{\bm{g}^{(i+1)},\dots,\bm{g}^{(N')}\}$ and $\bm{w}'' \in$ span$\{\bm{g}^{(N'+1)},\dots,\bm{g}^{(N)}\}$.

Now, $d(\bm{g}^{(i)}, \bm{w}) =  d(\bm{g}^{(i)}\oplus \bm{w}', \bm{w}'') = d(\bm{g}^{(i)}\oplus\bm{w}', \bm{w}''_{[N']} \bowtie_{[N']}  \bm{0}) + d(\bm{0}, \bm{0}  \bowtie_{[N']}  \bm{w}''_{[N'+1,N]}) \geq wt(\bm{g}^{(i)}\oplus\bm{w}') \geq D'_i$. The second-to-last inequality holds since $\bm{w}''_{[N']}$ is a sub-vector of $\bm{w}''_{[N'+1,N]}$. The equalities are achieved for some $\bm{w}'$ such that $d(\bm{g}^{(i)}, \bm{w}') = D'_i$ and $\bm{w}'' = \bm{0}$. Thus, we have $D_i = D'_i$.

2) When $i > N'$, then $d(\bm{g}^{(i)}, \bm{w}) =  d(\bm{g}^{(i)}_{[N']}, \bm{w}_{[N']}) + d(\bm{g}^{(i)}_{[N'+1,N]}, \bm{w}_{[N'+1,N]})$. Since the rows of $\bm{P}$ in Lemma \ref{lemma_r1} are part of the rows of $\bm{G}''$, we have $d(\bm{g}^{(i)}_{[N']}, \bm{w}_{[N']})  \leq d(\bm{g}^{(i)}_{[N'+1,N]}, \bm{w}_{[N'+1,N]})$. Therefore, $d(\bm{g}^{(i)}, \bm{w}) \leq 2(\bm{g}^{(i)}_{[N'+1,N]}, \bm{w}_{[N'+1,N]})$. Since $\bm{g}^{(i)}_{[N'+1,N]} = \bm{g}''^{(i-N)}$ and $\bm{w}_{[N'+1,N]}\in$ span$\{\bm{g}''^{(N-i+1)},\dots,\bm{g}''^{(N'')}\}$, we have $D_i \leq 2D''_{i-N'}$. 

On the other hand, suppose there exists $\bm{w}$ such that $D''_{i-N'} = d(\bm{g}^{(i)}_{[N'+1,N]}, \bm{w}_{[N'+1,N]})$. Then $d(\bm{g}^{(i)}, \bm{w}) =  d(\bm{g}^{(i)}_{[N']}, \bm{w}_{[N']}) + D''_{i-N'}\leq N' + D''_{i-N'}$. 
\end{proof}

Lemma \ref{lemma_r4} provides the coset spectrum of right-stitched polar codes when $N'\geq N''$.

\begin{lemma}\label{lemma_r4}
Denote $C = h_{R}(C', C'', \Gamma'')$  and $N'\geq N''$. $D_i = D'_i$ for $i \leq N'$ and $D_i = 2D''_{i-N'}$ for $i > N'$.
\end{lemma}

\begin{proof}
Let $\bm{w}$ be some vector in the span$\{\bm{g}^{(i+1)},\dots,\bm{g}^{(N)}\}$.

1) When $i\leq N'$, the proof is similar with the case 1) in Lemma \ref{lemma_r3}.

2) When $i > N'$, since the rows of $\bm{P}$ in Lemma \ref{lemma_r2} replicate those of $\bm{G}''$, we have $d(\bm{g}^{(i)}_{[N']}, \bm{w}_{[N']})  = 2d(\bm{g}^{(i)}_{[N'+1,N]}, \bm{w}_{[N'+1,N]})$. Therefore,  $D_i = 2D''_{i-N'}$. 

\end{proof}

\section{Construction of Stitched Polar Codes} \label{sec:GD}

In the previous section, we introduced the general definition of stitched polar codes. In this section, we present construction methods for stitched polar codes that exhibit superior performance. 

\subsection{Recursive Construction Algorithms of Right-Stitched Polar Codes}

To construct a right-stitched polar code $C = h_R(C', C'', \Gamma)$, we must select the upper branch $C'$, the lower branch $C''$, and the stitching positions $\Gamma$. The construction is recursive, leveraging previously built shorter codes for $C'$ and $C''$. In practice, we have observed that the specific choice of stitching positions has a negligible impact on overall performance. Therefore, to simplify the construction process and reduce complexity, we always assume $\Gamma$ consists of the first $\min\{N',N''\}$ bits.

Specifically, when constructing a right-stitched polar code $C_{N, K}$ of length $N$ and dimension $K$, we iterate over the code with all the length $N'$ and dimension $K'$ for the upper branch, while the corresponding parameters for the lower branch are $N - N'$ and $K-K'$. Then we retrieve pre-constructed codes $C_{N', K'}$ and $C_{N - N', K - K'}$ and stitch these two codes to form a new code. This process is repeated across multiple $N,K$. We apply DE/GA to compute channel reliability and error probabilities for all generated codes, ultimately selecting the optimal code $C_{N, K}$. Algorithm \ref{alg:CON1} formalizes this procedure.

\begin{figure}[!t]
\begin{algorithm}[H]
\caption{stitched\_polar\_code\_recursive\_partition($M$)}
\begin{algorithmic}[1]\label{alg:CON1}

\renewcommand{\algorithmicrequire}{\textbf{Input:}}
\renewcommand{\algorithmicensure}{\textbf{Output:}}
\REQUIRE Maximum code length $M$.
\ENSURE Code family $\mathcal{C}=\{C_{N,K}\}_{1\leq N\leq M, 0\leq K\leq N}$.
\STATE Initialize $C_{1,0}\gets [1]$ and $C_{1,1}\gets [1]$;
\FOR{$N=2$ to $M$}
\FOR{$K=0$ to $N$}
\STATE $e_{\min} \gets 1$;
\FOR{$N'=1$ to $N-1$}
\FOR{$K'=0$ to $N'$}
\STATE $C\gets h_R(C_{N',K'}, C_{N-N',K-K'}, \Gamma)$ with $\Gamma = [\min\{N', N-N'\}]$ ;
\STATE Calculate error probability $e$ of the $C$ through density evolution;
\IF{$e < e_{\min}$}
\STATE $C_{N,K}\gets C$;
\STATE $e_{\min}\gets e$;
\ENDIF
\ENDFOR
\ENDFOR
\ENDFOR
\ENDFOR
\end{algorithmic}
\end{algorithm}
\end{figure}

The recursive construction process generates all codes with lengths $N$ up to the maximum length $M$ and dimensions $0\leq K\leq N$. For each length-$N$ code being generated, there are $N^2$ construction attempts considered, with each attempt requiring DE with complexity $O(N\log N)$. This results in  $O(N^3\log N)$ computational complexity per code. Consequently, the overall complexity of the algorithm is $O(M^5\log M)$. The construction is an offline process, and the resulting codes can be saved for later use.

\begin{remark}
Note that left-stitched polar codes can also be built recursively. However, a key challenge is that well-performing short code constructions from one recursive iteration may not remain optimal for subsequent iterations. In order to obtain superior long codes, a large number of code configurations must be saved during the offline code construction iterations. Although the encoding and decoding complexity remain low at $O(N \log N)$, the construction complexity is considerably higher than the proposed right-stitched polar codes.
\end{remark}

\subsection{Partially Stitched Polar Codes}

Stitched polar codes can also be combined with regular polar codes. In this subsection, we propose an integration of stitched polar codes with regular construction. Initially, the family of stitched polar codes $\mathcal{C}$ with lengths up to $M=2^s$ is generated using Algorithm \ref{alg:CON1}. For code length $N>M$, we employ BRS regular polar codes with mother code length $N_0 = 2^{\lceil \log N \rceil}$. After $\log N_0 - s$ stages of polarization, the BRS polar code is partitioned into $\frac{N_0}{2^s}$ outer code sub-blocks, each with length $\frac{M}{2}\leq N_i\leq M$ and number of information bits $K_i$. Performance improvements are achieved by substituting all outer codes with $C_{N_i,K_i}\in \mathcal{C}$. The selection of BRS is motivated by the balanced lengths and polarization stages in each sub-block. While our initial design is based on BRS, other rate-matching techniques may also yield good performance. This hybrid structure is termed the \emph{$M$-partially stitched polar code}.

\subsubsection{Encoding and Decoding}

The encoding and SC decoding for $M$-partially stitched polar codes can be directly adopted from Section \ref{sec:GPC}, combined with regular encoding and SC decoding algorithm. The encoding process is  illustrated in Fig. \ref{fig_encode2}. First, codewords for all stitched sub-blocks are generated. Following this, the encoding results are zero-padded according to the rate-matching pattern.  The subsequent encoding procedure proceeds identically to the regular case.

\begin{figure*}[!t]
\centering
\includegraphics[width=0.95\textwidth, trim = 0 80 0 150, clip]{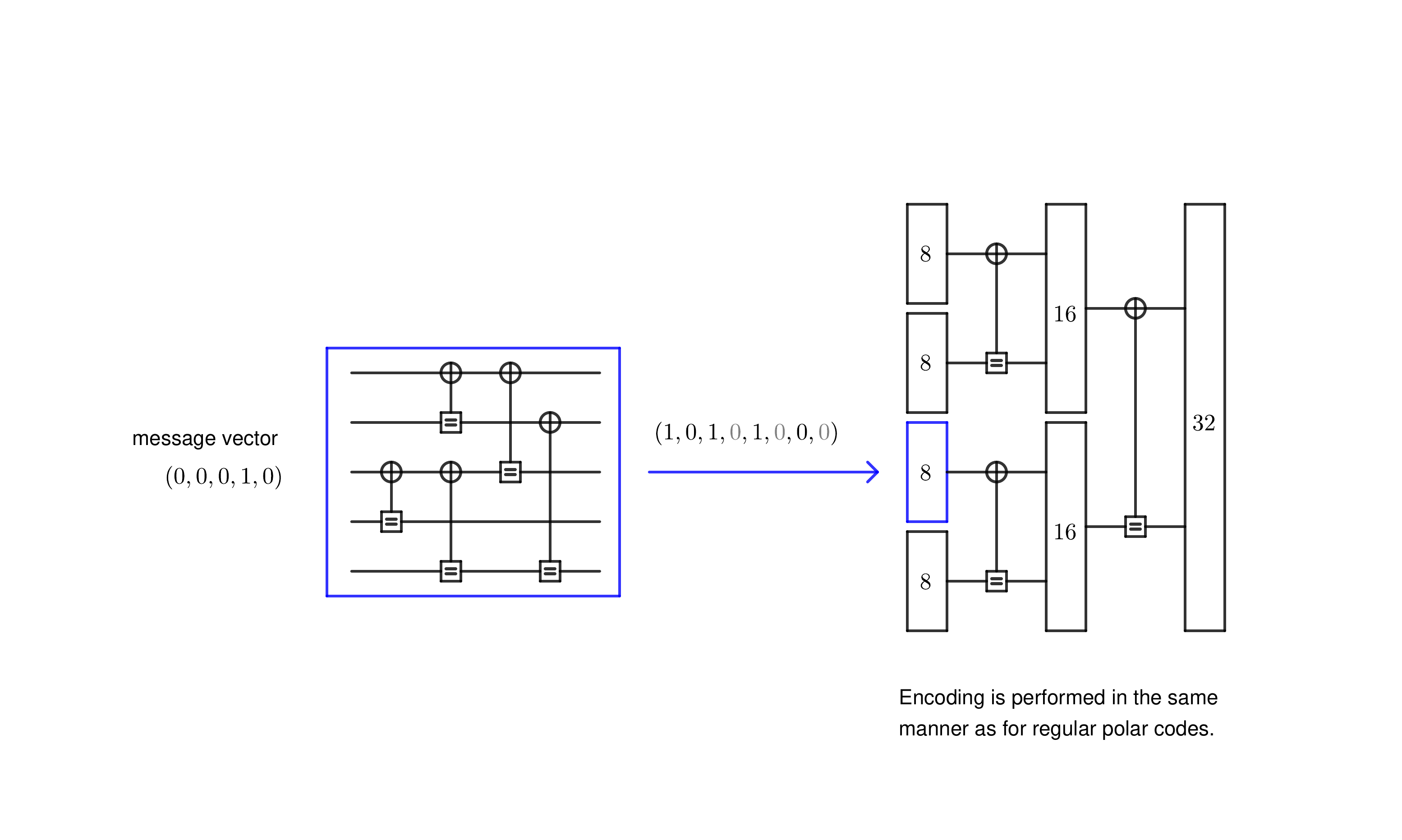}
\caption{The encoding process of 8-partially stitched polar code.}
\label{fig_encode2}
\end{figure*}

The SC decoding process is illustrated in Fig. \ref{fig_decode2}. After the sequential f/g-operations in regular polar codes, the decoder produces input LLRs for the stitched sub-blocks. The specific SC decoder according to the stitched sub-blocks is then applied. Finally, hard decisions are generated in the same manner as in the regular scenario.

\begin{figure*}[!t]
\centering
\includegraphics[width=0.95\textwidth, trim = 0 40 0 40, clip]{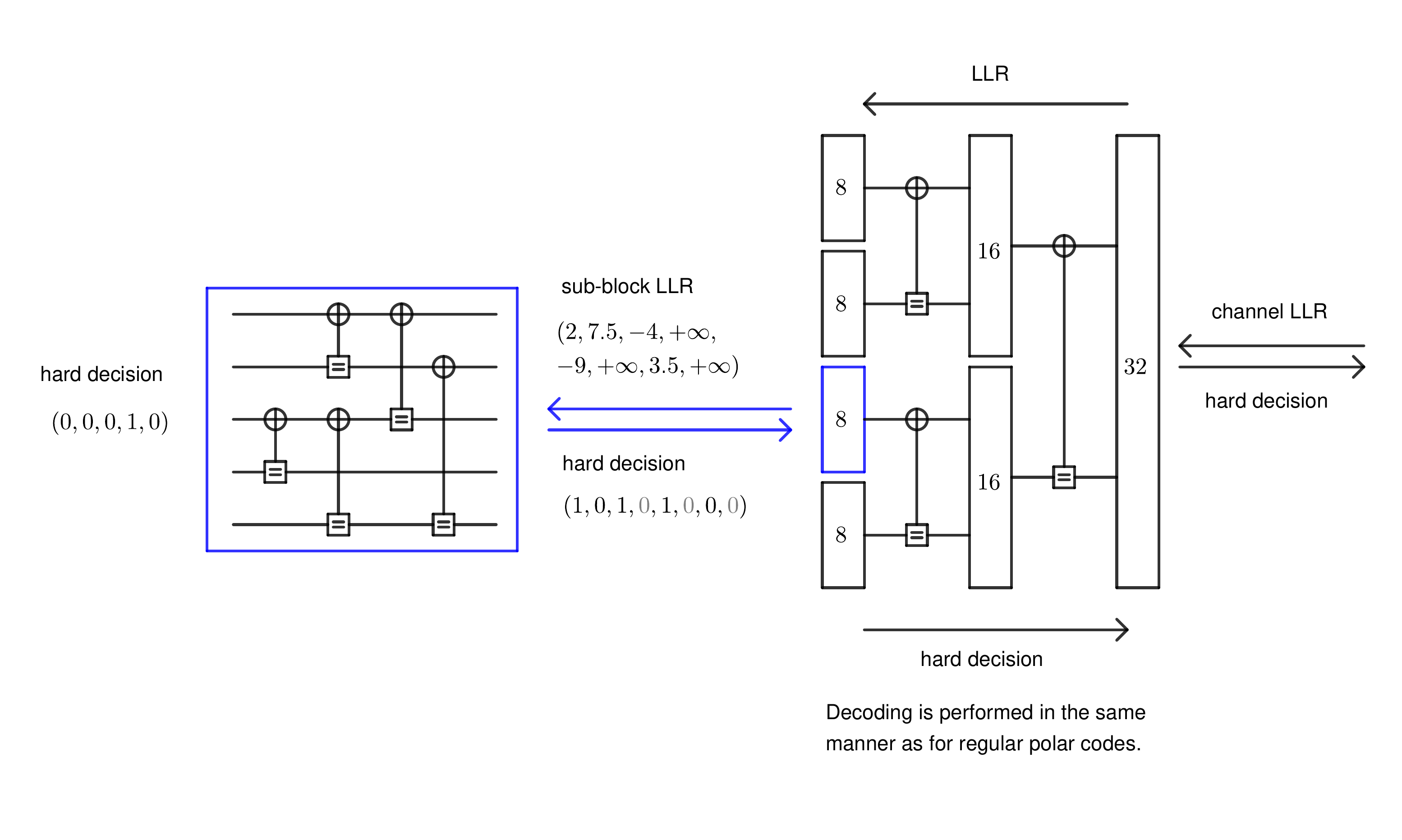}
\caption{The decoding process of 8-partially stitched polar code.}
\label{fig_decode2}
\end{figure*}

\subsubsection{Information Bits Allocation}

In partially stitched polar codes, we stitch multiple stitched polar sub-blocks to create a longer code. The code lengths and stitched locations of sub-blocks are determined by the original BRS polar code. However, it is vital to allocate information bits to the stitched sub-blocks. This task is non-trivial due to the many possible configurations of outer codes. A greedy information bit allocation between sub-blocks based on performance evaluation of these outer codes is outlined in Algorithm 4. Subsequently, the generator matrix and information set are derived from $\mathcal{C}$ based on the allocated information bits per sub-block.

In Algorithm \ref{alg:CON2}, we sequentially allocate information bits, one at a time, to the $\frac{N_0}{2^s}$ sub-blocks. In each step, an information bit is assigned to one of the sub-blocks to maximize the overall successful decoding probability. To this end, the sub-blocks' successful decoding probabilities for all $(N_i,K_i)$ are pre-computed as follows. We first compute the reliability via DE for $\log N_0-s$ stages, stopping before the stitched sub-blocks (line 3), and then calculate the successful decoding probability $p(N_i,K_i)$ fro all $K_i$ from 1 to $N_i$ (lines 6-8). In particular, we allocate an additional information bit to the sub-block that maximizes $p(N_i,K_i+1)/p(N_i,K_i)$, namely the relative successful decoding probability (lines 12-13).

\begin{figure}[!t]
\begin{algorithm}[H]
\caption{stitched\_polar\_code\_rate\_allocation($N, K, \mathcal{C}$)}
\begin{algorithmic}[1]\label{alg:CON2}

\renewcommand{\algorithmicrequire}{\textbf{Input:}}
\renewcommand{\algorithmicensure}{\textbf{Output:}}
\REQUIRE code length $N$ and dimension $K$, a code family $\mathcal{C}=\{C_{N',K'}\}_{1\leq N'\leq 2^s, 0\leq K'\leq N'}$, physical channel capacity $C_{cap}$. 
\ENSURE The number of information bits $K_i$ for each sub-block.
\STATE $N_0 \gets 2^{\lceil \log N \rceil}$;
\STATE $N_i$ is the length of the $i$-$th$ sub-block;
\STATE Compute the reliability via DE for $\log N_0 - s$ stages;
\FOR{$i=1$ to $\frac{N_0}{2^s}$}
\STATE $p(N_i,0) \gets 1$;
\FOR{$K_i=1$ to $N_i$}
\STATE Calculate the probability of successful SC decoding $p(N_i,K_i)$ for code $C_{N_i,K_i}$ via DE based on the input reliability of the $i$-$th$ sub-block;
\ENDFOR
\ENDFOR
\STATE $K_i\gets 0$ for all $1\leq i\leq \frac{N_0}{2^s}$;
\FOR{$i=1$ to $K$}
 \STATE $j \gets \arg\max_{1\leq i\leq  \frac{N_0}{2^s}} p(N_i,K_i+1)/p(N_i,K_i)$;
 \STATE $K_j \gets K_j + 1$;
\ENDFOR

\end{algorithmic}
\end{algorithm}
\end{figure}

The time complexity analysis of Algorithm \ref{alg:CON2} proceeds as follows:

1. Line 3: $O(\frac{N_0}{M} \log \frac{N_0}{M})$. 

2. Lines 6-8:  For the $i$-$th$ sub-block, the complexity of DE is $O(N_i\log N_i)$ with $N_i$ times per sub-block. Given there are $\frac{N_0}{M}$ sub-blocks, the total computation time is $O(N_0 M \log M)$. 

3.  Lines 11-13:  Sorting requires $O(\frac{N_0}{M}\log \frac{N_0}{M})$ operations when $i=1$. For $i=2$ to $K$, only one number changes per iteration, necessitating a position search with complexity $O( \log\frac{N_0}{M})$ per update. Thus, the complexity becomes $O((K+\frac{N_0}{M}) \log\frac{N_0}{M})$.

In conclusion, the overall complexity of Algorithm \ref{alg:CON2} is $O(N_0 M \log M + N_0 \log\frac{N_0}{M})$.

\section{Scaling Laws for Stitched Polar Codes} 

In this section, we dive into the theoretical analysis of the performance improvements of stitched polar codes. As previously demonstrated, rate-matched regular polar codes experience performance degradation when the code length exceeds powers of 2. Stitched polar codes, however, show consistent performance across all lengths. We will explore this notable behaviour by examining the number of un-polarized bit-channels. 

Due to the inherent diversity in the construction of stitched polar codes, we restrict our discussion to $2^s$-partially stitched polar codes, where $s$ is a positive integer. In these configurations, the stitched sub-blocks are precisely engineered to minimize the number of un-polarized bit-channels.

For a BEC with capacity $I$, denote $\alpha(N,I,[a, b]) = \frac{1}{N}|\{1\leq i\leq N\mid I_N^{(i)} \in [a, b])\}|$. Here, $I_N^{(i)}$ represents the capacity of the $i$-$th$ bit-channel in either a regular polar code or a stitched polar code with length $N$ over a BEC with capacity $I$. We assume that the scaling law (Assumption \ref{assum_scaling}) holds for all these codes, that is, $\alpha(N,I,[a, b]) = (c(I,a,b)+o(1))N^{-\frac{1}{\mu}}$, where $c(I,[a, b])$ and $\mu$ may vary depending on the specific code. For convenience, we will omit the $o(1)$ term in the subsequent discussion.

\begin{example}
Fig. \ref{fig_4lines} illustrates the polarization speed $\alpha(N,0.5,[0.01,0.99])$ for four different types of codes. The horizontal axis is the logarithm of the code length $m$, while the vertical axis represents the logarithm of $\alpha$. According to the scaling law, $\log r$ can be approximated as $\log c - \frac{1}{\mu} m$, where $c$ and $\mu$ are independent of code length. Therefore, the intercept in the figure corresponds to the constant $c$, and the slope is associated with the scaling exponent $\mu$.

In Fig. \ref{fig_4lines}, all four lines have the same slope, which indicates that their scaling exponents are equivalent. Notably, the intercept for the QUP regular polar code is larger than that for the original code with power-of-two length. This larger intercept signifies a larger constant term $c$, implying a greater number of un-polarized bit-channels and, consequently, a slower polarization speed. This phenomenon explains the performance loss when the code lengths switch at powers of two. On the other hand, the stitched polar codes-whether at lengths of $2^m$ or $\frac{33}{32}\cdot 2^m$-demonstrate nearly identical intercepts. Both of these intercepts are also lower than those for the regular polar codes. This finding strongly suggests superior polarization behavior in stitched codes.

\begin{figure}[!t]
\centering
\includegraphics[width=0.5\textwidth]{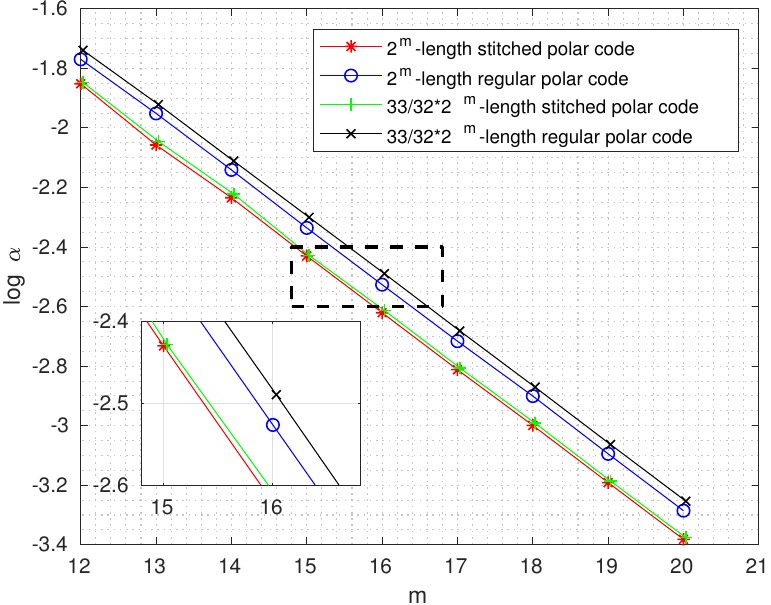}
\caption{The logarithm of the ratio of bit-channels with capacity in [0.01, 0.99] for regular versus stitched polar codes at different lengths over BEC(0.5).}

\label{fig_4lines}
\end{figure}

\end{example}

The proposed theoretical analysis will validate these simulation observations over BEC:

(1) Regular polar codes exhibit a larger constant term $c$ than stitched polar codes; 

(2) At switching points of mother code lengths, QUP and BRS polar codes show increased constant terms – with the QUP proof resting on an assumption; 

(3) Partially stitched codes maintain a stable constant term across different code lengths.

Let $N = 2^m$. For the length-$(1+2^{-t})N$ QUP and BRS polar codes, the associated  constants $c$ are denoted as $c_{QUP}^{(t)}(I, [a, b])$ and $c_{BRS}^{(t)}(I, [a, b])$, repetitively. For $2^s$-partially stitched polar codes, the constants are denoted as $c_{\text{STC},s}(I, [a, b])$ for length-$N$ and $c_{\text{STC},s}^{(t)}(I, [a, b])$ length-$(1+2^{-t})N$. If $a, b$ remain fixed throughout the proof, we abbreviate $c(I, [a, b])$ as $c(I)$. For $S = [a_1, b_1]\cup \dots\cup [a_t, b_t]$, we abbreviate $\sum_{i=1}^t c(I, [a_i,b_i])$ as $c(I,S)$.

\begin{remark}
In this section, the results are established for the BEC with capacity $I$. These results hold for all $0<a<a_0$ and $b_0<b<1$, where $a_0$ and $b_0$ are constants independent with $N$. That is, the results hold for $a$ sufficiently close to 0 and $b$ sufficiently close to 1. This captures the polarization behavior as the code length approaches infinity.
\end{remark}

Firstly, we prove that  regular polar codes and stitched polar codes share the exact same scaling exponent. Their performance gap lies in the constant term.

\begin{theorem}\label{thm_se}
The scaling exponents for original, QUP, BRS and stitched polar codes are the same. 
\end{theorem}

\begin{proof}
Denote $M_N$ to be the number of un-polarized bit-channels in regular polar codes with code length $N=2^m$ under BEC with capacity $I$, i,e., $M_N =  |\{1\leq i\leq N\mid I(W_m^{(i)}) \in [a, b])\}| = c(I) N^{\lambda}$. Let $M_{(1+2^{-t})N, \text{QUP}}$ be the number of un-polarized bit-channels in length-$(1+2^{-t})N$ QUP polar codes. Then $M_{(1+2^{-t})N, \text{QUP}} =  c(I^2) (2^{-t}N)^{\lambda} + c' N^{\lambda} = (c(I^2) 2^{-t\lambda} + c') N^{\lambda}$, where $c' N^{\lambda}$ is the number of un-polarized bit-channels in the lower branch of QUP polar codes. Thus, the scaling exponent for QUP polar codes remains unchanged. 

Let $M_{(1+2^{-t})N, \text{BRS}}$ be the number of un-polarized bit-channels in length-$(1+2^{-t})N$ BRS polar codes. Note that after $m-t$ polarization stages in a regular polar code, there are $2^{m-t}$ sub-blocks, each consisting of $2^t$ intermediate synthetic channels with identical capacity. Consider a length-$2^t$ sub-block in length-$N$ regular polar codes with input capacity $I_0$. When the code length is expanded to $(1+2^{-t})N$ as BRS pattern, a newly added channel is introduced in this sub-block. All $t+1$ intermediate channels (including the newly added one) then maintain intermediate capacities $I_0$, as illustrated in Fig. \ref{fig_BIV_PROOF}.

Now, if $I_0 < a_0 = \frac{a}{2^t+1}$ or $I_0 > b_0 = \frac{2^t+b}{2^t+1}$, all $2^t+1$ bit-channels must be polarized. The number of sub-blocks with input capacity in $[a_0,b_0]$ is equivalent to the number of bit-channels with capacity in $[a_0,b_0]$ in a length-$2^{m-t}$ regular polar codes, which is $c(I,[a_0,b_0])(2^{-t}N)^{\lambda}$. Then  $M_{N,\text{BRS}}\leq 2^t c(I,[a_0,b_0])(2^{-t}N)^{\lambda} = c(I,[a_0,b_0]) 2^{t(1-\lambda)} N^{\lambda}$.  Thus, the scaling exponent for BRS polar codes remains unchanged. The same proofs apply to partially stitched polar codes since the polarization of regular part is the same as that of BRS polar code.
 
Similarly, denote $M_{N,\text{STC}}$ to be the number of un-polarized bit-channels in length-$N$ $2^s$-partially stitched polar codes. As regular polar codes are special cases of stitched polar codes, we have $M_{N,\text{STC}} \leq M_N$. We replace outer codes with stitched polar codes for every length-$2^s$ sub-block. Suppose $a,b$ satisfy $a< \frac{1}{2^s}$ and $b>2^s a$ and $S = [2^s a, b]$. Then for $I\in S$, there is at least one un-polarized bit-channels in the sub-block. The number of such sub-blocks is $c(I,S)(2^{-s}N)^{\lambda}$, then we have $M_{N,\text{STC}}\geq c(I,S)(2^{-s}N)^{\lambda} \geq \frac{c(I,S)}{c(I)}2^{-s\lambda} M_N$. 
\end{proof}

Theorem \ref{thm_STC_p} shows that the constant term of stitched polar codes is smaller than that of regular polar codes.

\begin{theorem}\label{thm_STC_p}
$c_{\text{STC},s}(I) < c(I)$ for any $I\in (0,1)$.
\end{theorem}

\begin{proof}
In each length-$2^s$ sub-block of regular polar codes, we replace the outer codes with stitched polar codes. Then there exists a non-empty set $S\in [0,1]$, such that for $I'\in S$, there is at least one additional polarized bit-channels. The increase of polarized bit-channels is then at least $c(S,I)(2^{-t}N)^{\lambda}$. Therefore,
$$
c_{\text{STC},s}(I)N^{\lambda} \leq  c(I)N^{\lambda} - c(S,I)(2^{-s}N)^{\lambda}. 
$$
This implies $c_{\text{STC},s}(I) \leq  c(I) - c(S,I)2^{-s\lambda}$. 
\end{proof}

Theorem \ref{thm_STC} demonstrates that the constant term for stitched polar codes is stable when code length is slightly larger than power-of-two.

\begin{theorem}\label{thm_STC}
$c_{\text{STC},s'}^{(t)}(I,[a,b]) \leq c_{\text{STC},s}(I,[a,b])$ for BEC with capacity $I\in (0,1)$,$s > -\log_2(\frac{\log_2 b}{\log_2 a - \log_2 b})$ and sufficiently large $t$ together with $s'\geq t$.
\end{theorem}

\begin{proof}
Similar to the proof of Theorem \ref{thm_se}, consider a length-$(2^{t}+1)$ sub-block in length-$(1+2^{-t})N$ stitched polar codes with input capacity $I_0$. If $I_0<a$ or $I_0>b$, the newly added channel is already polarized and requires no further polarization. Thus, the number of un-polarized bit-channels does not increase in this sub-block. 

In cases where $I_0\notin [a, b]$, since $t$ is sufficiently large, there is a length-$2^s$ sub-block with input capacity $I_1$ satisfying $(a/b)^{2^{-s}} > I_1 > b$. We destruct the original polarization in this length-$2^s$ sub-block and polarize all the bit-channels in the sub-block with the new bit-channel. 

Following $2^s$ polarization operations, the capacity of the new bit-channel becomes $I_0I_1^{2^s} < a$. Since $I_1 > b$, the capacities of the other bit-channels remain  larger than $b$ after polarization. Consequently, by increasing the stitched sub-block size to at least $t$, the number of un-polarized bit-channels does not increase.
\end{proof}

\begin{figure*}[!t]
\centering
    \includegraphics[width=0.95\textwidth, trim = 0 150 0 150, clip]{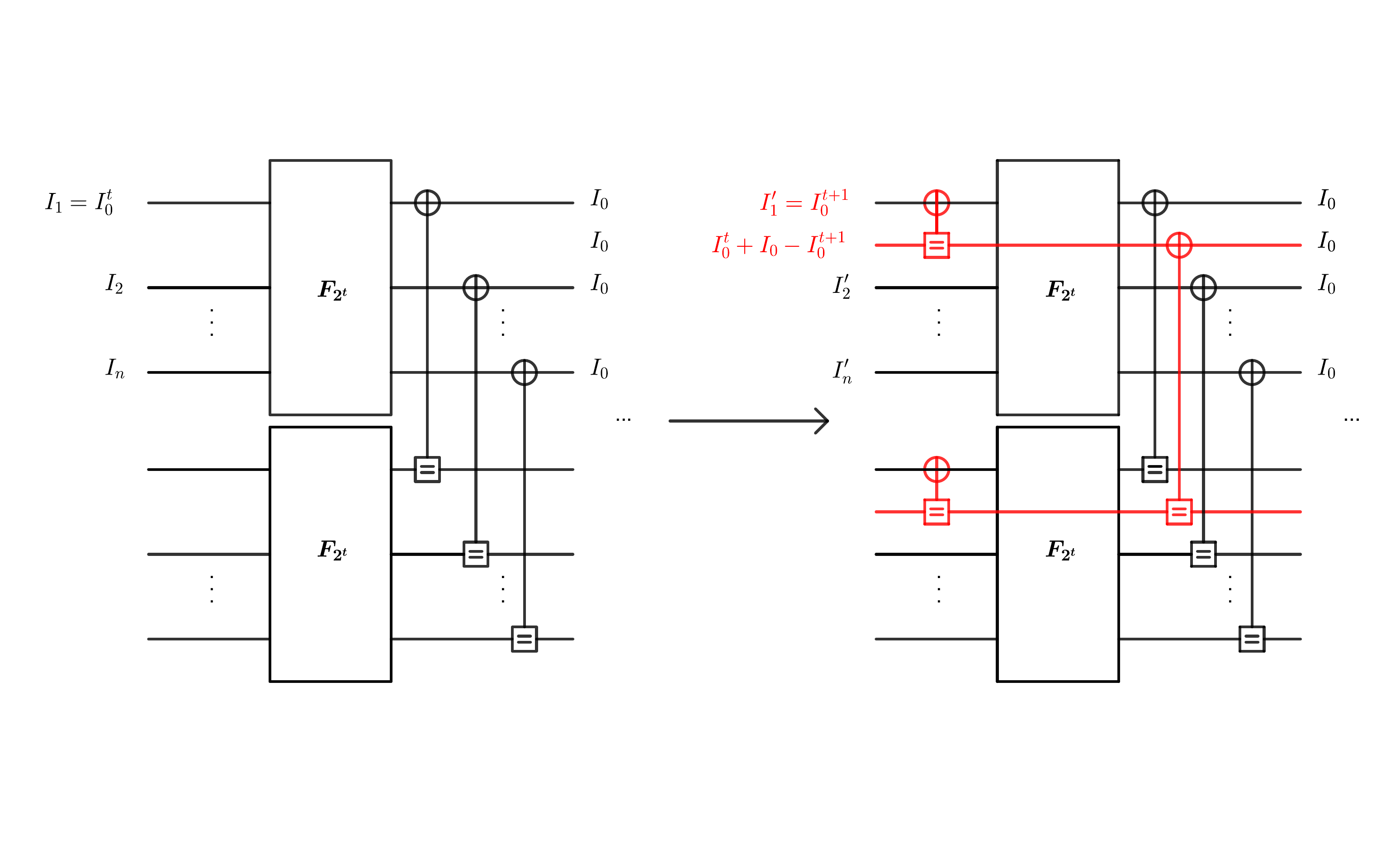}
    \caption{length-$2^t$ sub-block with capacity $I_0$ in length-$N$ regular polar codes and the same sub-block in length-$(1+2^{-t})N$ BRS polar codes.}
  \label{fig_BIV_PROOF}
\end{figure*}

\begin{remark}
Theorems \ref{thm_STC_p} and \ref{thm_STC} are readily generalizable to stitched polar codes of arbitrary lengths. Specifically, the constant term $c$ for partially stitched polar codes will be smaller than that of BRS polar codes of the same length. Moreover, partially stitched codes achieve a stable constant term when there is a slight increase in code lengths by strategically increasing the size of their stitched blocks.
\end{remark}

Theorem \ref{thm_BRS} demonstrate that the constant term for BRS polar codes increases at the mother code length switching point.

\begin{theorem}\label{thm_BRS}
$c_{BRS}^{(t)}(I,[a,1-a]) \geq c(I,[a,1-a])$ for BEC with capacity $I\in (0,1)$, $a<0.1$ and $t > max\{\log_a 0.9, \frac{1}{1-\lambda} \log_2 \left(\frac{\lambda c(I,[a,1-a])}{c(I, [0.1,0.9])} + \lambda\right)\}$. Here, $h(z) = z^{t}+z-z^{t+1}$.
\end{theorem}

\begin{proof}
Analogous to the preceding proof, consider a length-$(2^t+1)$ sub-block with capacity $I_0$ in Fig. \ref{fig_BIV_PROOF}. Note that $h(I_0)= I_{0}^{t}+I_{0}-I_{0}^{t+1}$ is increasing for $I_0\in [0,1]$. If $h(I_0) \in[a, 1-a]$, the new bit-channel is un-polarized. If $I_{0}^{t}\notin [a, 1-a]$ while $I_{0}^{t+1}\in [a, 1-a]$, the first bit-channel transitions from un-polarized to polarized, resulting in one additional un-polarized bit-channel. We are going to prove that for $I_0\in [h^{-1}(a),a^{1/t}]$, there is one more un-polarized bit-channel in the sub-block, otherwise, no additional polarized bit-channels arise in the sub-block.

The conditions $I_{0}^{t}\notin [a, 1-a]$ and $I_{0}^{t+1}\in [a, 1-a]$ imply $I_{0} \in\left[a^{1/t}, a^{1/(t+1)}\right]$ and $h(I_0) \in\left[h(a^{1/t}), h(a^{1/(t+1)})\right] \subset[a, 1-a]$. Therefore, the first bit-channel transitions from un-polarized to polarized only when $I_{0} \in \left[a^{1/t}, a^{1/(t+1)}\right]$, but in this case, the new bit-channel is un-polarized. In other words, there is no $I_0$ such that the number of un-polarized bit-channels increases.

When $I_{0} \in [0.1, 0.9]\subseteq [h^{-1}(a),a^{1/t}]$, $h(I_0)\in[a, h(a^{1/t})]\subseteq [a, 1-a] \backslash [h(a^{1/t}), h(a^{1/(t+1)})]$, so there will be a new un-polarized channel.

Denote $c_{BRS}^{(t)}(I,[a,1-a]) = c_1$,  $c(I,[a,1-a]) = c_2$ and $c(I, [h^{-1}(a),a^{1/t}]) = c_3$. We have
$$
c_1(2^m+2^{m-t})^{\lambda} \geq c_2(2^m)^{\lambda} + c_3(2^{m-t})^{\lambda}.
$$
Then
$$
c_1 \geq c_2(1 + 2^{-t})^{-\lambda} + c_3 2^{-t\lambda}(1 + 2^{-t})^{-\lambda}. 
$$
By Taylor expansion, $c_2(1 + 2^{-t})^{\lambda} > c_2 + c_2\lambda 2^{-t}$, so 
$$
c_1 - c_2 > c_2\lambda 2^{-t} + c_3 2^{-t\lambda}(1 + 2^{-t})^{\lambda}. 
$$

Since $2^{t(1-\lambda)} (1 + 2^{-t})^{-\lambda} > 2^{t(1-\lambda)-\lambda} > \frac{\lambda c_2}{c_3}$,
$$
c_2\lambda  2^{-t} + c_3 2^{-t\lambda}(1 + 2^{-t})^{\lambda} > 0
$$
Therefore, $c_{BRS}^{(t)}(I,[a,b]) >c(I,[a,b])$.

\end{proof}

The results for QUP polar codes relies on the following assumption:

\begin{conjecture}\label{CON_c}
For a non-negative function $h:[0,1]\to \RR$, if 
$$
2h(z) = q( h(z^2) + h(2z-z^2) )
$$
for all $z\in [0, 1]$ and some constant $q>0$, then
$$
h(z_1) + h(z_2) \leq q (h(z_1z_2) + h(z_1+z_2-z_1z_2)),
$$
for all $z_1,z_2\in [0, 1]$.
\end{conjecture}

\begin{remark}\label{Re:scaling}
For a length-$N$ regular polar codes over a BEC with capacity $I$, the number of un-polarized bit-channels is given by $c(I) N^{\lambda}$. This quantity can also be understood as the sum of un-polarized bit-channels in two length-$N/2$ polar codes, over BECs with capacities $I^2$ and $2I-I^2$ respectively, after one polarization stage, expressed as $(c(I^2) + c(2I-I^2)) (N/2)^{\lambda}$. This relationship implies that $c(I)$ satisfies the conditions of Conjecture \ref{CON_c} with $q=2^{1-\lambda}$. In other words, the total number of un-polarized bit-channels across two length-$N/2$ polar codes, each over BEC with capacity $I$, is $2c(I) (N/2)^{\lambda}$. It is equivalent to $2^{1-\lambda}$ times the number of un-polarized bit-channels in a length-$N$ polar code. In simpler terms, one polarization stage effectively reduces the number of un-polarized bit-channels by a factor of $2^{\lambda-1}$.

Now, consider a length-$N$ regular polar code over BECs where half have capacities $I_1$ and the other half have capacity $I_2$. When $I_1\neq I_2$, we hypothesize that the polarization effect will be weaker. This is because polarization inherently amplifies differences between channels. Specifically, we believe that  $2^{1-\lambda}$ times the number of un-polarized bit-channels in this unequal capacity scenario should be greater than that in two length-$N/2$ polar codes over BECs with capacities $I_1$ and $I_2$. This conviction forms the basis of our proposed conjecture. We have verified this through simulations, which show it holds nearly true.
\end{remark}

\begin{corollary}\label{thm_qup}
Assuming the Conjecture \ref{CON_c} holds, then $c_{QUP}^{(t)}(I) \geq c(I)$ for BEC with $I\in (0,1)$ and $t\geq \max\{\log_2\lambda+1, \frac{1}{\lambda}\log_2(1+2\lambda)\}$.
\end{corollary}

\begin{proof}
This observation in Remark \ref{Re:scaling} can be generalized to length-$N$ regular polar codes partitioned into $2^t$ sub-blocks over BECs with capacities $I_i$ for $1\leq i\leq 2^t$. The number of un-polarized bit-channels multiplying $2^{(1-\lambda) t}$ is larger than the average of the number of $2^t$ length-$2^{m-t}$ polar codes over BECs with capacities $I_i$. 

Consider the $2^t+1$ sub-blocks in $(1+2^{-t})N$ QUP polar codes. Except for the first sub-block, the capacities are $I$ for the middle $2^t-1$ sub-blocks and $2I-I^2$ for the last, as shown in Fig. \ref{fig_QUP_PROOF}. Then 
\begin{align*}
& c_{QUP}^{(t)}(I)(2^m+2^{m-t})^{\lambda} \\
& \geq \left[2^{t(\lambda-1)} \left((2^{t}-1) c(I)+c\left(2I-I^2\right)\right)+c\left(I^2\right)\right]\left(2^{m-t}\right)^{\lambda}.
\end{align*}
To prove $c_{QUP}^{(t)}(I)\geq c(I)$, we need to show
\begin{align*}
& \left[2^{t(\lambda-1)} \left((2^{t}-1) c(I)+c\left(2I-I^2\right)\right)+c\left(I^2\right)\right]\left(2^{m-t}\right)^{\lambda} \\
& \geq c(I)(2^m+2^{m-t})^{\lambda}.
\end{align*}
Multiplied both side by $2^{t-\lambda m}$ and combine like terms, we have
$$
2^{t(1-\lambda)} c(I^2) +  c(2I-I^2) \geq (2^t(1+2^{-t})^{\lambda} - (2^t-1))c(I).
$$
Since 
$$
(1+2^{-t})^{\lambda} < \frac{1}{1-\lambda 2^{-t}}
$$
we only need to prove
$$
2^{t\lambda} c(I^2) +  c(2I-I^2) \geq  (1+\frac{\lambda}{1-\lambda 2^{-t}})c(I),
$$
which is true when $t$ satisfies the condition.
\end{proof}

\begin{figure*}[!t]
\centering
    \includegraphics[width=0.95\textwidth,trim = 0 200 0 200, clip]{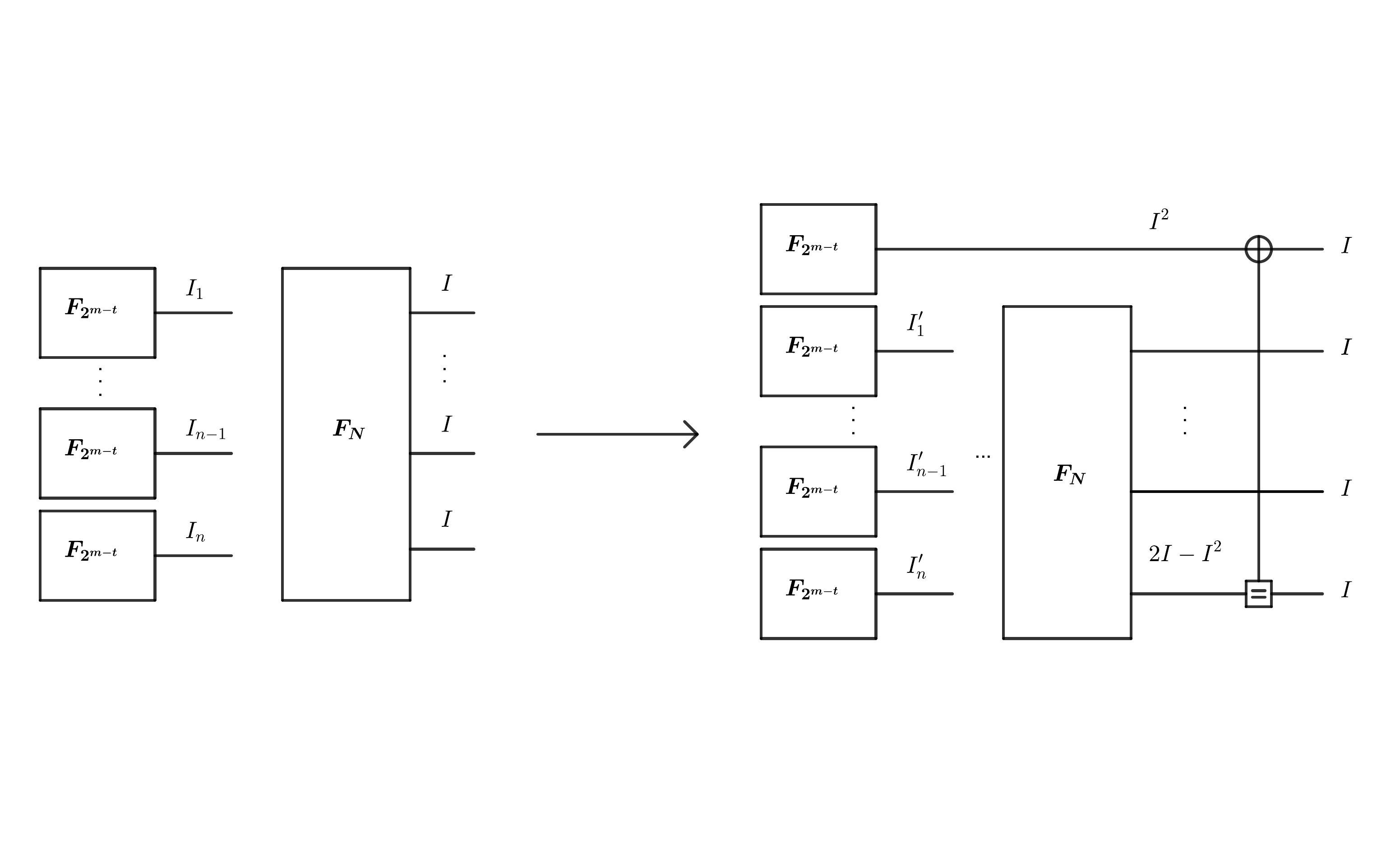}
    \caption{length-$N$ regular polar codes and length-$(1+2^{-t})N$ a QUP polar codes over BEC with capacity $I$.}
  \label{fig_QUP_PROOF}
\end{figure*}

\section{Discussion}

\subsection{Related Works}

\subsubsection{Large Kernels}

Large kernels are an effective approach for accelerating polarization. Initially proposed  in \cite{Korada2010}, polar codes with large kernels have seen significant advances in \cite{Presman2015, Lin2015, Bioglio2018, Trifonov2024}, achieving improved scaling and error exponents compared to regular polar codes \cite{Fazeli2021}. 

However, the decoding of large kernels is complicated, as existing algorithms in \cite{Trifonov2021Recursive, Trofimiuk2019, Cavatassi2019} exhibit exponential complexity growth with kernel size $l$. Also, the original simple f/g-operations in (\ref{eq:SCf}) and (\ref{eq:SCg}) cannot be reused in their implementations. Moreover, large-kernel polar codes natively support power-of-$l$ lengths only, which results in a performance drop under rate matching (similar to regular polar codes). Therefore, only expanding the length of the polarization kernel fails to solve the rate-matching challenge unless efficient decoding algorithms are developed. 

\subsubsection{Multi-Kernel Polar Codes}

The multi-kernel polar codes method, introduced in \cite{Bioglio2020}, utilizes a mixture of kernels with lengths 2, 3, and 5 to expand mother code lengths. For instance, a code length of 30 would employ one length-5 kernel, one length-3 kernel, and two length-2 kernels. However, this approach still necessitates rate-matching for code lengths not divisible by 2, 3, and 5. Furthermore, the inclusion of multiple kernel types introduces additional hardware complexity beyond the simple f/g-operations in (\ref{eq:SCf}) and (\ref{eq:SCg}).

\subsubsection{Concatenated Polar Codes}

The partially stitched polar codes are analogous to concatenated polar codes, where polar codes serve as inner codes and Reed-Solomon codes or BCH codes are used as outer codes \cite{Bakshi2010, Wang2014Concatenations}. However, unlike concatenated polar codes - which requires additional hardware modules to decode the outer codes - stitched codes can be decoded using standard f/g-operations only, without increasing the overall decoding complexity.

\subsection{Open Questions}

We identify several ongoing challenges in stitched polar codes. The primary challenge lies in the determination of optimal code structures. For small code lengths $N$, exhaustive search remains a feasible approach to identify optimal designs. For long codes, the search complexity increases quickly. 

Currently, we use a recursive approach with Algorithm \ref{alg:CON1} to build right-stitched polar codes. Additionally, Algorithm \ref{alg:CON2} employs the BRS pattern for length assignment and DE for rate assignment. Although BRS is a widely used general-purpose technique that often performs well, its optimality is not theoretically guaranteed. 

The challenge of optimal code length and rate assignment remains a significant area for exploration. This presents an open research opportunity: given specific total code length $N$ and number of information bits $K$, develop the methodology for allocating code length and information bits for stitched polar codes.

\section{Numerical Results}\label{Sec:Sim}

In this section, we present the BLER performance of the proposed stitched polar codes under the AWGN channel. These codes are constructed using Algorithm \ref{alg:CON1} and Algorithm \ref{alg:CON2}. 

As shown in Fig. \ref{fig_SC3},  \ref{fig_SC1} and \ref{fig_SC2}, we compare the fine-granularity performance of 64-partially stitched polar codes with QUP and BRS polar codes under SC decoding of different rates. The information sets are determined via GA. The results demonstrate that the stitched polar codes exhibit a much smoother curve and consistently outperform regular polar codes of all lengths, with gains up to 0.3 dB at no addtional complexity cost.

\begin{figure}[!t]
\centering
\includegraphics[width=0.5\textwidth]{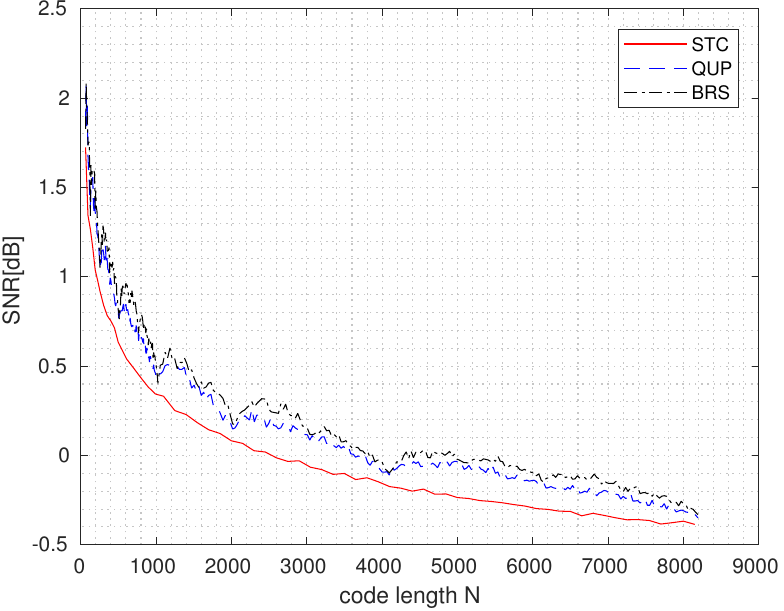}
\caption{Performance comparisons under SC decoding with $R=1/3$ and BLER $=0.01$.}
\label{fig_SC1}
\end{figure}

\begin{figure}[!t]
\centering
\includegraphics[width=0.5\textwidth]{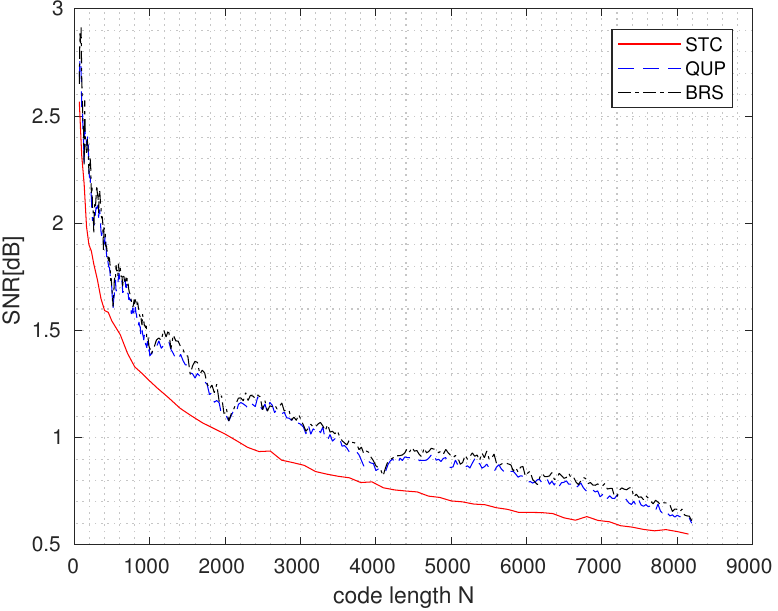}
\caption{Performance comparisons under SC decoding with $R=2/5$ and BLER $=0.01$.}
\label{fig_SC2}
\end{figure}

Fig. \ref{fig_SLC1}-\ref{fig_SCL3} illustrate the fine-granularity performance of QUP, BRS and 64-partially stitched polar codes under SCL decoding with list size 8 and CRC polynomial $D^{11} + D^{10} + D^9 + D^5 + 1$ of code length from 256 to 8192. The simulation results indicate that the stitched polar codes also exhibit universally robust and superior performance under SC-list decoding.

\begin{figure}[!t]
\centering
\includegraphics[width=0.5\textwidth]{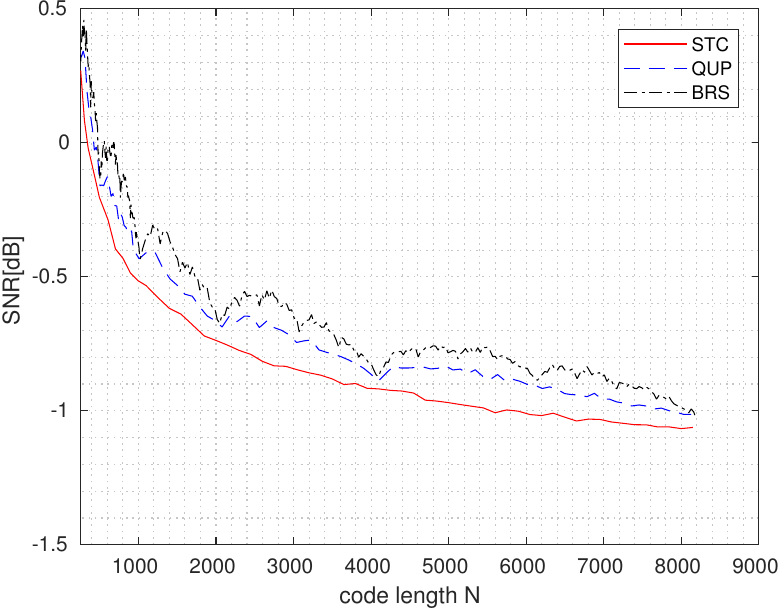}
\caption{Performance comparisons under SCL decoding with $R=1/3$ and BLER $=0.01$.}
\label{fig_SLC1}
\end{figure}

\begin{figure}[!t]
\centering
\includegraphics[width=0.5\textwidth]{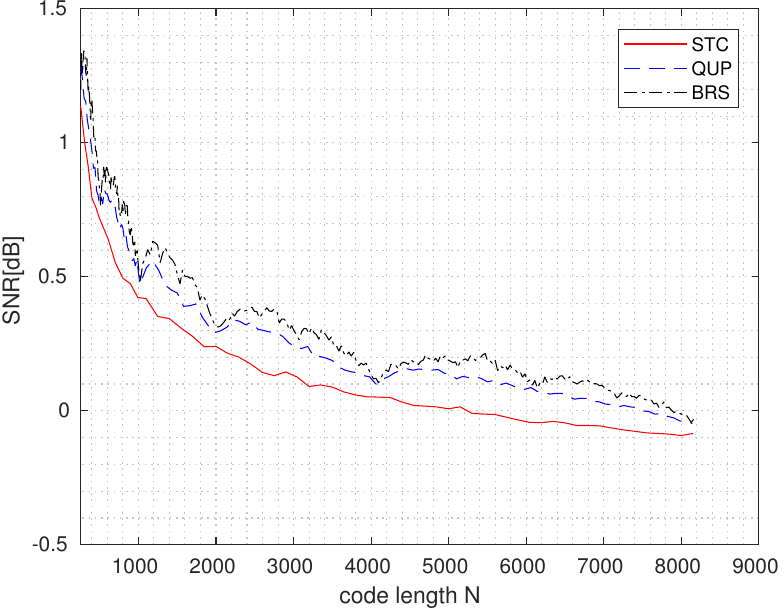}
\caption{Performance comparisons under SCL decoding with $R=2/5$ and BLER $=0.01$.}
\label{fig_SCL2}
\end{figure}

\begin{figure}[!t]
\centering
\includegraphics[width=0.5\textwidth]{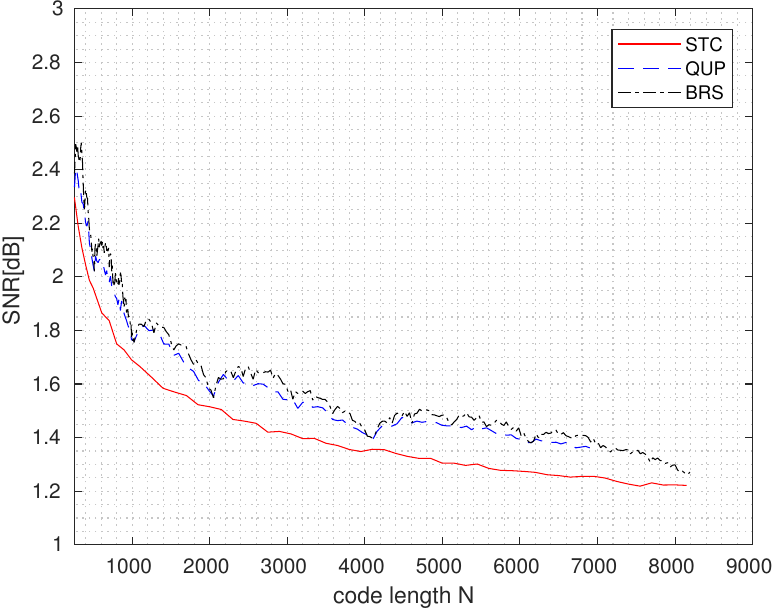}
\caption{Performance comparisons under SCL decoding with $R=1/2$ and BLER $=0.01$.}
\label{fig_SCL3}
\end{figure}

\section{Conclusion}

In conclusion, length-$N$ stitched polar codes accelerate polarization while preserving encoding and decoding complexity at $O(N\log N)$. This design addresses the performance loss due to incomplete polarization for non-power-of-two code lengths. The gain of stitched polar codes comes from rearranging the basic transformations to enhance the mutual information towards the information bits, in particular towards the least reliable one. The numerical results demonstrate substantial and consistent performance improvements across diverse scenarios.

\appendices
\section*{Appendix A}

The channel combining and splitting are not applicable for all coupling sequence. For instance, the code illustrated in Fig. \ref{fig_ce} cannot be decoded by f/g-operations. The reason is that after polarization $(1,3)$ and $(2,3)$, both the first and second channels carry the information from the third channel, so they are not independent. A strict inspection procedure, outlined in Algorithm \ref{alg:CS}, determines which coupling sequences are considered validated based on whether they pass this inspection. 

\begin{figure}[!t]
\centering
\includegraphics[width=0.4\textwidth]{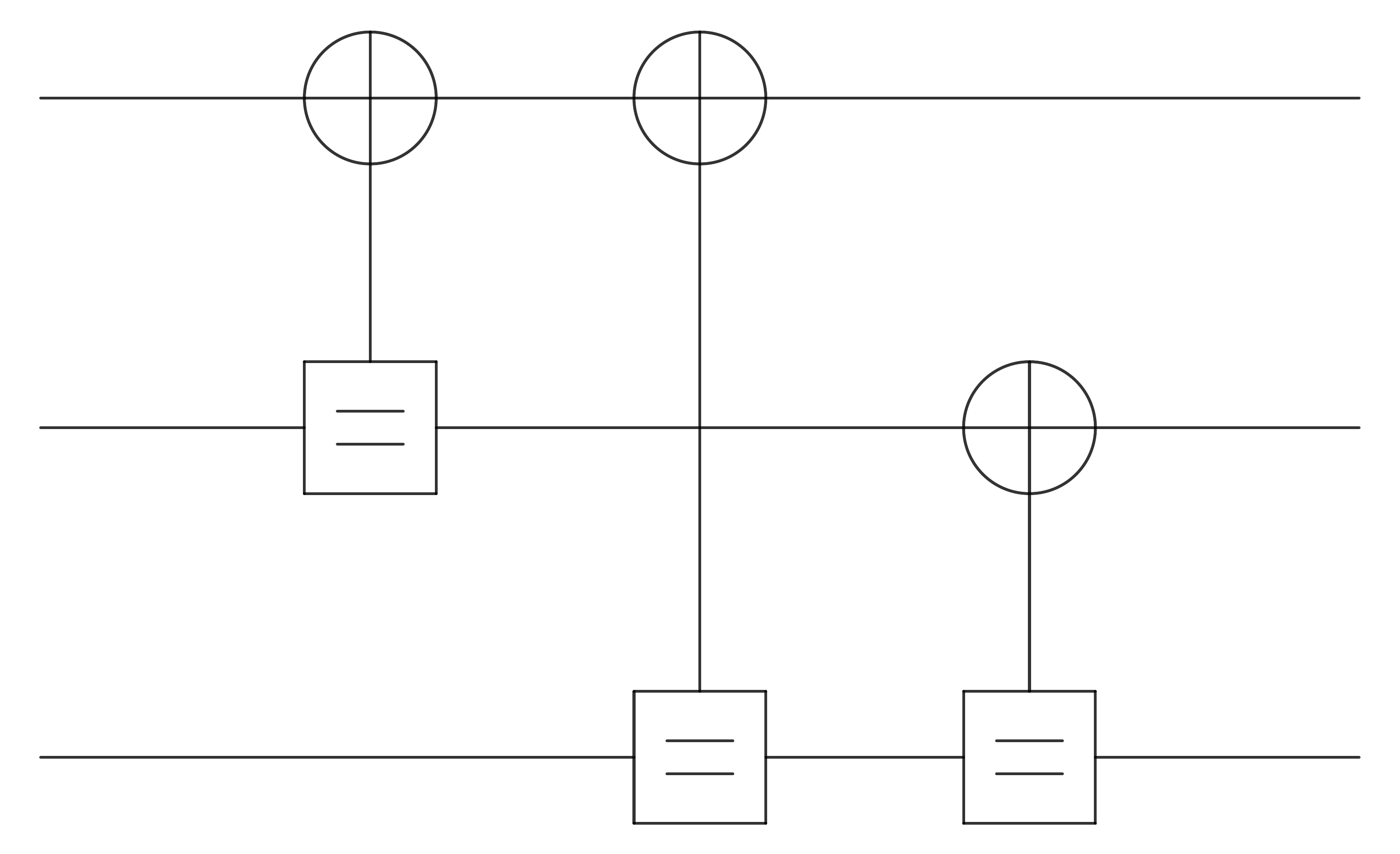}
\caption{A counter-example $(1,2), (1,3), (2,3)$ that cannot be decoded by f/g-operations.}
\label{fig_ce}
\end{figure}

\begin{figure}[!t]
\begin{algorithm}[H]
\caption{validate\_stitched\_polar\_codes($C$)}
\begin{algorithmic}[1]\label{alg:CS}

\renewcommand{\algorithmicrequire}{\textbf{Input:}}
\renewcommand{\algorithmicensure}{\textbf{Output:}}
\REQUIRE the length-$N$ coupling sequence $C = (a_1, b_1), \dots, (a_n, b_n)$.
\ENSURE \textbf{True} or \textbf{False}
\STATE Initialize $A_j^{(n)} \gets \{j\}$ and $\mathcal{V}_j^{(n)}\gets \varnothing$ for $1\leq j\leq N$.
\FOR{$i = n$ to $1$}
 \STATE $A^{(i-1)}_{a_i} \gets A^{(i)}_{a_j} \cup A^{(i)}_{b_i}$;
 \STATE $A^{(i-1)}_{b_i} \gets A^{(i-1)}_{a_i}$;
 \STATE $\mathcal{V}^{(i-1)}_{a_i} \gets \mathcal{V}^{(i)}_{a_i} \cup \mathcal{V}^{(i)}_{b_i}$; 
 \STATE $\mathcal{V}^{(i-1)}_{b_i} \gets \mathcal{V}^{(i)}_{a_i} \cup \mathcal{V}^{(i)}_{b_i} \cup \{u^{(i)}_{a_i}\}$;
 \IF{$\bm{x}_{A^{(i)}_{a_i}}$ and $\mathcal{V}_{A^{(i)}_{a_i}}$ are independent with $u^{(i)}_{b_i}$ conditioned on $u^{(i)}_{a_i}$, or $\bm{x}_{A^{(i)}_{b_i}}$ and $\mathcal{V}_{A^{(i)}_{b_i}}$ are independent with $u^{(i)}_{a_i}$, $\bm{x}_{A^{(i)}_{a_i}}$ and $\mathcal{V}_{A^{(i)}_{a_i}}$ conditioned on $u^{(i)}_{b_i}$}
   \STATE \textbf{Return False}
 \ENDIF
\ENDFOR
\STATE \textbf{Return True}
\end{algorithmic}
\end{algorithm}
\end{figure}

\begin{proof}[Proof of Proposition \ref{pro_CS}]
For simplicity, we denote all the channel transposition probabilities as $P$. The proposition is directly derived from  Eq. (\ref{eq5}) and (\ref{eq6}), predicated on the independence condition in Algorithm \ref{alg:CS}. 

\begin{figure*}
\begin{align}
\label{eq5}
P \left( \bm{y}_{A^{(i-1)}_{a_i}}, \mathcal{V}^{(i-1)}_{a_i} \mid u^{(i-1)}_{a_i} \right) &  = \sum_{u^{(i-1)}_{b_i}} P(u^{(i-1)}_{b_i} \mid u^{(i-1)}_{a_i})  P \left( \bm{y}_{A^{(i-1)}_{a_i}}, \mathcal{V}^{(i-1)}_{a_i} \mid u^{(i-1)}_{a_i}, u^{(i-1)}_{b_i} \right) \notag \\ 
&  = \sum_{u^{(i-1)}_{b_i}}  \frac{1}{2}  P \left( \bm{y}_{A^{(i)}_{a_i}\cup A^{(i)}_{b_i}}, \mathcal{V}^{(i)}_{a_i}\cup \mathcal{V}^{(i)}_{b_i} \mid u^{(i)}_{a_i}, u^{(i)}_{b_i} \right) \notag \\ 
&  = \sum_{u^{(i-1)}_{b_i}}  \frac{1}{2}   P \left( \bm{y}_{A^{(i)}_{a_i}}, \mathcal{V}^{(i)}_{a_i} \mid u^{(i)}_{a_i}, u^{(i)}_{b_i}\right) P \left( \bm{y}_{A^{(i)}_{b_i}}, \mathcal{V}^{(i)}_{b_i} \mid \bm{y}_{A^{(i)}_{a_i}}, \mathcal{V}^{(i)}_{a_i}, u^{(i)}_{a_i}, u^{(i)}_{b_i},  \right) \notag \\ 
&  = \sum_{u^{(i-1)}_{b_i}} \frac{1}{2}  P \left( \bm{y}_{A^{(i)}_{a_i}}, \mathcal{V}^{(i)}_{a_i} \mid u^{(i)}_{a_i}\right) P \left( \bm{y}_{A^{(i)}_{b_i}}, \mathcal{V}^{(i)}_{b_i} \mid u^{(i)}_{b_i} \right). 
\end{align}
\begin{align}
\label{eq6}
P \left( \bm{y}_{A^{(i-1)}_{b_i}}, \mathcal{V}^{(i-1)}_{b_i} \mid u^{(i-1)}_{b_i} \right) &  =  P(u^{(i-1)}_{a_i} \mid u^{(i-1)}_{b_i})  P \left( \bm{y}_{A^{(i-1)}_{b_i}}, \mathcal{V}^{(i)}_{a_i}\cup  \mathcal{V}^{(i)}_{b_i} \mid u^{(i-1)}_{a_i}, u^{(i-1)}_{b_i} \right) \notag \\  \
&  =   \frac{1}{2}  P \left( \bm{y}_{A^{(i)}_{a_i}\cup A^{(i)}_{b_i}}, \mathcal{V}^{(i)}_{a_i}\cup \mathcal{V}^{(i)}_{b_i} \mid u^{(i)}_{a_i}, u^{(i)}_{b_i} \right) \notag \\ 
&  = \frac{1}{2} P \left( \bm{y}_{A^{(i)}_{a_i}}, \mathcal{V}^{(i)}_{a_i} \mid u^{(i)}_{a_i} \right) P \left( \bm{y}_{A^{(i)}_{b_i}}, \mathcal{V}^{(i)}_{b_i} \mid u^{(i)}_{b_i} \right). 
\end{align}
\end{figure*}
\end{proof}

Consider the code in Example \ref{Ex:Pol_Gra}. Note that the channels are independent at every step of its polarization process. For instance, examine the process $(3,5)$. Here $y_1$ is determined by $x_1 = u_1\oplus u_2\oplus u_3\oplus u_4\oplus u_5$. This determination is independent of $u_5$ given $u_3\oplus u_4\oplus u_5$. Similarly, $y_3$ (which itself is determined by $u_3\oplus u_4\oplus u_5$) and $u_1\oplus u_2$ are independent of $u_5$ given $u_3\oplus u_4\oplus u_5$. Furthermore,  $y_2, y_5$ and $u_2$ are independent of $y_1, y_3, u_1\oplus u_2$ and $u_3\oplus u_4\oplus u_5$ given $u_5$.

\section*{Appendix B}

The table \ref{tab_N4-8} presents the generator matrices for optimal stitched polar codes for code length $4\leq N\leq 8$ and dimension $1\leq K\leq N-1$. In the table, the term `regular' denotes the optimal code is the mother-length regular polar code. Meanwhile, `shortening' and `puncturing' denote the optimal code is the quasi random shortened and punctured regular polar code, respectively. For other irregular cases, the corresponding codes, including their generator matrices $\bm{G}_{N,K}$ and information sets $\MI$, are listed below.

\begin{table}[htbp] 
\begin{center}
\setlength{\tabcolsep}{1.5mm}{
\begin{tabular}{c|c|c|c|c|c}
\hline
& $N=8$ & $N=7$	& $N=6$	& $N=5$ & $N=4$
 \\
\hline
$K=7$	& shortening &  & & & \\
\hline
$K=6$	& $C_{8,6}$ & shortening & & & \\
\hline
$K=5$	& $C_{8,5}$ & $C_{7,5}$ & shortening  & & \\
\hline
$K=4$	& regular &  puncturing & $C_{6,2}$  & shortening & \\
\hline
$K=3$	& $C_{8,3}$ & shortening & $C_{6,3}$  & $C_{5,3}$ & regular\\
\hline
$K=2$	& $C_{8,2}$ & $C_{7,2}$ & $C_{6,4}$  & $C_{5,2}$ & $C_{4,2}$\\
\hline
$K=1$	& regular & puncturing & puncturing  & puncturing & regular\\
\hline
\end{tabular}}

\caption{The table of optimal stitched polar codes.}
\label{tab_N4-8}
\end{center}
\end{table}

$C_{4,2} = (2,3), (1,3), (1,4), \MI = \{3,4\}$
$$
\bm{G}_{4,2}=\begin{pmatrix}
1 & 0 & 0 & 0 \\
0 & 1 & 0 & 0 \\
1 & 1 & 1 & 0 \\
1 & 0 & 0 & 1
\end{pmatrix}
$$

$C_{5,2} = (3,4), (1,2), (3,5), (1,3), (2,5), \MI = \{4,5\}$
$$
\bm{G}_{5,2}=\begin{pmatrix}
1 & 0 & 0 & 0 & 0 \\
1 & 1 & 0 & 0 & 0 \\
1 & 0 & 1 & 0 & 0 \\
1 & 0 & 1 & 1 & 0 \\
1 & 1 & 1 & 0 & 1
\end{pmatrix}
$$

$C_{5,3} = (2,3), (1,2), (4,5), (1,4), (2,5), \MI = \{3,4,5\}$
$$
\bm{G}_{5,3}=\begin{pmatrix}
1 & 0 & 0 & 0 & 0 \\
1 & 1 & 0 & 0 & 0 \\
1 & 1 & 1 & 0 & 0 \\
1 & 0 & 0 & 1 & 0 \\
1 & 1 & 0 & 1 & 1
\end{pmatrix}
$$

$C_{6,2} = (2,3), (4,5), (1,2), (3,5), (4,6), (1,4), (2,6), \MI = \{5,6\}$
$$
\bm{G}_{6,2}=\begin{pmatrix}
1 & 0 & 0 & 0 & 0 & 0 \\
1 & 1 & 0 & 0 & 0 & 0 \\
1 & 1 & 1 & 0 & 0 & 0 \\
1 & 0 & 0 & 1 & 0 & 0 \\
1 & 0 & 1 & 1 & 1 & 0 \\
1 & 1 & 0 & 1 & 0 & 1
\end{pmatrix}
$$

$C_{6,3}$ is equivalent to $C_{6,2}$ with $\MI = \{3,5,6\}$

$C_{6,4}$ is equivalent to $C_{6,2}$ with $\MI = \{3,4,5,6\}$

$C_{7,2} = (1,2), (3,4), (5,6), (1,3), (4,6), (5,7), (1,5), (2,4),$  \\
$(3,7), \MI = \{6,7\}$
$$
\bm{G}_{7,2}=\begin{pmatrix}
1 & 0 & 0 & 0 & 0 & 0 & 0 \\
1 & 1 & 0 & 0 & 0 & 0 & 0 \\
1 & 0 & 1 & 0 & 0 & 0 & 0 \\
1 & 1 & 1 & 1 & 0 & 0 & 0 \\
1 & 0 & 0 & 0 & 1 & 0 & 0 \\
1 & 1 & 0 & 1 & 1 & 1 & 0 \\
1 & 0 & 1 & 0 & 1 & 0 & 1
\end{pmatrix}
$$

$C_{7,5} = (2,3), (4,5), (6,7), (1,2), (3,5), (4,6), (1,4), (2,6),$  \\
$(3,7), \MI = \{3,4,5,6,7\}$
$$
\bm{G}_{7,5}=\begin{pmatrix}
1 & 0 & 0 & 0 & 0 & 0 & 0 \\
1 & 1 & 0 & 0 & 0 & 0 & 0 \\
1 & 1 & 1 & 0 & 0 & 0 & 0 \\
1 & 0 & 0 & 1 & 0 & 0 & 0 \\
1 & 0 & 1 & 1 & 1 & 0 & 0 \\
1 & 1 & 0 & 1 & 0 & 1 & 0 \\
1 & 1 & 1 & 1 & 0 & 1 & 1
\end{pmatrix}
$$

$C_{8,2} = (2,3), (4,5), (6,7), (2,4), (3,5), (6,8), (1,2), (4,8),$ \\
$(1,6), (3,7), \MI = \{7,8\}$
$$
\bm{G}_{8,2}=\begin{pmatrix}
1 & 0 & 0 & 0 & 0 & 0 & 0 & 0 \\
0 & 1 & 0 & 0 & 0 & 0 & 0 & 0 \\
0 & 1 & 1 & 0 & 0 & 0 & 0 & 0 \\
0 & 1 & 0 & 1 & 0 & 0 & 0 & 0 \\
0 & 1 & 0 & 1 & 1 & 0 & 0 & 0 \\
1 & 1 & 0 & 0 & 0 & 1 & 0 & 0 \\
1 & 1 & 1 & 0 & 1 & 1 & 1 & 0 \\
1 & 1 & 0 & 1 & 0 & 1 & 0 & 1 \\
\end{pmatrix}
$$

$C_{8,3} = (4,5), (1,2), (3,4), (6,7), (1,3), (2,4), (6,8), (1,6),$ \\
$(2,7), (3,8), \MI = \{5,7,8\}$
$$
\bm{G}_{8,3}=\begin{pmatrix}
1 & 0 & 0 & 0 & 0 & 0 & 0 & 0 \\
1 & 1 & 0 & 0 & 0 & 0 & 0 & 0 \\
1 & 0 & 1 & 0 & 0 & 0 & 0 & 0 \\
1 & 1 & 1 & 1 & 0 & 0 & 0 & 0 \\
1 & 1 & 1 & 1 & 1 & 0 & 0 & 0 \\
1 & 0 & 0 & 0 & 0 & 1 & 0 & 0 \\
1 & 1 & 0 & 0 & 0 & 1 & 1 & 0 \\
1 & 0 & 1 & 0 & 0 & 1 & 0 & 1 \\
\end{pmatrix}
$$

$C_{8,5} = (4,5), (2,3), (4,6), (7,8), (1,2), (3,6), (4,7), (1,4),$ \\
$(2,7), (3,8), \MI = \{3,5,6,7,8\}$
$$
\bm{G}_{8,5}=\begin{pmatrix}
1 & 0 & 0 & 0 & 0 & 0 & 0 & 0 \\
1 & 1 & 0 & 0 & 0 & 0 & 0 & 0 \\
1 & 1 & 1 & 0 & 0 & 0 & 0 & 0 \\
1 & 0 & 0 & 1 & 0 & 0 & 0 & 0 \\
1 & 0 & 0 & 1 & 1 & 0 & 0 & 0 \\
1 & 1 & 1 & 1 & 0 & 1 & 0 & 0 \\
1 & 1 & 0 & 1 & 0 & 0 & 1 & 0 \\
1 & 1 & 1 & 1 & 0 & 0 & 1 & 1 \\
\end{pmatrix}
$$

$C_{8,5} = (2,3), (5,6), (2,4), (5,7), (1,2), (4,7), (5,8), (1,5),$ \\
$(2,8), (3,6), \MI = \{3,4,5,6,7,8\}$
$$
\bm{G}_{8,6}=\begin{pmatrix}
1 & 0 & 0 & 0 & 0 & 0 & 0 & 0 \\
1 & 1 & 0 & 0 & 0 & 0 & 0 & 0 \\
1 & 1 & 1 & 0 & 0 & 0 & 0 & 0 \\
1 & 1 & 0 & 1 & 0 & 0 & 0 & 0 \\
1 & 0 & 0 & 0 & 1 & 0 & 0 & 0 \\
1 & 0 & 1 & 0 & 1 & 1 & 0 & 0 \\
1 & 0 & 0 & 1 & 1 & 0 & 1 & 0 \\
1 & 1 & 0 & 0 & 1 & 0 & 0 & 1 \\
\end{pmatrix}
$$

\section*{Acknowledgement}

The authors would like to express their sincere gratitude to Prof. Erdal Ar{\i}kan for his valuable suggestions and constructive feedback during the preparation of this manuscript. His insights were instrumental in improving the quality and clarity of this work.

\end{document}